\documentclass{article}
\usepackage{fullpage}
\usepackage{amsthm,amssymb,amsmath}  
\usepackage{authblk}
\usepackage{tikz} 

\usepackage{mathtools, eqparbox}%

\usepackage{wrapfig}
\usepackage[short,nocomma]{optidef}
\usepackage{cases}
\usepackage{todonotes}
\usepackage{subcaption}
\usepackage{nicefrac}
\usepackage{amsthm,amssymb,amsmath}  
\usepackage{xspace,enumerate}
\usepackage[utf8]{inputenc}
\usepackage{thmtools}
\usepackage{thm-restate}
\usepackage{todonotes}
\usepackage{hyperref}
\usepackage{verbatim}
\usepackage{multirow}
\usepackage{url}
\usepackage{mathtools}
\usepackage{amsthm}
\usepackage[utf8]{inputenc}
\usepackage{color}
\usepackage{caption}
\usepackage{scalerel}
\usepackage{subcaption}
\usepackage{todonotes}
\usepackage{algorithm}
\usepackage[noend]{algpseudocode}
\usepackage{algorithmicx}
\usepackage{etoolbox}
\makeatother
\usepackage[normalem]{ulem}
\usetikzlibrary{calc,snakes,shapes,arrows.meta}
\usepackage{booktabs}
\usepackage{bold-extra}

\newcommand{\cO}{\mathcal{O}}
\newcommand{\ctO}{\mathcal{\tilde{O}}}

\def\dd{\mathinner{.\,.}}
\newcommand{\LCP}{\textsf{LCP}}
\newcommand{\occ}{\textsf{occ}}

\newcommand{\defDSproblem}[3]{
   \vspace{2mm}
 \noindent\fbox{
   \begin{minipage}{0.96\textwidth}
   \textsc{#1}\\
   {\bf{Preprocess:}} #2  \\
   {\bf{Query:}} #3
   \end{minipage}
   }
   \vspace{2mm}
}

  \theoremstyle{plain}
  \newtheorem{theorem}{Theorem}
  \newtheorem{lemma}{Lemma}

  \newtheorem{fact}{Fact}
  
  \theoremstyle{definition}
  \newtheorem{definition}{Definition}
  
  \newtheorem{example}{Example}
  
  {\bfseries}{\itshape}

\title{Text indexing for long patterns using locally consistent anchors}
\author[1]{Lorraine A. K. Ayad}
\author[2]{Grigorios Loukides}
\author[3,4]{Solon P.\ Pissis}

\affil[1]{Brunel University London, London, UK}
\affil[2]{King's College London, London, UK}
\affil[3]{CWI, Amsterdam, The Netherlands}
\affil[4]{Vrije Universiteit, Amsterdam, The Netherlands}

\date{}

\begin{document}

\maketitle

\begin{abstract}
In many real-world database systems, a large fraction of the data is represented by strings: sequences of letters over some alphabet. This is because strings can easily encode data arising from different sources. It is often crucial to represent such string datasets in a compact form but also to \emph{simultaneously} enable fast pattern matching queries. This is the classic text indexing problem. The four absolute measures anyone should pay attention to when designing or implementing a text index are: \textbf{(i)} index space; \textbf{(ii)} query time; \textbf{(iii)} construction space; and \textbf{(iv)} construction time. Unfortunately, however, most (if not all) widely-used indexes (e.g., suffix tree, suffix array, or their compressed counterparts) are not optimized for all four measures simultaneously, as it is difficult to have the best of all four worlds. Here, we take an important step in this direction by showing that text indexing with sampling based on locally consistent anchors (lc-anchors) offers remarkably good performance in all four measures, when we have at hand a lower bound  $\ell$ on the length of the queried patterns --- which is arguably a quite reasonable assumption in practical applications. Specifically, we improve on a recently proposed index that is based on \emph{bidirectional string anchors} (bd-anchors), a new type of lc-anchors, by: \textbf{(i)} introducing a randomized counterpart of bd-anchors which outperforms bd-anchors; \textbf{(ii)} designing an average-case linear-time algorithm to compute (the randomized) bd-anchors; and \textbf{(iii)} developing a semi-external-memory implementation and an internal-memory implementation to construct the index in \emph{small space} using near-optimal work. 
Our index offers average-case guarantees. In our experiments using real benchmark datasets, we show that it compares favorably based on the four measures to all classic indexes: (compressed) suffix tree; (compressed) suffix array; and the FM-index.
Notably, we also present a counterpart of our index with worst-case guarantees based on the lc-anchors notion of \emph{partitioning sets}. To the best of our knowledge, this is the first index achieving the best of all worlds in the regime where we have at hand a lower bound $\ell$ on
the length of the queried patterns.
\end{abstract}

\section{Introduction}\label{sec:intro}

In many real-world database systems, including bioinformatics systems~\cite{DBLP:journals/nar/Rodriguez-TomeSCF96}, Enterprise Resource Planning (ERP) systems~\cite{DBLP:conf/edbt/0002RF14}, or Business Intelligence (BI) systems~\cite{DBLP:conf/sigmod/VogelsgesangHFK18}, a large fraction of the data is represented by strings: sequences of letters over some alphabet. This is because strings can easily encode data arising from different sources such as: nucleic acid sequences read by sequencing machines (e.g., short or long DNA reads); natural language text generated by humans (e.g., description or comment fields); or identifiers generated by machines (e.g., URLs, email addresses, IP addresses). Given the ever increasing size of such datasets, it is crucial to represent them compactly~\cite{DBLP:journals/pvldb/Boncz0L20} but also to \emph{simultaneously} enable fast pattern matching queries. This is the classic text indexing problem~\cite{DBLP:books/cu/Gusfield1997,DBLP:books/daglib/0020103,DBLP:books/daglib/0038982}. More formally, \emph{text indexing}  asks to preprocess a string $S$ of length $n$ over an alphabet $\Sigma$ of size $\sigma$, known as the \emph{text}, into a compact data structure that supports efficient pattern matching queries; i.e., \emph{decide} if a \emph{pattern} $P$ occurs or not in $S$ or \emph{report} the set of all positions in $S$ where an occurrence of $P$ starts. 

\subsection{Motivation and related work}

A considerable amount of algorithmic research has been devoted to text indexes over the past decades~\cite{DBLP:conf/focs/Weiner73,DBLP:journals/siamcomp/ManberM93,DBLP:conf/focs/Farach97,DBLP:journals/jacm/KarkkainenSB06,DBLP:journals/jacm/FerraginaM05,DBLP:journals/siamcomp/GrossiV05,DBLP:journals/siamcomp/HonSS09,DBLP:conf/stoc/Belazzougui14,DBLP:journals/algorithmica/0001KL15,DBLP:conf/soda/MunroNN17,DBLP:conf/stoc/KempaK19,DBLP:journals/talg/BelazzouguiCKM20,DBLP:journals/jacm/GagieNP20,DBLP:conf/esa/LoukidesP21}. This is mainly due to the fact that myriad string processing tasks (see~\cite{DBLP:books/cu/Gusfield1997,DBLP:journals/cacm/ApostolicoCFGM16} for comprehensive reviews) require fast access to the substrings of $S$. These tasks rely on such text indexes, which typically arrange the suffixes of $S$ lexicographically in an ordered tree or in an ordered array. The former is known as the \emph{suffix tree}~\cite{DBLP:conf/focs/Weiner73} and the latter is known as the \emph{suffix array}~\cite{DBLP:journals/siamcomp/ManberM93}.
These are classic data structures, which occupy $\Theta(n)$ words of \textbf{space} (or, equivalently, $\Theta(n \log n)$ bits) and can count the number of occurrences of $P$ in $S$ in $\ctO(|P|)$ time\footnote{The $\ctO(\cdot)$ notation suppresses polylogarithmic factors.}. The time for reporting is $\cO(1)$ per occurrence for both structures. Thus, if a pattern $P$ occurs $\occ$ times in $S$, the total \textbf{query time} for reporting is $\ctO(|P|+\occ)$. 

From early days, and in contrast to the traditional data
structure literature, where the focus is on space-query
time trade-offs, the main focus in text indexing has
been on the \textbf{construction time}. That was until the
breakthrough result of Farach~\cite{DBLP:conf/focs/Farach97}, who showed
that suffix trees (and thus suffix arrays, indirectly) can be constructed in $\cO(n)$ time when $\Sigma$ is an integer alphabet of size $\sigma=n^{\cO(1)}$. After Farach's result, more and more attention had been given to reducing the space of the index via compression techniques.
This is due to the fact that, although the space is linear in the number of words, there is an $\cO(\log_\sigma n)$ factor blowup when we consider the actual text size, which is $n\lceil \log \sigma \rceil$ bits. This factor is not negligible when $\sigma$ is considerably smaller than $n$. For instance, the space occupied by the suffix tree of the whole human genome, even with a very efficient implementation~\cite{DBLP:journals/spe/Kurtz99} is about $40$GB, whereas the genome occupies less than $1$GB. To address the above issue, Grossi and Vitter~\cite{DBLP:journals/siamcomp/GrossiV05} and Ferragina and Manzini~\cite{DBLP:journals/jacm/FerraginaM05}, and later Sadakane~\cite{DBLP:journals/mst/Sadakane07}, introduced, respectively, the \emph{compressed suffix array} (CSA), the \emph{FM-index}, and the \emph{compressed suffix tree} (CST).
These data structures occupy $\cO(n \log \sigma)$ bits, instead of $\cO(n \log n)$ bits, at the expense of a factor of $\cO(\log^{\epsilon} n)$ penalty in the query time, where $\epsilon > 0$ is an arbitrary predefined constant. These indexes have the dominant role in some of the most widely used bioinformatics tools~\cite{DBLP:journals/bioinformatics/LiD09,DBLP:journals/bioinformatics/LiYLLYKW09,Bowtie}; while mature and highly engineered implementations of these indexes are available through the \textsf{sdsl-lite} library of Gog et al.~\cite{sdsl}.

Nowadays, as the data volume grows rapidly, \textbf{construction space} is as well becoming crucial for several string processing tasks~\cite{DBLP:journals/talg/BelazzouguiCKM20,DBLP:conf/soda/KempaK23}. Suffix arrays can be constructed in optimal time using $\cO(1)$ words of extra space with the algorithm of Franceschini and  Muthukrishnan~\cite{DBLP:conf/icalp/FranceschiniM07} for general alphabets or with the algorithms of Goto~\cite{DBLP:conf/stringology/Goto19} or Li et al.~\cite{DBLP:journals/iandc/LiLH22} for integer alphabets. \emph{Extra} refers to the required space except for the space of $S$ and the output. Grossi and Vitter's~\cite{DBLP:journals/siamcomp/GrossiV05} original algorithm for constructing the CSA takes $\cO(n \log \sigma)$ time using $\cO(n \log n)$ bits of construction space. Hon et al.~\cite{DBLP:journals/siamcomp/HonSS09} showed an $\cO(n \log \log \sigma)$-time construction reducing the construction space to $\cO(n \log \sigma)$ bits. More recently, Belazzougui~\cite{DBLP:conf/stoc/Belazzougui14} and, independently, Munro et al.~\cite{DBLP:conf/soda/MunroNN17}, improved the time complexity of the CSA/CST construction to $\cO(n)$ using $\cO(n \log \sigma)$ bits of construction space. Recently, Kempa and Kociumaka~\cite{DBLP:conf/soda/KempaK23} have presented a new data structure that occupies $\cO(n \log \sigma)$ bits, can be constructed in $\cO(n \log \sigma/\sqrt{\log n})$ time using $\cO(n \log \sigma)$ bits of construction space, and has the same query time as the CSA and the FM-index. 

This completes the four absolute measures anyone should pay attention to when designing or implementing a text index: space (index size); query time; construction time; and construction space. 

Unfortunately, however, most (if not all) widely-used indexes are not optimized for all four measures \emph{simultaneously}, as it is
difficult to have the best of all four worlds; e.g.:
\begin{description}
\item The \emph{suffix array}~\cite{DBLP:journals/siamcomp/ManberM93} supports very fast pattern matching queries; it can be constructed very fast using very little construction space; but then it occupies $\Theta(n\log n)$ bits of space and that \emph{in any case}.
\item The \emph{CSA}~\cite{DBLP:journals/fttcs/Vitter06}, the \emph{CST}~\cite{DBLP:journals/mst/Sadakane07}, or the \emph{FM-index}~\cite{DBLP:journals/jacm/FerraginaM05} occupy $\cO(n\log \sigma)$ bits of space; they 
can be constructed very fast using $\cO(n\log \sigma)$ bits of construction space; but then they answer pattern matching queries much less efficiently than the suffix array.
\end{description}

\subsection{Our contributions}\label{sec:intro:contributions}

The purpose of this paper is to show that text indexing with \emph{locally consistent anchors} (lc-anchors, in short) offers remarkably good performance in all four measures, when we have at hand a lower bound $\ell$ on the length of the queried patterns. 
This is arguably a quite reasonable assumption in  practical applications.  For instance, in bioinformatics~\cite{Winnowmap2,Wenger2019,Logsdon2020}, the length of sequencing reads (patterns) ranges from a few hundreds to 30 thousand~\cite{Logsdon2020}. 
Even when at most $k$ errors must be accommodated for matching, at least one out of $k+1$ fragments must be matched exactly.
In natural language processing, the queried patterns can also be long~\cite{TOIS17}. Examples of such patterns are queries in question answering systems~\cite{longqueries0}, \emph{description queries} in TREC datasets~\cite{longqueries1,longqueries2}, and  representative phrases in documents~\cite{representativephrases1}. Similarly, a query pattern can be long when it  encodes an entire document (e.g., a webpage in the context of deduplication~\cite{webpageduplicate}), or  machine-generated messages~\cite{machinegenerated}.  

Informally, given a string $S$ and an integer $\ell>0$, the goal is to sample some of the positions of $S$ (the \textbf{lc-anchors}), so as to simultaneously satisfy the following: 
\begin{itemize}
    \item \emph{Property 1 (approximately uniform sampling):} Every fragment of length at least $\ell$ of $S$ has a representative position sampled. This ensures that any pattern $P$, $|P|\geq \ell$, will not be missed during search.
    \item \emph{Property 2 (local consistency):} Exact matches between fragments of length at least $\ell$ of $S$ are preserved unconditionally by having the same (relative) representative positions sampled. This ensures that similarity between similar strings of length at least $\ell$ will be preserved during search.
\end{itemize}

Loukides and Pissis~\cite{DBLP:conf/esa/LoukidesP21} have recently introduced \emph{bidirectional string anchors} (bd-anchors, in short), a new type of lc-anchors. The set $\mathcal{A}_{\ell}(S)$ of bd-anchors of order $\ell$ of a string $S$ of length $n$ is the set of starting positions of the leftmost lexicographically smallest rotation of every length-$\ell$ fragment of $S$ (see Section~\ref{sec:prel} for a formal definition). The authors have shown that bd-anchors are $\cO(n/\ell)$ in expectation~\cite{DBLP:conf/esa/LoukidesP21,TKDE2023} and that can be constructed in $\cO(n)$ worst-case time~\cite{DBLP:conf/esa/LoukidesP21}. Loukides and Pissis have also proposed a text index, which is based on $\mathcal{A}_{\ell}(S)$ and occupies linear extra space in the size $|\mathcal{A}_{\ell}(S)|$ of the sample; the index can be constructed in $\ctO(n)$ time and can report all $\occ$ occurrences of any pattern $P$ of length $|P|\geq \ell$ in $S$ in $\ctO(|P|+\occ)$ time (see Section~\ref{sec:prel} for a formal theorem). The authors of~\cite{DBLP:conf/esa/LoukidesP21} have also implemented their index and presented some very promising experimental results (see also~\cite{TKDE2023}). 

We identify here two important aspects for improving the construction of the bd-anchors index: \textbf{(i)} for computing the set $\mathcal{A}_{\ell}(S)$, which is required to construct the index, Loukides and Pissis implement a simple $\Theta(n\ell)$-time algorithm, because their $\cO(n)$-time worst-case algorithm seems too complicated to implement and unlikely to be efficient in practice~\cite{DBLP:conf/esa/LoukidesP21}; \textbf{(ii)} their index construction implementation uses $\Theta(n)$ construction space \emph{in any case}, which is much larger than the expected size $|\mathcal{A}_{\ell}(S)|=\cO(n/\ell)$ of the index.

Here, we improve on the index proposed by Loukides and Pissis~\cite{DBLP:conf/esa/LoukidesP21} by addressing these aspects as follows:
\begin{itemize}
\item We introduce the notion of randomized bd-anchors; the randomized counterpart of bd-anchors. Informally, the set $\mathcal{A}^{\text{ran}}_{\ell}(S)$ of \emph{randomized bd-anchors} of order $\ell$ of a string $S$ of length $n$ is the set of starting positions of the leftmost ``smallest'' rotation of every length-$\ell$ fragment of $S$ (see Section~\ref{sec:rr-bd-anchors} for a formal definition). The order of rotations is defined based on a random hash function \emph{and}, additionally, on the standard lexicographic order. We show that randomized bd-anchors are $\cO(n/\ell)$ in expectation.
\item We design a novel average-case $\cO(n)$-time algorithm to compute either $\mathcal{A}_{\ell}(S)$ or $\mathcal{A}^{\text{ran}}_{\ell}(S)$.  To achieve this, we use minimizers~\cite{DBLP:conf/sigmod/SchleimerWA03,DBLP:journals/bioinformatics/RobertsHHMY04}, another well-known type of lc-anchors, as anchors to compute $\mathcal{A}_{\ell}(S)$ (or $\mathcal{A}^{\text{ran}}_{\ell}(S)$) after carefully setting the sampling parameters. We employ longest common extension (LCE) queries~\cite{DBLP:conf/stoc/KempaK19} on $S$ to compare anchored rotations (i.e., rotations of $S$ starting at minimizers) of fragments of $S$ efficiently. The fact that minimizers are $\cO(n/\ell)$ in expectation lets us realize the average-case $\cO(n)$-time computation of $\mathcal{A}_{\ell}(S)$ (or $\mathcal{A}^{\text{ran}}_{\ell}(S)$). 
\item We propose a semi-external-memory and an internal-memory implementation to construct the index in small space using near-optimal work when the set of bd-anchors is given. We first compute $\mathcal{A}_{\ell}(S)$ (or $\mathcal{A}^{\text{ran}}_{\ell}(S)$) using $\cO(\ell)$ space using the aforementioned algorithm. For the semi-external-memory implementation, we show that if we have $\mathcal{A}_{\ell}(S)$ (or $\mathcal{A}^{\text{ran}}_{\ell}(S)$) in internal memory and the suffix array of $S$ in external memory, it suffices to scan the suffix array sequentially to deduce the information required to construct the index thus using $\cO(\ell+|\mathcal{A}_{\ell}(S)|)$ space. For the internal-memory implementation, we use \emph{sparse suffix sorting}~\cite{DBLP:conf/latin/AyadLPV24} on $\mathcal{A}_{\ell}(S)$ (or $\mathcal{A}^{\text{ran}}_{\ell}(S)$) to construct the index in $\cO(\ell+|\mathcal{A}_{\ell}(S)|)$ space.
\item We present an extensive experimental evaluation using real benchmark datasets. First, we show that our new algorithm for computing $\mathcal{A}_{\ell}(S)$ (or $\mathcal{A}^{\text{ran}}_{\ell}(S)$) is \emph{more than two orders of magnitude faster} as $\ell$ increases than the simple $\Theta(n\ell)$-time algorithm, while using a very similar amount of memory. 
We also show that randomized bd-anchors outperform bd-anchors \emph{in all measures}: construction time, construction space, and sample size.
We then go on to examine the index construction based on the above four measures. The results show that, for long patterns, the index constructed using randomized bd-anchors and our improved construction algorithms compares favorably to all classic indexes (implemented using \textsf{sdsl-lite}): (compressed) suffix tree; (compressed) suffix array; and the FM-index. For instance, our index offers about $27\%$ faster query time, for all datasets on average, for  all $\ell \in [32,1024]$, compared to the suffix array (which performs best among the competitors in this measure), while occupying up to two orders of magnitude less space. Also, our index occupies up to $5\times$ less space on average over all datasets, for $\ell=1024$, compared to the FM-index (which performs best among the competitors in this measure), while being up to one order of magnitude faster in query time. As another example, our index on the full human genome occupies 11MB for $\ell=2^{14}$ and answers queries more than $72\times$ faster than the FM-index, which  occupies more than 1GB. 
\item The above algorithms are based on an \emph{average-case} analysis. In the worst case, however,  $|\mathcal{A}_{\ell}(S)|$ (or $|\mathcal{A}^{\text{ran}}_{\ell}(S)|$) are in $\Theta(n)$. This implies that our index does not offer satisfactory worst-case guarantees. Motivated by this, we also present a counterpart of our index with \emph{worst-case} guarantees based on the lc-anchors notion of \emph{partitioning sets}~\cite{DBLP:conf/cpm/KosolobovS24}. This construction is mainly of theoretical interest. However, we stress that, to the best of our knowledge, it is the first index achieving the best of all worlds in the regime where we have at hand a lower bound $\ell$ on the length of the queried patterns. Specifically, this index occupies $\cO(n/\ell)$ space, it can be constructed in near-optimal time using $\cO(n/\ell)$ space, and it supports near-optimal pattern matching queries.
\end{itemize}

\subsection{Other related work}
The connection of lc-anchors to text indexing and other string-processing applications is not new.
The perhaps most popular lc-anchors in practical applications are \emph{minimizers}, which have been introduced independently by Schleimer et al.~\cite{DBLP:conf/sigmod/SchleimerWA03} and by Roberts et al.~\cite{DBLP:journals/bioinformatics/RobertsHHMY04} (see Section~\ref{sec:prel} for a definition). Although minimizers have been mainly used for sequence comparison~\cite{Winnowmap2}, Grabowski and Raniszewski~\cite{DBLP:journals/spe/GrabowskiR17} showed how minimizers can be used to sample the suffix array -- see also~\cite{DBLP:journals/jda/ClaudeNPST12}, which uses an alphabet sampling approach. Another well-known notion of lc-anchors is \emph{difference covers}, which have been introduced by Burkhardt and K{\"{a}}rkk{\"{a}}inen~\cite{DBLP:conf/cpm/BurkhardtK03} for suffix array construction in small space (see also~\cite{DBLP:journals/tocs/Maekawa85}). Difference covers play also a central role in the elegant linear-time suffix array construction algorithm of K{\"{a}}rkk{\"{a}}inen et al.~\cite{DBLP:journals/jacm/KarkkainenSB06}; and they have been used in other string processing applications~\cite{DBLP:conf/cpm/Charalampopoulos18,DBLP:conf/cpm/Nun0KK20}. Another very powerful type of lc-anchors is  \emph{string synchronizing sets}, which have been recently proposed by Kempa and Kociumaka~\cite{DBLP:conf/stoc/KempaK19},  for constructing, among others, an optimal data structure for LCE queries (see also~\cite{DBLP:conf/esa/Dinklage0HKK20}). String synchronizing sets have applications in designing sublinear-time algorithms for classic string problems~\cite{DBLP:conf/esa/Charalampopoulos21,DBLP:conf/cpm/Charalampopoulos22}, whose textbook linear-time solutions rely on suffix trees. Another type of lc-anchors, very much related to string synchronizing sets, is the notion of \emph{partitioning sets}~\cite{DBLP:conf/soda/BirenzwigeGP20,DBLP:conf/cpm/KosolobovS24}.

\subsection{Paper organization and differences from the conference version} 
In Section~\ref{sec:prel}, we present the necessary definitions and notation, a brief overview of classic text indexes, as well as some existing results on minimizers and bd-anchors. In Section~\ref{sec:index}, we present a brief overview of the index proposed by Loukides and Pissis~\cite{DBLP:conf/esa/LoukidesP21}. In Section~\ref{sec:rr-bd-anchors}, we describe randomized bd-anchors, our new sampling mechanism, in detail. In Section~\ref{sec:eng}, we improve the construction of the index of Section~\ref{sec:index} by
presenting: \textbf{(i)} our fast algorithm for (randomized) bd-anchors computation (see Section~\ref{sec:bd-anchors}); and \textbf{(ii)} the index construction in small space using near-optimal work (see Section~\ref{sec:small}). 
In Section~\ref{sec:worst-case}, we present the counterpart of our index offering worst-case guarantees.
In Section~\ref{sec:imp}, we provide the full details of our implementations, and in Section~\ref{sec:exp}, we present our extensive experimental evaluation. We conclude this paper in Section~\ref{sec:finale}.

A preliminary version of this paper was announced at PVLDB 2023~\cite{DBLP:journals/pvldb/AyadLP23}. Therein we identified two ways for potentially improving the \emph{construction time} of the index: \textbf{(i)} design a new sampling mechanism that reduces the number of minimizers to potentially improve the construction of bd-anchors; \textbf{(ii)} plug sparse suffix sorting in the semi-external-memory implementation to make it work fully in internal memory thus potentially improving the construction of the index. In response (to \textbf{(i)}), we introduce here the notion of randomized bd-anchors (see Section~\ref{sec:rr-bd-anchors}). It is well-known than randomized minimizers have usually a lower density than their lexicographic counterparts~\cite{DBLP:journals/bioinformatics/ZhengKM20}; the former can also be computed faster and using less space as they can be computed using a rolling hash function~\cite{DBLP:journals/ibmrd/KarpR87}. Unfortunately, for bd-anchors, it does not make sense to maintain the hash value of the rotations in every window. Consider the hash value of the rotation starting at position $i$ in the window starting at position $j<i$; and consider that this $i$ is the smallest rotation in $S[j\dd j+\ell-1]$ according to the hash function used. The hash value of the prefix $S[i\dd j+\ell-1]$ of the rotation is in general completely different than the one of $S[i\dd j+\ell]$: the substring obtained by appending the next letter. Thus when shifting from window $S[j\dd j+\ell-1]$ to window $S[j+1\dd j+\ell]$, the hash value of $S[i\dd j+\ell]$ changes completely and thus \emph{the minimizer position will likely change at the next window}. This is why we partition the rotations in two parts: the first one is of length $r$; and the second one is of length $\ell-r$. We compute the \emph{hash value} of the former and the \emph{lexicographic rank} of the latter. These two integers form a pair. The rotation sampled is the rotation with the smallest such pair (in lexicographic order). This lets us \emph{simultaneously} exploit the benefits of randomized minimizers (small sample, small space, and fast computation) and bd-anchors (effective tie-breaking). For \textbf{(ii)}, we plug the sparse suffix sorting implementation of Ayad et al.~\cite{DBLP:conf/latin/AyadLPV24} in our semi-external-memory implementation (see Section~\ref{sec:eng}).

Since our new implementations (see Section~\ref{sec:imp}) based on the aforementioned two changes (randomized bd-anchors and internal-memory implementation) have dramatically improved the construction in the preliminary version~\cite{DBLP:journals/pvldb/AyadLP23} (which was based on lexicographic bd-anchors and on external memory), \emph{the experimental section is also completely new} (see Section~\ref{sec:exp}): we have re-conducted all experiments from scratch; we have re-plotted the results, and a new description has been added.

Section~\ref{sec:worst-case}, presenting the counterpart of our index with
worst-case guarantees, is also completely new.

\section{Preliminaries}\label{sec:prel}

An \emph{alphabet} $\Sigma$ is a finite set of elements called \emph{letters}; we denote by $\sigma$ the size $|\Sigma|$ of $\Sigma$. A \emph{string} $S=S[1\dd n]$ is a sequence of letters over some alphabet $\Sigma$; we denote by $|S|=n$ the length of $S$. The fragment $S[i\dd j]$ of $S$ is an \emph{occurrence} of the underlying \emph{substring} $P=S[i]\ldots S[j]$. We also write that $P$ occurs at \emph{position} $i$ in $S$ when $P=S[i] \ldots S[j]$. A {\em prefix} of $S$ is a fragment of $S$ of the form $S[1\dd j]$ and a {\em suffix} of $S$ is a fragment of $S$ of the form $S[i\dd n]$. Given a string $S$ and an integer $1\leq i \leq |S|$, we define $S[i \dd |S|]S[1\dd i-1]$ to be the \emph{$i$th rotation} of $S$. Given a string $S$ of length $n$, we denote by $\overleftarrow{S}$ the \emph{reverse} $S[n]\ldots S[1]$ of $S$. 

\defDSproblem{Text Indexing}{A string $S$ of length $n$ over an integer alphabet of size $\sigma=n^{\cO(1)}$.}{Given a string $P$ of length $m$, report all positions $i\iff P=S[i\dd i+m-1]$.}

\paragraph{Classic text indexes}

For a string $S$ of length $n$ over an ordered alphabet of size $\sigma$, the \emph{suffix array} $\textsf{SA}[1\dd n]$ stores the permutation of $\{1,\ldots, n\}$ such that $\textsf{SA}[i]$ is the starting position of the $i$th lexicographically smallest suffix of $S$. The standard application of $\textsf{SA}$ is text indexing, in which we consider $S$ to be the \emph{text}: given any string $P[1\dd m]$, known as the \emph{pattern}, the suffix array of $S$ allows us to report all $\occ$ occurrences of $P$ in $S$ using only $\cO(m \log n + \occ)$ operations~\cite{DBLP:journals/siamcomp/ManberM93}. 
We perform binary search in $\textsf{SA}$ resulting in a range $[s,e)$ of suffixes of $S$ having $P$ as a prefix. Then, $\textsf{SA}[s\dd e-1]$ contains the starting positions of all occurrences of $P$ in $S$. The $\textsf{SA}$ is often augmented with the $\textsf{LCP}$ array~\cite{DBLP:journals/siamcomp/ManberM93} storing the length of longest common prefixes of lexicographically adjacent suffixes (i.e., consecutive entries in the \textsf{SA}). In this case, reporting all $\occ$ occurrences of $P$ in $S$ can be done in $\cO(m+ \log n + \occ)$ time by avoiding to compare $P$ with suffixes of $S$ from scratch during binary search~\cite{DBLP:journals/siamcomp/ManberM93}\label{textindexing} (see~\cite{DBLP:journals/algorithmica/0001KL15,DBLP:conf/cpm/0001G15,DBLP:journals/siamcomp/NavarroN17} for subsequent improvements). The suffix array occupies $\Theta(n)$ space and it can be constructed in $\cO(n)$ time, when $\sigma=n^{\cO(1)}$, using $\cO(n)$ space~\cite{DBLP:conf/focs/Farach97}.
Given \textsf{SA} of $S$, we can compute the \textsf{LCP} array of $S$ in $\cO(n)$ time~\cite{DBLP:conf/cpm/KasaiLAAP01}.

Given a set $\mathcal{F}$ of strings, the \emph{compacted trie} of these strings is the trie obtained by compressing each path of nodes of degree one in the trie of the strings in $\mathcal{F}$, which takes $\cO(|\mathcal{F}|)$ space~\cite{DBLP:journals/jacm/Morrison68}. Each edge in the compacted trie has a label represented as a fragment of a string in $\mathcal{F}$. The \emph{suffix tree} of $S$, which we denote by $\textsf{ST}(S)$, is the compacted trie of the suffixes of $S$~\cite{DBLP:conf/focs/Weiner73}. Assuming $S$ ends with a unique terminating symbol, every leaf in $\textsf{ST}(S)$ represents a suffix $S[i\dd n]$ and is decorated by index $i$. The set of indices stored at the leaf nodes in the subtree rooted at node $v$ is the \emph{leaf-list} of $v$, and we denote it by $LL(v)$. Each edge in $\textsf{ST}(S)$ is labelled with a nonempty substring of $S$ such that the path from the root to the leaf annotated with index $i$ spells the suffix $S[i\dd n]$. The substring of $S$ spelled by the path from the root to node $v$ is the \emph{path-label} of $v$, and we denote it by $L(v)$. Given any pattern $P[1\dd m]$, $\textsf{ST}(S)$ allows us to report all $\occ$
occurrences of $P$ in $S$ using only $\cO(m \log \sigma + \occ)$ operations. We simply spell $P$ from the root of $\textsf{ST}(S)$ (to access edges by the first letter of their label, we use binary search) until we arrive (if possible) at the first node $v$ such that $P$ is a prefix of $L(v)$. Then all $\occ$ occurrences (starting positions) of $P$ in $S$ are $LL(v)$. The suffix tree occupies $\Theta(n)$ space and it can be constructed in $\cO(n)$ time, when $\sigma=n^{\cO(1)}$, using $\cO(n)$ space~\cite{DBLP:conf/focs/Farach97}.
To improve the query time to the optimal $\cO(m + \occ)$ we can use randomization to construct a perfect hash table~\cite{DBLP:journals/jacm/FredmanKS84} to access edges by the first letter of their label in $\cO(1)$ time.

Let $S$ be a string of length $n$. Given two integers $1 \leq i,j \leq n$, we denote by $\LCP_S(i,j)$ the length of the longest common prefix (LCP) of $S[i\dd n]$ and $S[j\dd n]$. When $S$ is over an integer alphabet of size $\sigma=n^{\cO(1)}$, we can construct a data structure in $\cO(n/\log_\sigma n)$ time that answers $\LCP_S(i,j)$ queries in $\cO(1)$ time~\cite{DBLP:conf/stoc/KempaK19}.

\paragraph{Lc-anchors}

Given a string $S$ of length $n$, two integers $w,k>0$, and the $i$th length-$(w+k-1)$ fragment $F=S[i\dd i+w+k-2]$ of $S$, the \emph{$(w,k)$-minimizers} 
of $F$ are defined as the positions $j\in[i,i+w)$ where a lexicographically minimal length-$k$ substring of $F$ occurs~\cite{DBLP:journals/bioinformatics/RobertsHHMY04}. The set $\mathcal{M}_{w,k}(S)$ of $(w,k)$-minimizers of $S$ is defined as the set of $(w,k)$-minimizers of each fragment $S[i\dd i+w+k-2]$, for all $i\in[1,n-w-k+2]$. 

\begin{example}[Minimizers]
Let $S=\texttt{aacaaacgcta}$ and $w=k=3$. We consider fragments of length $w+k-1=5$. The first fragment is \texttt{aacaa}. The lexicographically minimal length-$k$ substring is \texttt{aac}, starting at position $1$, 
so we add position $1$ to $\mathcal{M}_{3,3}(S)$. The second fragment is \texttt{acaaa}. The lexicographically minimal length-$k$ substring is \texttt{aaa}, starting at position $4$,
so we add position $4$ to $\mathcal{M}_{3,3}(S)$. 
The third fragment, \texttt{caaac}, and fourth fragment, \texttt{aaacg}, have \texttt{aaa} as the lexicographically minimal length-$k$ substring, so $\mathcal{M}_{3,3}(S)$ does not change.
The fifth fragment is \texttt{aacgc}. The lexicographically minimal length-$k$ substring is \texttt{aac}, starting at position $5$, and so we add position $5$ to $\mathcal{M}_{3,3}(S)$. The sixth fragment is \texttt{acgct}. The lexicographically minimal length-$k$ substring is \texttt{acg}, starting at position $6$, and so we add position $6$ to $\mathcal{M}_{3,3}(S)$. The seventh fragment is \texttt{cgcta}. The lexicographically minimal length-$k$ substring is \texttt{cgc}, starting at position $7$, and so we add position $7$ to $\mathcal{M}_{3,3}(S)$. Thus  $\mathcal{M}_{3,3}(S)=\{1,4,5,6,7\}$.
\end{example}

\begin{lemma}[\cite{DBLP:journals/bioinformatics/ZhengKM20}]\label{lem:minimizer}
If $S$ is a string of length $n$, randomly generated by a memoryless source over an alphabet of size $\sigma \geq 2$ with identical letter probabilities, then the expected size of $\mathcal{M}_{w,k}(S)$ is $\cO(n/w)$ if and only if $k\geq \log_\sigma w + \cO(1)$.
\end{lemma}

\begin{lemma}[\cite{DBLP:conf/esa/LoukidesP21}]\label{lem:minimizers_con}
For any string $S$ of length $n$ over an integer alphabet of size $n^{\cO(1)}$ and integers $w,k>0$,
the set $\mathcal{M}_{w,k}$ can be computed in $\cO(n)$ time.
\end{lemma}

Loukides and Pissis defined the following alternative notion of lc-anchors~\cite{DBLP:conf/esa/LoukidesP21}.  

\begin{definition}[Bidirectional anchor]
Given a string $F$ of length $\ell>0$, the \emph{bidirectional anchor} (bd-anchor) of $F$ is the lexicographically minimal rotation $j\in[1,\ell]$ of $F$ with minimal $j$. The set of order-$\ell$  bd-anchors of a string $S$ of length $n>\ell$, for some integer $\ell>0$, is defined as the set $\mathcal{A}_{\ell}(S)$ of bd-anchors of $S[i\dd i+\ell-1]$, for all $i\in [1,n-\ell+1]$. 
\end{definition}

\begin{example}[Bd-anchors]
Let $S=\texttt{aacaaacgcta}$ and $\ell=5$. We consider fragments of length $\ell=5$. The first fragment is \texttt{aacaa}. We need to consider all of its rotations and select the leftmost lexicographically minimal. All rotations of \texttt{aacaa} are: \texttt{aacaa}, \texttt{acaaa}, \texttt{caaaa}, \texttt{aaaac}, and \texttt{aaaca}. We thus select \texttt{aaaac}, which starts at position $4$, and add position $4$ to $\mathcal{A}_{5}(S)$. The second fragment is \texttt{acaaa}. All rotations of \texttt{acaaa} are: \texttt{acaaa}, \texttt{caaaa}, \texttt{aaaac}, \texttt{aaaca}, and \texttt{aacaa}. We thus select \texttt{aaaac}, which starts at position $4$, and so we do not need to add position $4$ to  $\mathcal{A}_{5}(S)$. The third fragment is \texttt{caaac}. All rotations of \texttt{caaac} are: \texttt{caaac}, \texttt{aaacc}, \texttt{aacca}, \texttt{accaa}, and \texttt{ccaaa}. We thus select \texttt{aaacc}, which starts at position $4$, and so we do not need to add position $4$ to $\mathcal{A}_{5}(S)$. The fourth fragment is \texttt{aaacg}, which is also the lexicographically minimal rotation. This rotation starts at position $4$, and so we do not need to add position $4$ to $\mathcal{A}_{5}(S)$. The fifth fragment is \texttt{aacgc}, which is also the lexicographically minimal rotation. This rotation starts at position $5$, and so we add position $5$ to $\mathcal{A}_{5}(S)$. The sixth fragment is \texttt{acgct}, which is also the lexicographically minimal rotation. This rotation starts at position $6$, and so we add position $6$ to $\mathcal{A}_{5}(S)$. The seventh fragment is \texttt{cgcta}. The lexicographically minimal rotation of this fragment is \texttt{acgct}, which starts at position $11$, and so we add position $11$ to $\mathcal{A}_{5}(S)$. Thus $\mathcal{A}_{5}(S)=\{4,5,6,11\}$.\label{ex:bd-anchors}
\end{example}

The bd-anchors notion was also parameterized by Loukides et al.~\cite{TKDE2023} according to the following definition.

\begin{definition}[Reduced bidirectional anchor]
Given a string $F$ of length $\ell>0$ and an integer $0\leq r\leq \ell-1$, we define the \emph{reduced bidirectional anchor} of $F$ as the lexicographically minimal rotation $j\in[1,\ell-r]$ of $F$ with minimal $j$. The set of order-$\ell$ reduced bd-anchors of a string $S$ of length $n>\ell$ is defined as the set $\mathcal{A}_{\ell,r}(S)$ of reduced bd-anchors of $S[i\dd i+\ell-1]$, for all $i\in [1,n-\ell+1]$. 
\end{definition}

Informally, when using the reduced bd-anchor mechanism, we neglect the $r$ rightmost rotations from the sampling process.

\begin{example}[Reduced bd-anchors]
Let $S=\texttt{aacaaacgcta}$, $\ell=5$, and $r=1$. Recall from Example~\ref{ex:bd-anchors} that $\mathcal{A}_{5}(S)=\{4,5,6,11\}$. To see the difference between bd-anchors and reduced bd-anchors, consider the seventh (last) fragment of length $\ell=5$ of $S$. This fragment is \texttt{cgcta} and its rotations are: \texttt{cgcta}, \texttt{gctac}, \texttt{ctacg}, \texttt{tacgc}, and \texttt{acgct}. The lexicographically minimal rotation is \texttt{acgct} and this is why $11\in \mathcal{A}_{5}(S)$. In the case of reduced bd-anchors we are asked to neglect the $r=1$ rightmost rotations, and so the only candidates are: \texttt{cgcta}, \texttt{gctac}, \texttt{ctacg}, and \texttt{tacgc}. Out of these, the lexicographically minimal rotation is \texttt{cgcta}, which starts at position $7$, and this is why $\mathcal{A}_{5,1}(S)=\{4,5,6,7\}$.\label{example:3}
\end{example}

\begin{lemma}[\cite{DBLP:conf/esa/LoukidesP21,TKDE2023}]\label{lem:bd-anchors}
If $S$ is a string of length $n$, randomly generated by a memoryless source over an alphabet of size $\sigma \geq 2$ with identical letter probabilities, then, for any integer $\ell>0$, the expected size of $\mathcal{A}_{\ell,r}(S)$ with $r=\lceil4 \log\ell /\log \sigma\rceil$ is in $\cO(n/\ell)$.
\end{lemma}

\begin{lemma}[\cite{DBLP:conf/esa/LoukidesP21,TKDE2023}]\label{lem:compute-bda}
For any string $S$ of length $n$ over an integer alphabet of size $n^{\cO(1)}$ and integers $\ell,r>0$,
the set $\mathcal{A}_{\ell,r}(S)$ can be computed in $\cO(n)$ time.
\end{lemma}

\section{The index}\label{sec:index}

\begin{figure*}[ht]
     \centering
     \begin{subfigure}[b]{0.49\textwidth}
         \centering
         \includegraphics[width=7cm]{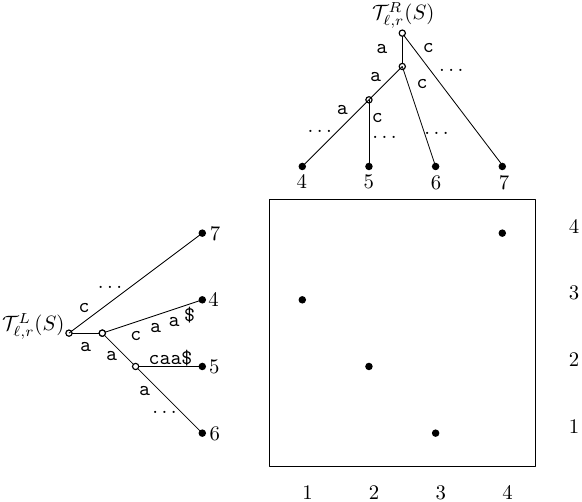}
         \caption{The $\mathcal{I}_{\ell,r}(S)$ index for $S=\texttt{aacaaacgcta}$, $\ell=5$, and $r=1$.}
         \label{fig:ds}
     \end{subfigure}
     \hfill
     \begin{subfigure}[b]{0.49\textwidth}
         \centering
         \includegraphics[width=7cm]{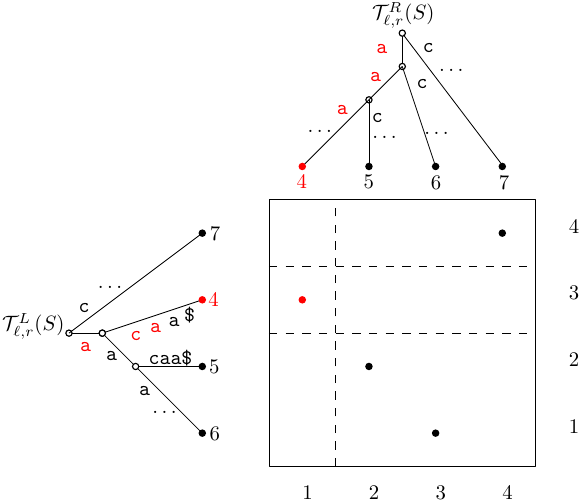}
         \caption{Querying $P=\texttt{acaaa}$.}
         \label{fig:q}
     \end{subfigure}
         \caption{\label{fig:index}Let $S=\texttt{aacaaacgcta}$, $\ell=5$, and $r=1$.  We have that $\mathcal{A}_{5,1}(S)=\{4,5,6,7\}$, and thus the indexed suffixes and reversed suffixes of $S$ are as shown in (a). Assume that we have a query pattern $P=\texttt{acaaa}$. We first find the reduced bd-anchor of $P[1\dd \ell]$ for $\ell=5$ and $r=1$, which is $j=3$: the minimal lexicographic rotation is \texttt{aaaac}. We then split $P$ in $\protect\overleftarrow{P[1\dd 3]}=\texttt{aca}$ and $P[3\dd |P|]=\texttt{aaa}$, and
         search $\protect\overleftarrow{P[1\dd 3]}=\texttt{aca}$ in $\mathcal{T}^L_{\ell,r}$ and
         $P[3\dd |P|]=\texttt{aaa}$ in $\mathcal{T}^R_{\ell,r}$. The search induces a rectangle, which encloses all the occurrences of $P$ in $S$, as shown in (b). Indeed, $P=S[2\dd 6]=\texttt{acaaa}$: the fragment $S[2 \dd 6]$ is anchored at position $4$.}
\end{figure*}

Loukides and Pissis~\cite{DBLP:conf/esa/LoukidesP21} proposed the following text index, which is based on (reduced~\cite{TKDE2023}) bd-anchors. Fix string $S$ of length $n$ over an integer alphabet of size $n^{\cO(1)}$ as well as $\ell$ (and $r$). Given $\mathcal{A}_{\ell,r}(S)$ we construct two compacted tries: one for strings $S[i\dd n]$, for all $i\in \mathcal{A}_{\ell,r}(S)$, which we denote by $\mathcal{T}^R_{\ell,r}(S)$; and one for strings $\overleftarrow{S[1\dd i]}$, for all $i\in \mathcal{A}_{\ell,r}(S)$, which we denote by $\mathcal{T}^L_{\ell,r}(S)$. We also construct a 2D range reporting data structure~\cite{DBLP:conf/compgeom/ChanLP11} over $|\mathcal{A}_{\ell,r}(S)|$ points $(x,y)$, where $x$ is the lexicographic rank of $S[i\dd n]$ and $y$ is the lexicographic rank of $\overleftarrow{S[1\dd i]}$, for all $i\in \mathcal{A}_{\ell,r}(S)$. Upon a query $P$ of length $|P|\geq \ell$, we find the (reduced) bd-anchor of $P[1\dd \ell]$, say $j$, which gives $P[j\dd |P|]$ and $\overleftarrow{P[1\dd j]}$. We search the former in $\mathcal{T}^R_{\ell,r}$ and the latter in $\mathcal{T}^L_{\ell,r}$. This search induces a rectangle (inspect Figure~\ref{fig:index}, for an example), which we use to query the 2D range reporting data structure. The reported points are all occurrences of $P$ in $S$. We denote the resulting index by $\mathcal{I}_{\ell,r}(S)$. Loukides and Pissis showed the following result~\cite{DBLP:conf/esa/LoukidesP21,TKDE2023}.

\begin{theorem}[\cite{DBLP:conf/esa/LoukidesP21,TKDE2023}]\label{the:index}
Given any string $S$ of length $n$ over an integer alphabet of size $n^{\cO(1)}$ and integers $\ell,r>0$, the index $\mathcal{I}_{\ell,r}(S)$ occupies $\cO(|\mathcal{A}_{\ell,r}(S)|)$ extra space and reports all $\occ$ occurrences of any pattern $P$ of length $|P|\geq \ell$ in $S$ in $\ctO(|P|+\occ)$ time. Moreover, the index $\mathcal{I}_{\ell,r}(S)$ can be constructed in $\ctO(n)$ time.
\end{theorem}

\emph{Extra} refers to the $n\lceil \log \sigma \rceil$ bits required to store $S$. Alternatively one could use $nH_0(S) + o(nH_0(S)) + o(n)$ bits without any penalty~\cite{DBLP:journals/algorithmica/BarbayCGNN14,DBLP:journals/talg/BelazzouguiN15}, where $H_0(S)$ is the \emph{zeroth-order entropy} of $S$. Note that the same index can be constructed for any lc-anchors; e.g., for minimizers.

\section{Randomized reduced bidirectional string anchors}\label{sec:rr-bd-anchors}

We introduce the notion of randomized reduced bd-anchors
by injecting Karp-Rabin (KR) fingerprints~\cite{DBLP:journals/ibmrd/KarpR87} to the notion of reduced bd-anchors~\cite{TKDE2023}. This is analogous to the idea of using randomization to select minimizers~\cite{DBLP:conf/sigmod/SchleimerWA03}
instead of using the standard lexicographic total order~\cite{DBLP:journals/bioinformatics/RobertsHHMY04}. In practice, implementing a randomized minimizers mechanism is straightforward via using a hash function (e.g., KR fingerprints). The other benefit of randomized minimizers is that they usually have lower density than their lexicographic counterparts~\cite{DBLP:journals/bioinformatics/ZhengKM20}. 

Unfortunately, it is not as straightforward to inject randomization to (reduced) bd-anchors. Intuitively, this is because bd-anchors wrap around the window. We next explain how this can be achieved by using a mix of lexicographic total order and KR fingerprints. Importantly, our experiments show a clear superiority (lower density) of this new randomized notion over the deterministic (lexicographic) counterpart.

Let us start by introducing KR fingerprints.
Let $S$ be a string of length $n$ over an integer alphabet.
Let $p$ be a prime and choose $r \in[0,p-1]$ uniformly at random. 
The KR fingerprint of $S[i\dd j]$ is defined as: 
$$\phi_S(i,j)=\sum^{j}_{k=i}S[k]r^{j-k}\mod p.$$

Clearly, if $S[i\dd i + \ell] = S[j\dd j + \ell]$ then $\phi_S(i,i+\ell)=\phi_S(j,j+\ell)$. 
On the other hand, if
$S[i\dd i + \ell] \neq S[j\dd j + \ell]$ then $\phi_S(i,i+\ell)\neq \phi_S(j,j+\ell)$ with probability at least $1-\ell/p$~\cite{DBLP:conf/icalp/DietzfelbingerGMP92}.
Since we are comparing only substrings of equal length, the number of different possible substring comparisons is less than $n^3$. Thus, for any constant $c\geq 1$, we can set $p$ to be a
prime larger than $\max(|\Sigma|,n^{c+3})$ to make the KR fingerprint
function perfect (i.e., no collisions) with probability at least $1 - n^{-c}$ (with high probability).
Any KR fingerprint of $S$ or $p$ fits in one machine word of $\Theta(\log n)$ bits. In our methods, the length $j-i+1$ of the substrings of $S$ considered will be fixed, hence the quantity $r^{j-i+1}\mod p$ needs to be computed only once.

The definition of randomized reduced bd-anchors (Definition~\ref{def:rr}) is a bit involved. We will thus start with some intuition to facilitate the reader's comprehension. Given $\ell$ and $r$ (exactly as in the reduced bd-anchors case) and a fragment $F$ of $S$, we select as \emph{candidates}, the $(\ell-r,r+1)$-minimizers of $F$ based on the total order induced by the KR fingerprints. If there exists only one candidate selected, then this is the randomized reduced bidirectional anchor of $F$.
If there are many candidates, then we select the leftmost lexicographically smallest rotation out of all the rotations starting right after the last position of the candidates.
The candidate whose rotation is selected is then selected as the randomized reduced bidirectional anchor of $F$.
We next formalize this intuition.

\begin{definition}[Randomized reduced bidirectional anchor]\label{def:rr}
Given a string $F$ of length $\ell>0$ and an integer $0\leq r\leq \ell-1$, we define the \emph{randomized reduced bidirectional anchor} of $F$ as the starting position of the length-$(r+1)$ fragment of $F$ with the minimal KR fingerprint. 
If there are $h>1$ such fragments $F[i_1\dd j_1],\ldots,F[i_h\dd j_h]$ in $F$ (i.e., $h$ fragments that share the same minimal KR fingerprint), then we resolve the ties 
by choosing $j-(r+1)$ as the randomized reduced bidirectional anchor of $F$, where $j$ is the lexicographically minimal rotation from $\{j_1+1,\ldots,j_h+1\}$ with minimal $j$. The set of order-$\ell$ randomized reduced bidirectional anchors of a string $S$ of length $n>\ell$ is defined as the set $\mathcal{A}^{\text{ran}}_{\ell,r}(S)$ of randomized reduced bd-anchors of $S[i\dd i+\ell-1]$, for all $i\in [1,n-\ell+1]$.   
\end{definition}

\begin{example}[Randomized reduced bd-anchor]
Let $F=\texttt{aacaaacgcta}$ of length $\ell=11$ and $r=2$.
We need to consider the following set of length-$3$ substrings of $F$: $$\{\texttt{aac},\texttt{aca}, \texttt{caa}, \texttt{aaa}, \texttt{acg}, \texttt{cgc}, \texttt{gct}, \texttt{cta}\}.$$ 
Assume that $\phi_F(1,3)$, that is, the KR fingerprint of $F[1\dd 3]=F[5\dd 7]=\texttt{aac}$, is the smallest one among the KR fingerprints of all length-$3$ substrings of $F$. We have that $h=2$ because \texttt{aac} occurs at two positions: $1$ and $5$. To resolve the tie, we need to consider the rotations starting at positions in $\{4,8\}$. These are:
$\texttt{aaacgctaaac}$
and
$\texttt{gctaaacaaac}$.
The former $(j=4)$ is the lexicographically smallest, and so we select $j-(r+1)=4-(2+1)=1$ as the randomized reduced bd-anchor of $F$.
\end{example}

Randomized reduced bd-anchors are a new type of lc-anchors (see Section~\ref{sec:prel}). Indeed, like minimizers~\cite{DBLP:journals/jcb/ZhengMK23} and (reduced) bd-anchors~\cite{TKDE2023}, it is straightforward to prove that randomized reduced bd-anchor samples enjoy the following two useful properties (see Section~\ref{sec:intro}):
\begin{itemize}
    \item \emph{Property 1 (approximately uniform sampling)}
    \item \emph{Property 2 (local consistency)}
\end{itemize}

We denote by $\mathcal{M}^{\text{ran}}_{w,k}(S)$ the set of $(w,k)$-minimizers of $S$ when the underlying total order on length-$k$ substrings is determined randomly (e.g., by a random hash function). The following lemma is known.

\begin{lemma}[\cite{DBLP:journals/bioinformatics/ZhengKM20}]\label{lem:random_minimizers}
If $S$ is a string of length $n$, randomly generated by a memoryless source over an alphabet of size $\sigma \geq 2$ with identical letter probabilities, then the expected size of $\mathcal{M}^{\text{ran}}_{w,k}(S)$ with $k>(3+\epsilon)\log_\sigma(w+1)$ is in $\cO(n/w)$, for any $\epsilon>0$.
\end{lemma}

The following lemma follows directly from Lemma~\ref{lem:random_minimizers} and the fact that by construction $|\mathcal{A}^{\text{ran}}_{\ell,r}(S)|\leq |\mathcal{M}^{\text{ran}}_{w,k}(S)|$, where $k=r+1$ and $w=\ell-r$. 

\begin{lemma}\label{lem:rrbd-anchors}
If $S$ is a string of length $n$, randomly generated by a memoryless source over an alphabet of size $\sigma \geq 2$ with identical letter probabilities, then, for any integer $\ell>0$, the expected size of $\mathcal{A}^{\text{ran}}_{\ell,r}(S)$ with $r=\lceil4 \log\ell /\log \sigma\rceil$ is in $\cO(n/\ell)$.
\end{lemma}

Comparing Lemma~\ref{lem:rrbd-anchors} with Lemma~\ref{lem:bd-anchors}, we see that this new notion of bd-anchors has the same expected size in the worst case as reduced bd-anchors. However, as mentioned before, our experiments show a clear superiority (lower density) of this new randomized notion over the deterministic (lexicographic) counterpart.
In the next section, we show that both notions can be computed in linear time
with the same technique.

\section{Improving the index}\label{sec:eng}

In this section, we improve the construction of the index $\mathcal{I}_{\ell,r}(S)$ (Theorem~\ref{the:index}) focusing on two aspects: \textbf{(i)} the fast computation of (randomized) reduced bd-anchors (see Section~\ref{sec:bd-anchors}); and \textbf{(ii)} the index construction in small space using near-optimal work (see Section~\ref{sec:small}). 
With \emph{small}, we mean space that is close to the index size.

\subsection{Computing bd-anchors in linear time}\label{sec:bd-anchors}

Recall that the first step to construct the $\mathcal{I}_{\ell,r}(S)$ index is to compute the set $\mathcal{A}_{\ell,r}(S)$ of reduced bd-anchors of $S$.
Loukides and Pissis~\cite{DBLP:conf/esa/LoukidesP21,TKDE2023} showed a worst-case linear-time algorithm to compute $\mathcal{A}_{\ell,r}(S)$ (Lemma~\ref{lem:compute-bda}). 
Although this algorithm is optimal in the worst case, it seems quite complex to implement and it is unlikely to be efficient in practice.
The worst-case linear-time algorithm for computing $\mathcal{A}_{\ell,r}(S)$ relies on a data structure introduced by Kociumaka in~\cite{DBLP:conf/cpm/Kociumaka16}. The latter data structure relies heavily on the construction of \emph{fusion trees}~\cite{DBLP:conf/stoc/FredmanW90}, and hence it is mostly of theoretical interest: it is widely accepted that naive solutions can be more practical than such sophisticated data structures -- see~\cite{DBLP:conf/wea/Dinklage0H21} and references therein.
That is why Loukides and Pissis have instead proposed and implemented a simple $\Theta(n\ell)$-time algorithm to compute $\mathcal{A}_{\ell,r}(S)$ in their experiments~\cite{DBLP:conf/esa/LoukidesP21}.
Here we give a novel average-case linear-time algorithm to compute $\mathcal{A}_{\ell,r}(S)$. In the worst case, the algorithm runs in $\cO(n\ell)$ time as the simple algorithm from~\cite{DBLP:conf/esa/LoukidesP21} does. 

Let us remark that the same technique (with very small differences) can be used to compute $\mathcal{A}^{\text{ran}}_{\ell,r}(S)$, the set of randomized reduced bd-anchors of $S$. We will highlight these differences at the end of this section.

\paragraph{Main idea} We use $(w,k)$-minimizers as anchors to compute reduced bd-anchors of order $\ell$ after carefully setting parameters $w$ and $k$. We employ $\LCP_S(i,j)$ queries to compare anchored rotations (i.e., rotations of $S$ starting at $(w,k)$-minimizers) of fragments of $S$ efficiently. The fact that $(w,k)$-minimizers are $\cO(n/\ell)$ in expectation lets us realize $\cO(n)$ time on average. 

Let us start with the following simple fact.

\begin{fact}\label{fct:simple}
For any string $F$ of length $\ell>0$, the reduced bd-anchor of $F$, for any $1 \leq r\leq \ell-1$, is an  $(\ell-r,r+1)$-minimizer of $F$.
\end{fact}

\begin{proof}
Let $j$ be the reduced bd-anchor of $F$: the lexicographically smallest rotation of $F$ with minimal $j\in[1,\ell-r]$. By definition, $F[j\dd |F|]$ is of length at least $r+1$ and there cannot be another $j'$ such that $F[j'\dd j'+r]$ is lexicographically smaller than $F[j\dd j+r]$. Since $\ell=w+k-1 \implies w=\ell-r$, $j$ is an $(\ell -r,r+1)$-minimizer of $F$.
\end{proof}

\begin{example}
[Cont'd from Example~\ref{example:3}] 
The set $\mathcal{M}_{w,k}(S)$ of minimizers for $w=\ell-r=4$ and $k=r+1=2$ is  $\mathcal{M}_{4,2}(S)=\{1,4,5,6,7\}$. Since any reduced bd-anchor for $\ell=5$ and $r=1$ is a $(4,2)$-minimizer of $S$,  $\mathcal{A}_{5,1}(S)=\{4,5,6,7\}\subseteq \mathcal{M}_{4,2}(S)$.   
\end{example}

In particular, Fact~\ref{fct:simple} implies that, for any string $S$ of length $n$ and integers $w,k>0$, the set of $(w,k)$-minimizers in $S$ is a superset of the set of reduced bd-anchors of order $\ell=w+k-1$ with $r=k-1$.
Thus, we will use $(w,k)$-minimizers (computed by means of Lemma~\ref{lem:minimizers_con}) to compute reduced bd-anchors by setting $w=\ell-r$ and $k=r+1$. The following lemma is crucial for the efficiency of our algorithm.
The main idea is to use LCP queries to quickly identify suitable substrings (of the input rotations) that are then compared, in order to determine the lexicographically smallest rotation.
\begin{lemma}\label{lem:LCP}
For any string $F$ and two positions $i$ and $j$ on $F$, we can determine which of the rotations $i$ or $j$ is the smaller lexicographically in the time to answer three $\LCP_F$ queries and three letter comparisons in $F$.
\end{lemma}

\begin{figure}[t]
     \centering
      \includegraphics[width=4.5cm]{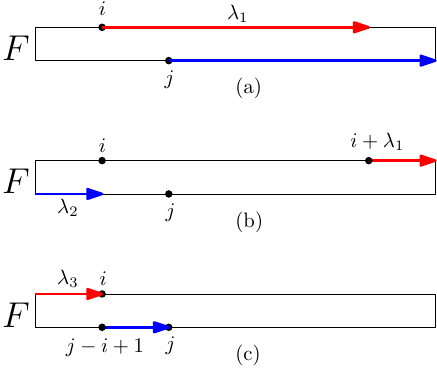}
      \caption{Illustration of Lemma~\ref{lem:LCP}.}
       \label{fig:LCP}
\end{figure}

\begin{proof}
Assume without loss of generality that $i<j$ (inspect Figure~\ref{fig:LCP}).
We first ask for the length $\lambda_1$ of the LCP of $F[i\dd |F|]$ and $F[j\dd |F|]$. If $\lambda_1<|F|-j+1$, then we simply compare $F[i\dd i+\lambda_1]$ and $F[j\dd j+\lambda_1]$ to find the answer. 
If $\lambda_1=|F|-j+1$ (see top part (a)), we ask for the length $\lambda_2$ of the LCP of $F[i+\lambda_1\dd |F|]$ and $F$.
If $\lambda_2<j-i$, then we simply compare $F[i+\lambda_1\dd i+\lambda_1+\lambda_2]$ and $F[1\dd 1+\lambda_2]$ to find the answer. 
If $\lambda_2=j-i$ (see middle part (b)), we ask for the length $\lambda_3$ of the LCP of $F$ and $F[j-i+1\dd |F|]$. If $\lambda_3<j-i$, then we simply compare $F[1\dd 1+\lambda_3]$ and $F[j-1\dd j-1+\lambda_3]$ to find the answer. Otherwise, rotation $i$ of $F$ is equal to rotation $j$ of $F$ (see bottom part (c)).
\end{proof}

We next use Lemma~\ref{lem:LCP} as a building block to obtain Lemma~\ref{lem:main}, the main lemma used by our algorithm.

\begin{lemma}\label{lem:main}
Let $D$ be a string of length $|D|\geq \ell$. Let $A$ be a set of $d$ positions on $D$ such that for every range $[i,i+\ell-1]\subseteq [1,|D|]$, $1\leq i \leq |D|-\ell+1$, there exists at least one element $j\in A: j\in [i,i+\ell-1]$. Given a data structure for answering $\LCP_D$ queries in $\cO(1)$ time, we can find the smallest lexicographic rotation $j$ in every length-$\ell$ fragment of $D$, such that $j\in A$ and $j$ is minimal in $\cO(d|D|)$ total time.
\end{lemma}

\begin{proof}
First we sort the elements of $A$ in increasing order using radix sort in $\cO(|D|)$ time. Then for every fragment $F=D[i\dd i+\ell -1]$ of length $\ell$ of $D$, $i\in[1,|D|-\ell+1]$, we consider pairs of elements from $A$ that are in $[i,i+\ell-1]$. We can consider these pairs from left to right because the elements of $A$ have been sorted. For any two elements, we perform an application of Lemma~\ref{lem:LCP}, which takes $\cO(1)$ time, and maintain the leftmost smallest rotation. Inspect Figure~\ref{fig:min}.

\begin{figure}[t]
     \centering
      \includegraphics[width=4.5cm]{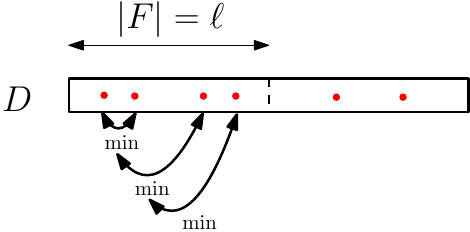}
      \caption{Illustration of three applications of Lemma~\ref{lem:LCP} on $F$. We mark the elements of $A$ in red circles.}
       \label{fig:min}
\end{figure}

Since we have at most $|D|$ length-$\ell$ fragments in $D$ and at most $d-1$ pairs that we consider per fragment, the total time to process all fragments is $\cO(d|D|)$.
\end{proof}

\begin{figure}[t]
     \centering
      \includegraphics[width=8.5cm]{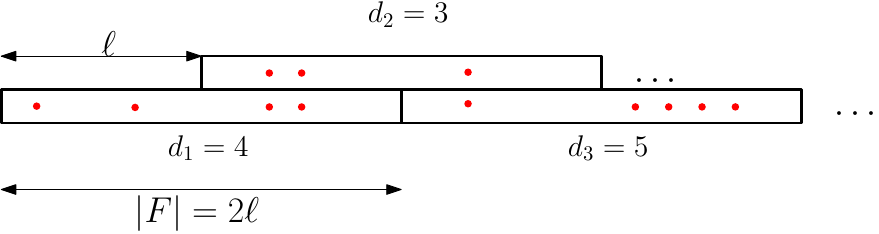}
      \caption{Illustration of the decomposition of $S$. We mark the $(w,k)$-minimizers in red circles.}
       \label{fig:strick}
\end{figure}

\begin{theorem}\label{the:main}
Given a string $S$ of length $n$, randomly generated by a memoryless source over an alphabet $\Sigma$ of size $\sigma \geq 2$ with identical letter probabilities, and an integer $\ell>0$, the expected number of reduced bd-anchors of order $\ell$ for $r=\lceil4 \log\ell /\log \sigma\rceil$ in $S$ is $\cO(n/\ell)$. Moreover, if $\Sigma$ is an integer alphabet of size $\sigma=n^{\cO(1)}$, $\mathcal{A}_{\ell,r}(S)$ can be computed in $\cO(n)$ time on average.
\end{theorem}

\begin{proof}
The first part is merely Lemma~\ref{lem:bd-anchors} from~\cite{TKDE2023}.

For the second part, we employ the so-called standard trick to conceptually decompose $S$ in $q$ fragments $F$ each of length $|F|=2\ell$ overlapping by $\ell$ positions (apart perhaps from the last one). Inspect Figure~\ref{fig:strick}.
We compute $\mathcal{M}_{w,k}(S)$ for $w=\ell - r$ and $k=r+1$. This can be done in $\cO(n)$ time by Lemma~\ref{lem:minimizers_con}~\cite{DBLP:conf/esa/LoukidesP21}.
We also construct an $\LCP_S$ data structure in $\cO(n)$ time~\cite{DBLP:conf/stoc/KempaK19}.
Let us denote by $d_i$ the number of  $(w,k)$-minimizers in the $i$th fragment $F_i$ of $S$ as per the above decomposition of $S$. We apply Lemma~\ref{lem:main} with $D=F_i, d=d_i$ and $A=\mathcal{M}_{w,k}(F_i)$. This takes $d_i|F_i|=\cO(d_i\ell)$ time.
Note that $q=\Theta(n/\ell)$. Then the total work of the algorithm, by Lemma~\ref{lem:main}, is bounded \emph{on average} by:

$n + \ell d_1+ \ell d_2+ \ldots + \ell d_q = n + \ell (d_1+ d_2+ \ldots + d_q)$

$\leq n + 2\ell \cdot \mathcal{M}_{w,k}(S)=\cO(n)$.

The last inequality holds by the fact that the fragments overlap by $\ell$ positions, which means that $d_i$ will generally be considered twice (see Figure~\ref{fig:strick}), and by Lemma~\ref{lem:minimizer}, which gives an expected asymptotic upper bound on the number of $(w,k)$-minimizers in $S$. The algorithm is correct by Fact~\ref{fct:simple}.
\end{proof}

In the worst case (e.g., $S=\texttt{a}^n$), $\mathcal{M}_{w,k}(S)=\Theta(n)$, and so the algorithm of Theorem~\ref{the:main} takes $\Theta(n\ell)$ time. 

Let us now highlight the differences for computing the set of randomized reduced bd-anchors. The following lemma follows directly from Lemma~\ref{lem:minimizers_con} and the fact that the KR fingerprints for all fixed-length fragments of $S$ can be computed in $\cO(n)$ time~\cite{DBLP:journals/ibmrd/KarpR87}.

\begin{lemma}\label{lem:random_minimizers_con}
For any string $S$ of length $n$ over an integer alphabet and integers $w,k>0$, the set $\mathcal{M}^{\text{ran}}_{w,k}(S)$ can be computed in $\cO(n)$ time.
\end{lemma}

Theorem~\ref{the:main2} follows directly from Theorem~\ref{the:main} by making the following two changes in the proof of Theorem~\ref{the:main}: 

\begin{itemize}
    \item Replace Lemma~\ref{lem:bd-anchors} by Lemma~\ref{lem:rrbd-anchors}.
    \item Replace Lemma~\ref{lem:minimizers_con} by Lemma~\ref{lem:random_minimizers_con}.
\end{itemize}

\begin{theorem}\label{the:main2}Given a string $S$ of length $n$, randomly generated by a memoryless source over an alphabet $\Sigma$ of size $\sigma \geq 2$ with identical letter probabilities, and an integer $\ell>0$, the expected number of randomized reduced bd-anchors of order $\ell$ for $r=\lceil4 \log\ell /\log \sigma\rceil$ in $S$ is $\cO(n/\ell)$. Moreover, $\mathcal{A}^{\text{ran}}_{\ell,r}(S)$ can be computed in $\cO(n)$ time on average.\end{theorem}

\subsection{Index construction in small space}\label{sec:small}

In this section, we show the index construction in
small space using near-optimal work.
Similar to the previous section,
the same technique (with very small differences) can be used to construct the index for reduced bd-anchors and for randomized reduced bd-anchors of $S$. We will highlight these differences appropriately.

Note that a \emph{straightforward} implementation of the index $\mathcal{I}_{\ell,r}(S)$ (see also~\cite{DBLP:conf/esa/LoukidesP21}) requires $\Theta(n)$ space in any case.
In particular, this is because the \textsf{SA} and \textsf{LCP} array of $S$ (and $\overleftarrow{S}$) are required for: \textbf{(i)} the implementation of Lemma~\ref{lem:minimizers_con}~\cite{DBLP:conf/esa/LoukidesP21} (reduced bd-anchors); and \textbf{(ii)} the construction of $\mathcal{T}^R_{\ell,r}(S)$ and $\mathcal{T}^L_{\ell,r}(S)$~\cite{DBLP:conf/esa/LoukidesP21}.

However, the size of $\mathcal{I}_{\ell,r}(S)$ can be asymptotically smaller than $\Theta(n)$.
In fact, Theorem~\ref{the:index} tells us that the size of $\mathcal{I}_{\ell,r}(S)$ is $\cO(|\mathcal{A}_{\ell,r}(S)|)$, which is expected to be $\cO(n/\ell)$. 
Note that $|\mathcal{A}_{\ell,r}(S)|=\Omega(n/\ell)$ in any case~\cite{DBLP:conf/esa/LoukidesP21}, so we cannot hope for less memory. The goal of this section is to show that we can construct $\mathcal{I}_{\ell,r}(S)$ \emph{efficiently} using only $\cO(\ell+|\mathcal{A}_{\ell,r}(S)|)$ space (and external memory). (Note that in practical applications $\ell$ never exceeds $n/\ell$ and so $\ell$ is negligible.)
We will show two constructions: one utilizing internal and external memory; and another one utilizing only internal memory.

In the standard \emph{external memory} (EM) model~\cite{DBLP:journals/fttcs/Vitter06}, with internal memory (RAM) size $M$ and disk block size $B$, the standard I/O complexities are: $\textsf{scan}(n) = n/B$,
which is the complexity of scanning $n$ elements sequentially; and $\textsf{sort}(n) =
(n/B) \log_{M/B} (n/B)$, which is the complexity of sorting $n$ elements~\cite{DBLP:journals/fttcs/Vitter06}. 
\emph{Semi-EM} model is a relaxation of the EM model, in which we are allowed to store some of the data in internal memory~\cite{DBLP:journals/algorithmica/AbelloBW02}. We proceed in four steps: In Step 1, we compute the set $\mathcal{A}_{\ell,r}(S)$ using $\cO(\ell)$ extra space; in Step 2, we compute the \textsf{SA} and \textsf{LCP} array of $S$ in the EM model using existing algorithms; in Step 3, we construct $\mathcal{T}^R_{\ell,r}(S)$ and $\mathcal{T}^L_{\ell,r}(S)$ in the semi-EM model (by having $\mathcal{A}_{\ell,r}(S)$ in internal memory); in Step 4, we construct the 2D range reporting data structure in internal memory using existing algorithms. The only difference of the second construction is that it skips Step 2, and directly constructs $\mathcal{T}^R_{\ell,r}(S)$ and $\mathcal{T}^L_{\ell,r}(S)$ utilizing only internal memory.

We next describe in detail how every step is implemented using near-optimal work.

\paragraph{Step 1} We make use of $\cO(\ell)$ extra space. Specifically, we use Theorem~\ref{the:main} to construct $\mathcal{A}_{\ell,r}(S)$ via considering fragments of $S$ of length-$2\ell$ (apart perhaps from the last one), overlapping by $\ell$ positions, without significantly increasing the time. If the size of the alphabet of a fragment $F$ is not polynomial in $\ell$ (i.e., $F$ consists of large integers), we use $\ctO(\ell)$ time instead of $\cO(\ell)$ by first sorting the letters of $F$, and then assigning each letter to a rank in $\{1,\ldots,2\ell\}$. This clearly does not affect the lexicographic rank of rotations. Thus the total time to construct $\mathcal{A}_{\ell,r}(S)$ is $\ctO(n)$: there are $\cO(n/\ell)$ fragments and each one is processed in $\ctO(\ell)$ time. We finally implement $\mathcal{A}_{\ell,r}(S)$ as a perfect hash table~\cite{DBLP:journals/jacm/FredmanKS84}, denoted by $\mathcal{H}_{\ell,r}(S)$, which we keep in RAM. This is done in $\cO(|\mathcal{A}_{\ell,r}(S)|)$ time; and so the total time is $\ctO(n)$. The only difference for randomized reduced bd-anchors is that we use Theorem~\ref{the:main2} instead to construct $\mathcal{A}^{\text{ran}}_{\ell,r}(S)$. (The alphabet size plays no role in Theorem~\ref{the:main2}.)

\paragraph{Step 2} We compute the \textsf{SA} and \textsf{LCP} array of $S$ in external memory using $M$ words of internal memory and disk block size $B$. To this end we can use existing algorithms, which compute the two arrays  simultaneously~\cite{DBLP:journals/jacm/KarkkainenSB06,DBLP:journals/jea/Bingmann0O16}; these algorithms are optimal with respect to internal work $\cO(n \log_{M/B} (n/B))$ and I/O complexity $\cO((n/B) \log_{M/B} (n/B))$---see also~\cite{DBLP:conf/cpm/KarkkainenKP15a,DBLP:journals/jea/KarkkainenK16,DBLP:conf/spire/KarkkainenK16,DBLP:conf/alenex/KarkkainenKPZ17,DBLP:journals/jea/KarkkainenK19}.
We also compute the \textsf{SA} and \textsf{LCP} array of $\overleftarrow{S}$, the reverse of $S$, analogously.
Clearly this step is the same for randomized reduced bd-anchors. For the internal memory construction, this step is omitted.

\paragraph{Step 3} We construct the two compacted tries $\mathcal{T}^R_{\ell,r}(S)$ and $\mathcal{T}^L_{\ell,r}(S)$ with the aid of four arrays, each of size $|\mathcal{A}_{\ell,r}(S)|$: \textsf{RSA}; \textsf{RLCP}; \textsf{LSA}; and \textsf{LLCP}.
Specifically, \textsf{RSA} (Right \textsf{SA}) stores a permutation of $\mathcal{A}_{\ell,r}(S)$ such that $\textsf{RSA}[i]$ is the starting position of the $i$th lexicographically smallest suffix of $S$ with $\textsf{RSA}[i]\in \mathcal{A}_{\ell,r}(S)$. $\textsf{RLCP}[i]$ array (Right \textsf{LCP} array) stores the length of the LCP of $\textsf{RSA}[i-1]$ and $\textsf{RSA}[i]$.
\textsf{LSA} and \textsf{LLCP} array are defined analogously for $\overleftarrow{S}$. Let us show how we construct \textsf{RSA} and \textsf{RLCP}; the other case for \textsf{LSA} and \textsf{LLCP} array is symmetric. To compute these arrays we use the \textsf{SA} and the \textsf{LCP} array of $S$ constructed in Step 2.
We scan the \textsf{SA} and the \textsf{LCP} array of $S$ sequentially and sample them using the hash table $\mathcal{H}_{\ell,r}(S)$ constructed in Step 1.
Let us suppose that we want to sample the $k$th value after reading $\textsf{SA}[i]$ and $\textsf{LCP}[i]$.
This is possible using $\cO(1)$ words of memory.
If $\textsf{SA}[i]$ is in $\mathcal{H}_{\ell,r}(S)$, which we check in $\cO(1)$ time, we set $\textsf{RSA}[k]=\textsf{SA}[i]$. It is also well known that for any $i_1<j_2$ the length of the LCP between $S[\textsf{SA}[i_1]\dd n]$ and $S[\textsf{SA}[i_2]\dd n]$ is the minimum value lying in $\textsf{LCP}[i_1+1],\ldots,\textsf{LCP}[i_2]$.
Since we scan also the $\textsf{LCP}$ array simultaneously, we maintain the value we need to store in $\textsf{RLCP}[k]$. Finally we increment $k$ and $i$ by one. 
Scanning \textsf{RSA} and \textsf{RLCP} takes $\cO(n/B)$ I/Os.
Using \textsf{RSA} and \textsf{RLCP}
we can construct $\mathcal{T}^R_{\ell,r}(S)$ in $\cO(|\mathcal{A}_{\ell,r}(S)|)$ time using a folklore algorithm (cf.~\cite{DBLP:conf/cpm/KasaiLAAP01}). 
We construct $\mathcal{T}_{\ell,r}^L(S)$ analogously from \textsf{LSA} and \textsf{LLCP}. 
For the internal-memory construction,
Step 3 can be implemented in $\ctO(n)$ time using $\cO(|\mathcal{A}_{\ell,r}(S)|)$ space assuming read-only random access to $S$.
This can be achieved by using any algorithm
for \emph{sparse suffix sorting}~\cite{DBLP:conf/stacs/IKK14,DBLP:conf/soda/GawrychowskiK17,DBLP:conf/soda/BirenzwigeGP20,DBLP:conf/latin/AyadLPV24,DBLP:conf/cpm/KosolobovS24}. 
Clearly for randomized reduced bd-anchors, we use
$\mathcal{A}^{\text{ran}}_{\ell,r}(S)$ instead of
$\mathcal{A}_{\ell,r}(S)$.

\paragraph{Step 4} By using \textsf{RSA} and \textsf{LSA}, we construct the 2D range reporting data structure in $\ctO(|\mathcal{A}_{\ell,r}(S)|)$ time using $\cO(|\mathcal{A}_{\ell,r}(S)|)$ space~\cite{DBLP:conf/compgeom/ChanLP11,DBLP:conf/latin/MakinenN06,DBLP:conf/soda/BelazzouguiP16,DBLP:conf/esa/Gao0N20}.

This completes the construction: we have shown how to construct $\mathcal{I}_{\ell,r}(S)$ in near-optimal work using only $\cO(\ell+|\mathcal{A}_{\ell,r}(S)|)$ space (and external memory). This is good because the size of $\mathcal{I}_{\ell,r}(S)$ is $\cO(|\mathcal{A}_{\ell,r}(S)|)$ (Theorem~\ref{the:index}) and the $\cO(\ell)$ factor is negligible in practice.

\section{Index with worst-case guarantees}\label{sec:worst-case}
The index based on bd-anchors has
$\cO(|\mathcal{A}_{\ell,r}(S)|)$ size,
and supports pattern matching queries in $\ctO(|P|+\occ)$ time.
Moreover, it can be constructed in $\ctO(n)$ time 
using $\cO(\ell+|\mathcal{A}_{\ell,r}(S)|)$ space. 
Unfortunately, while in many real-world datasets we have $|\mathcal{A}_{\ell,r}(S)|=\Theta(n/\ell)$, there are worst cases with $|\mathcal{A}_{\ell,r}(S)|=\Theta(n)$; e.g., for $S=a_1a_2\ldots a_n$, where all $n$ letters are different.

In this section, we combine some ideas presented in~\cite{DBLP:conf/cpm/BathieCS24} with the bd-anchors index to prove the following.

\begin{theorem}\label{the:worst-case}
For any string $S$ of length $n$ over an integer alphabet of size $n^{\cO(1)}$ and any integer $\ell>0$, we can construct an index that occupies $\cO(n/\ell)$ extra space and reports all $\occ$ occurrences of any pattern $P$ of length $|P|\geq \ell$ in $S$ in $\ctO(|P|+\occ)$ time. The index can be constructed in $\ctO(n)$ time and $\cO(n/\ell)$ working space.
\end{theorem}

We define one of the most basic string notions.
An integer $p>0$ is a \emph{period} of a string $P$ if $P[i] = P[i + p]$ for all $i \in [1,|P|-p]$. The smallest period of $P$ is referred to as \emph{the period} of $P$ and is denoted by $\textsf{per}(P)$. 

\begin{example}
For $P=\texttt{abaaabaaabaaaba}$, $\textsf{per}(P)=4$; $p=8$ is also a period of $P$ but it is not the smallest.    
\end{example}

We are now in a position to define the following powerful notion of locally consistent anchors.

\begin{definition}[$\tau$-partitioning set~\cite{DBLP:conf/cpm/KosolobovS24}]\label{def:part-set}
For any string $S$ of length $n$ and any integer $\tau\in[4,n/2]$, a set of positions
$\mathcal{P} \subseteq [1,n]$ is called a \emph{$\tau$-partitioning set} if it satisfies the following properties:
\begin{enumerate}
    \item if $S[i-\tau\dd i+\tau ] = S[j-\tau\dd j+\tau]$, for $i,j \in [\tau, n-\tau)$, then $i \in \mathcal{P}$ if and only if $j \in \mathcal{P}$; \label{prop:a}
    \item if $S[i \dd i+\ell] = S[j \dd j+\ell]$, for $i, j \in \mathcal{P}$ and some $\ell \geq 0$, then, for each $d \in [0, \ell-\tau)$, $i + d \in \mathcal{P}$ if and only if $j + d \in \mathcal{P}$;
    \item if $i, j \in [1, n]$ with $j - i > \tau$ and $(i,j)\cap \mathcal{P}=\emptyset$, then $S[i\dd j]$ has period at most $\tau/4$.
\end{enumerate}
\end{definition}

Let us explain these properties. Property 1 tells us that for any two equal sufficiently long fragments of $S$, we either pick no anchor or we pick an anchor having the same relative position in both (local consistency). Property 2 tells us that for any two equal sufficiently long fragments of $S$, the anchors we pick from both are in sync (forward synchronization). Property 3 tells us that for a sufficiently long fragment of $S$, if no anchor is picked, then this fragment is highly periodic (approximately uniform sampling).

\begin{theorem}[\cite{DBLP:conf/cpm/KosolobovS24}]\label{the:part-set}
For any string $S$ of length $n$ over an integer alphabet of size $n^{\cO(1)}$ and any integer $\tau \in [4,\cO(n/ \log^2 n)]$, we can construct a $\tau$-partitioning set $\mathcal{P}$ of size $\cO(n/\tau)$.
The set $\mathcal{P}$ can be constructed in $\ctO(n)$ time 
using $\cO(n/\tau)$ working space.
The set $\mathcal{P}$ additionally satisfies the property that if a fragment $S[i \dd j]$ has period at most $\tau/4$, then $\mathcal{P} \cap [i + \tau , j - \tau ] = \emptyset$.
\end{theorem}

We now define the notion of $\tau$-runs: sufficiently long maximal fragments with the same periodic structure.

\begin{definition}[$\tau$-run~\cite{DBLP:conf/cpm/BathieCS24}]
A fragment $F$ of a string $S$ is a \emph{$\tau$-run} if and only if $|F|> 3\tau$,
$\textsf{per}(F) \leq \tau/4$, and $F$ cannot be extended in either direction without its period changing. The \emph{Lyndon root} of a $\tau$-run $F$ is the lexicographically smallest rotation of $F[1 \dd \textsf{per}(F)]$.
String $F$ can be expressed by $L$ and three non-negative integers $x_F,e_F,y_F$:   
$$F=L[|L|-x_F+1\dd |L|]\cdot L^{e_F}\cdot L[1\dd y_F].$$
\end{definition}

\begin{example}
For $F=\texttt{abaaabaaabaaaba}$ with $\textsf{per}(F)=4$, the Lyndon root is $L=\texttt{aaab}$. 
Specifically, we have $x_F=2$, $e_F=3$, and $y_F=1$:
$$F=\texttt{ab}\cdot \texttt{aaab}^{3}\cdot \texttt{a}.$$
\end{example}

We also state a few combinatorial and algorithmic results on $\tau$-runs that are crucial in our construction.

\begin{lemma}[\cite{DBLP:journals/corr/abs-2105-03106}]\label{lem:runs}
Two $\tau$-runs can overlap by at most $\tau/2$ positions.
The number of $\tau$-runs in a string of length $n$ is $\cO(n/\tau)$.    
\end{lemma}

\begin{lemma}[\cite{DBLP:conf/cpm/BathieCS24}]\label{the:runs}
For any string $S$ of length $n$ over an integer alphabet of size $n^{\cO(1)}$ and any integer $\tau \in [4,\cO(n/ \log^2 n)]$, 
all $\tau$-runs in $S$ can be computed
and grouped by Lyndon root in $\ctO(n)$ time using $\cO(n/\tau)$ space. Within the same complexities, we can compute the first occurrence of the Lyndon root in each $\tau$-run.    
\end{lemma}

\paragraph{Data structure} Let $\ell$ be an integer in $[20,\lfloor n/\log^2 n \rfloor]$ and let $\tau = \lfloor\ell/5\rfloor$. 
We use Theorem~\ref{the:part-set} and Lemma~\ref{the:runs} with parameter $\tau$ to compute a
partitioning set $\mathcal{P}$ of size $\cO(n/\tau)$ and all $\tau$-runs in $S$.
For every $\tau$-run $T$ of group $L$, we have the first occurrence of its Lyndon root.
We can thus compute the non-negative integers $x_T,y_T<|L|$ and $e_T$, such that $T=L[|L|-x_T+1\dd |L|]\cdot L^{e_T}\cdot L[1\dd y_T]$, in constant time.
We maintain the lists of pairs $(e_T,x_T)$ and $(e_T,y_T)$ for all $T$, sorted, for every $L$ separately. 
This can be done in $\ctO(n)$ time and $\cO(n/\tau)$ space using merge sort.
To access the lists of $L$, we index the collection of lists by the KR fingerprint of $L$. 
We also maintain a priority search tree~\cite{DBLP:journals/siamcomp/McCreight85} over points $(x_T,y_T)$, for every $(L,e_T)$ separately, on the grid $[1,|L|]^2$. This can be done in $\ctO(n/\tau)$ time and $\cO(n/\tau)$ space. 
To access a specific priority search tree for $L$, we further index the collection of trees by $e_T$.
For every $i \in \mathcal{P}$ we also store the KR fingerprint  (see Section~\ref{sec:rr-bd-anchors}) of $S[i - \tau \dd i + \tau]$ in a hash table $H$. This can be done in $\cO(n)$ time using $\cO(|\mathcal{P}|)$ space. Let $\mathcal{L}$ be a set that contains the starting and ending position of each $\tau$-run.
We construct $\mathcal{A}:=\mathcal{P} \cup \mathcal{L}$ and construct the index of Section~\ref{sec:index} over string $S$ and $\mathcal{A}$.
By Lemma~\ref{lem:runs} and Theorem~\ref{the:part-set} the size of $\mathcal{A}$ is $\cO(n/\tau)$ in the worst case. By Theorem~\ref{the:part-set} and Theorem \ref{the:runs}, $\mathcal{A}$ is constructed in $\ctO(n)$ time. By Theorem~\ref{the:index}, the index over string $S$ and $\mathcal{A}$ is constructed in $\ctO(n)$ time and occupies $\cO(|\mathcal{A}|)=\cO(n/\tau)$ extra space. The working space for the two compacted tries is also $\cO(|\mathcal{A}|)=\cO(n/\tau)$ by employing the following result on sparse suffix sorting.

\begin{theorem}[\cite{DBLP:conf/soda/BirenzwigeGP20,DBLP:conf/cpm/KosolobovS24}]\label{the:SST}
For any string $S$ of length $n$ over an integer alphabet of size $n^{\cO(1)}$, we can construct the compacted trie of $b$ arbitrary suffixes of $S$ in $\ctO(n)$ time using $\cO(b)$ space.
\end{theorem}

The working space for the accompanying 2D range reporting data structure is also $\cO(|\mathcal{A}|)=\cO(n/\tau)$ by using the data structure of M{\"{a}}kinen and Navarro~\cite{DBLP:conf/latin/MakinenN06}.

This completes the construction of the index.

\paragraph{Querying} We are given a query $P$, $|P|\geq \ell$. 
Let us give an overview of the querying part:
if $P$ occurs in $S$, then it should be because it contains a sufficiently long aperiodic fragment, for which we have an anchor stored in $\mathcal{P}\subseteq \mathcal{A}$, or because it has a sufficiently long periodic fragment, which occurs also in $S$. The latter case is split in the following two subcases: (i) $P$ has a periodic fragment shorter than $P$, in which case, we use one of its endpoints as an anchor stored in $\mathcal{L}\subseteq \mathcal{A}$; or the whole of $P$ is periodic, in which case, we have no anchor to use, but we can use the Lyndon root of $P$ to find its occurrences in $S$. We thus start by computing $\textsf{per}(P)$ in $\cO(|P|)$ time~\cite{DBLP:books/daglib/0020103} and proceed as follows:

\begin{description}
    \item[\textbf{Aperiodic case}:] If $\textsf{per}(P)> \tau/4$, we compute the KR fingerprint of every fragment $P[j - \tau \dd j + \tau]$ of $P$. If a KR fingerprint is found in $H$, we spell $P[j\dd |P|]$ and $\overleftarrow{P[1\dd j]}$ in the two compacted tries, and obtain the occurrences of $P$ in $S$ by using a 2D range reporting query just as in Section~\ref{sec:index}. The total query time is $\ctO(|P|+\occ)$. 
    \item[\textbf{Fully-periodic case}:] In this case, $\textsf{per}(P)\leq \tau/4$.
    We start by computing the Lyndon root $L$ of $P$ and the non-negative integers $x_P,y_P<|L|$ and $e_P$ such that $P=L[|L|-x_P+1\dd |L|]\cdot L^{e_P}\cdot L[1\dd y_P]$ in $\cO(|P|)$ time~\cite{DBLP:conf/cpm/BathieCS24}. Let $T$ be a $\tau$-run in $S$ with the same $L$.
    Further let $x_T,y_T<|L|$ and $e_T$ such that $T=L[|L|-x_T+1\dd |L|]\cdot L^{e_T}\cdot L[1\dd y_T]$. $P$ occurs in $T$ if and only if at least one of the following conditions is met~\cite{DBLP:conf/cpm/BathieCS24}: (1) $e_P = e_T$, $x_P \leq x_T$, and $y_P \leq y_T$; or (2) $e_P + 1 = e_T$ and $x_P \leq x_T$; or (3) $e_P + 1 = e_T$ and $y_P \leq y_T$; or (4) $e_P + 2 \leq e_T$. We first locate the sorted lists for the $\tau$-runs $T$ with Lyndon root $L$ in $S$,
    in $\cO(|L|)$ time, by computing the KR fingerprint  of $L$, as well as the corresponding priority search trees.
    Each valid $T$ is located in $\ctO(1)$ time with binary search on the stored lists or with one of the priority search trees. For instance, for checking the first condition, we locate the priority search tree for $L$ and $e_P$ and search for
    the points included in the axis-aligned rectangle $[x_P,\infty)\times [y_P,\infty)$.
    All $s$ such points are found in $\cO(\log n + s)$ time~\cite{DBLP:journals/siamcomp/McCreight85}.
    Let $(x_T,y_T)$ be a retrieved point corresponding to $\tau$-run $T$ of $S$.
    This ensures that $e_P = e_T$, $x_P \leq x_T$, and $y_P \leq y_T$.
    The other three conditions are trivially checked with binary search.
    Since the length of $T$ is $x_T+|L| \cdot e_T+y_T$; we know the length of $P$; and we also know the first occurrence of $L$ in both $T$ and $P$, we can easily report all $k$ occurrences of $P$ in $T$ in $\cO(k)$ extra time. The total query time is $\ctO(|P|+\occ)$.
    \item[\textbf{Quasi-periodic case}:] It could be the case that $P$ is not highly-periodic itself (i.e., $\textsf{per}(P)> \tau/4$) but contains a long highly-periodic fragment that prevents us from finding a fragment $P[j - \tau \dd j + \tau]$ with a KR fingerprint in $H$ (see Property 3, Definition~\ref{def:part-set}). Inspect Figure~\ref{fig:quasi-periodic}. By the properties of $\mathcal{P}$, we know that if $P$ occurs in $S$ then $\textsf{per}(P[\tau\dd |P|-\tau+1])\leq \tau/4$; note that, otherwise, we would have found at least one fragment of length $2\tau+1$ whose KR fingerprint is in $H$.
    \begin{figure}[t]
        \centering
        \includegraphics[width=0.45\linewidth]{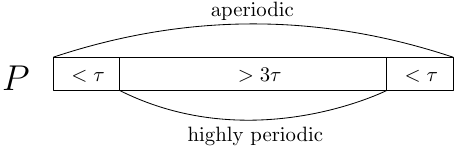}
        \caption{The quasi-periodic case: in this setting, the long highly-periodic fragment is neither a prefix nor a suffix, and so we can use any of its endpoints as an anchor.}
        \label{fig:quasi-periodic}
    \end{figure}
    \begin{example}
Let $\ell=20$, $\tau=4$, and $|P|=\ell$. Further assume that $P$ occurs in $S$ but $P[i-\tau \dd i+\tau]$, for all $i\in[5,16]$ was not found in $H$. Property 3 of set $\mathcal{P}$ tells us that if $P$ occurs in $S$ then $P[4 \dd 17]=P[\tau\dd |P|-\tau+1]$ has period at most $\tau/4$.  
\end{example}
In $\cO(|P|)$ time, we check if $\textsf{per}(P[\tau\dd |P|-\tau+1])\leq \tau/4$~\cite{DBLP:books/daglib/0020103}, and if so, we extend periodicity on both sides as much as possible, hence computing a periodic fragment $P[i\dd j]$, $i\leq \tau$ and $j\geq |P|-\tau+1$. If $i>1$, we construct $P[i\dd |P|]$ and $\overleftarrow{P[1\dd i]}$; otherwise, if $j<|P|$, we construct $P[j\dd |P|]$ and $\overleftarrow{P[1\dd j]}$. 
Note that it cannot be that both $i=1$ \emph{and} $j=|P|$, because then we would be in the previous case (fully-periodic case). We then spell the two sub-patterns in the two compacted tries, and obtain the occurrences of $P$ in $S$ by using a 2D range reporting query just as in Section~\ref{sec:index}. The total query time is $\ctO(|P|+\occ)$.
\end{description}

Note that if any of the above three cases is successful in finding occurrences of $P$ in $S$ we stop, because we would in fact have found \emph{all} the occurrences of $P$ in $S$. On the other hand, if none of the above cases is successful, then it means that $P$ does not occur in $S$.

Finally, let us handle the edge cases. Recall that, in the above, we have assumed that $\ell$ is an integer in $[20,\lfloor n/\log^2 n \rfloor]$. If $\ell<20=\cO(1)$, we simply construct the suffix tree of $S$ (as our only index) in $\cO(n)$ time~\cite{DBLP:conf/focs/Farach97}. The space occupied by the suffix tree is $\cO(n)$, which for any $\ell=\cO(1)$, is in $\cO(n/\ell)$. If $\ell>\lfloor n/\log^2 n \rfloor$ then we do not construct any text index. Upon a query $P$, with $|P|>\lfloor n/\log^2 n \rfloor$, we use a real-time constant-space string matching algorithm~\cite{DBLP:journals/tcs/BreslauerGM13} to report all occurrences of $P$ in $S$ in $\cO(n)$ time, in which case we have that $\cO(n)=\ctO(|P|+\occ)$. We have arrived at Theorem~\ref{the:worst-case}.

Our construction is randomized (correct with high probability) due to the use of KR fingerprints.

\section{Implementations}\label{sec:imp}
In this section, we provide the full details of our implementations, which have all been written in \texttt{C++}. 
In addition to the classic text indexes referred to in Section~\ref{sec:intro}, we have considered the r-index~\cite{DBLP:journals/jacm/GagieNP20}, a text index, which performs specifically well for highly repetitive text collections. We did not consider: \textbf{(i)} the suffix tree, as it is not competitive at all with respect to space; or \textbf{(ii)} sampling the suffix array with minimizers~\cite{DBLP:journals/spe/GrabowskiR17}, as their number is generally greater than (reduced) bd-anchors; see the results of~\cite{DBLP:conf/esa/LoukidesP21,TKDE2023}. As we show later in experiments (Section~\ref{sec:exp}), the set of randomized reduced bd-anchors can be computed \emph{faster} and using \emph{less space} than the set of reduced bd-anchors. Moreover, we show that the former set is \emph{always smaller} than the latter. We have thus used randomized reduced bd-anchors instead of reduced bd-anchors to construct our indexes. We have also not implemented our index with worst-case guarantees (Section~\ref{sec:worst-case}) as it relies on partitioning sets~\cite{DBLP:conf/cpm/KosolobovS24}, an extremely intricate sampling scheme, which is highly unlikely to be efficient in practice,\footnote{Personal communication with Dmitry Kosolobov.} as well as on the involved querying scheme consisting of many cases, which is also highly unlikely to be efficient in practice.

Since our focus is on algorithmic ideas (not on low-level code optimization) and to ensure fairness across different implementations, we have re-used the same algorithm or code whenever it was possible. For instance, although many other engineered implementations for classic text indexes exist, e.g.~\cite{DBLP:journals/algorithmica/GogKKPP19} for FM-index or~\cite{DBLP:conf/esa/BertramE021} for suffix array, we have based our implementations on \textsf{sdsl-lite}~\cite{sdsl} as much as possible for fairness:
\begin{description}
\item[\textsf{SA}:] The construct SA we used the \texttt{divsufsort} class, written by Yuta Mori, as included in \textsf{sdsl-lite}. To report all  occurrences of a pattern, we implemented the algorithm of Manber and Myers~\cite{DBLP:journals/siamcomp/ManberM93}
that uses as well the LCP array~\cite{DBLP:conf/cpm/KasaiLAAP01} augmented with a succinct RMQ data structure~\cite{DBLP:journals/siamcomp/FischerH11} (\texttt{rmq\_succinct\_sct} class).  
\item[\textsf{CSA}:] The CSA was constructed using the \texttt{csa\_sada} class of \textsf{sdsl-lite}. To report all occurrences of a pattern, we used the \textsf{SA} method. 
\item[\textsf{CST}:] The CST was constructed using the \texttt{cst\_sct3} class~\cite{DBLP:conf/spire/OhlebuschFG10} of \textsf{sdsl-lite}. To report all occurrences of a pattern, we traverse the tree (see Section~\ref{sec:prel}) by using its supported functionality. 
\item[\textsf{FM-index}:] The FM-index was implemented using the \texttt{csa\_wt} class of \textsf{sdsl-lite}. To report all occurrences of a pattern, the supported \texttt{count} and \texttt{locate} methods of \textsf{sdsl-lite} were used. 
\item[\textsf{r-index}:] The \textsf{r-index} is not part of \textsf{sdsl-lite} and so we had to resort to the open-source implementation provided by Gagie et al.~in~\cite{DBLP:journals/jacm/GagieNP20}. The binary \texttt{ri-build} was used to build the index (using the \texttt{divsufsort} class of the \textsf{sdsl-lite}) and then \texttt{ri-locate} was used to locate all occurrences of the patterns.
\item[\textsf{rBDA-compute}:] We implemented the  algorithm for computing the set of reduced bd-anchors as described in Section~\ref{sec:bd-anchors}. To improve construction space, we implemented the algorithm in fragments as described in Step 1 of the space-efficient algorithm presented in Section~\ref{sec:small}. The implementation takes the fragment length $b$ as input. We call these fragments \emph{blocks} henceforth.
\item[\textsf{rrBDA-compute}:] Similar to the above, we implemented the algorithm for computing the set of randomized reduced bd-anchors as described in Section~\ref{sec:bd-anchors}. The implementation takes the block length $b$ as input. As we will show in Section~\ref{sec:rbd-anchors_results}, randomized reduced bd-anchors performed consistently better than reduced bd-anchors. Therefore, we have based our index construction implementations, presented right below, on the notion of randomized reduced bd-anchors.

\item[\textsf{rrBDA-index I (ext)}:] The construction was implemented using Steps 1 to 3 as described in Section~\ref{sec:small}. For Step 2, we used the \textsf{pSAscan} algorithm of K{\"{a}}rkk{\"{a}}inen et al.~\cite{DBLP:conf/cpm/KarkkainenKP15a} to construct the \textsf{SA} in EM and the \textsf{EM-SparsePhi} algorithm of K{\"{a}}rkk{\"{a}}inen and Kempa~\cite{DBLP:journals/jea/KarkkainenK19} to compute the \textsf{LCP} array in EM. For Step 3, instead of constructing the compacted tries, we used the arrays \textsf{LSA}, \textsf{LLCP}, \textsf{RSA}, and \textsf{RLCP} directly. The \textsf{LLCP} and \textsf{RLCP} arrays were augmented each with an RMQ data structure (similar to \textsf{SA}). For Step 4, we implemented the 2D range reporting data structure of M{\"{a}}kinen and Navarro~\cite{DBLP:conf/latin/MakinenN06}. To report all occurrences of a pattern, we implemented the algorithm of Loukides and Pissis~\cite{DBLP:conf/esa/LoukidesP21} (see Section~\ref{sec:index}).

\item[\textsf{rrBDA-index I (int)}:] For the internal-memory construction,
we implemented Step 3 using the \textsf{PA} algorithm for sparse suffix sorting of Ayad et al.~from~\cite{DBLP:conf/latin/AyadLPV24}. Sparse suffix sorting is not included in \textsf{sdsl-lite}, and so we had to resort to \textsf{PA}, being the state of the art.

\item[\textsf{rrBDA-index II (ext)}:] The construction was implemented using Steps 1 to 3 as in \textsf{rrBDA-index I (ext)}. Unlike \textsf{rrBDA-index I (ext)}, however, no 2D range reporting data structure was used. We instead used the bidirectional search algorithm presented by Loukides and Pissis in~\cite{DBLP:conf/esa/LoukidesP21}. This algorithm reports all occurrences of a pattern $P$, by first searching for either $P[j\dd |P|]$ or $\overleftarrow{P[1\dd j]}$ using the four arrays, where $j$ is the reduced bd-anchor of $P[1\dd \ell]$; and then using letter comparisons to verify the remaining part of candidate occurrences. This index was the most efficient one in practice in~\cite{DBLP:conf/esa/LoukidesP21}; however, the query time is not bounded in the worst case.

\item[\textsf{rrBDA-index II (int)}:] Similar to \textsf{rrBDA-index I}, for the internal-memory construction, we implemented Step 3 using the \textsf{PA} algorithm of Ayad et al.~from~\cite{DBLP:conf/latin/AyadLPV24}.

\end{description}

\section{Experiments}\label{sec:exp}

\subsection{Experimental setup} We downloaded five  datasets (texts) from the \textsf{Pizza\&Chili} corpus~\cite{pizzachili} -- see Table~\ref{tab:datasets} for their characteristics. We generated patterns of length 32, 64, 128, 256, 512, and 1024 for all datasets with 500,000 distinct patterns created per length. The patterns were generated by selecting occurrences uniformly at random from the datasets. We also downloaded the \emph{Homo sapiens} genome (version GRCh38.p14), which we denote by HUMAN. Its length $n$ is 3,136,819,257 and the alphabet size $\sigma$ is 15. We generated patterns of length 64, 256, 1,024, 4,096, and 16,384 for HUMAN with 500,000 distinct patterns created per length. Similar to the \textsf{Pizza\&Chili} datasets, these patterns were generated by selecting occurrences uniformly at random from HUMAN. 

The experiments ran on an Intel Xeon Gold 648 at 2.5GHz with 192GB RAM. All programs were compiled with \texttt{g++} version \texttt{10.2.0} at the \texttt{-O3} optimization level. Our code and datasets are freely available at \url{https://github.com/lorrainea/rrBDA-index}. 

\begin{table}[t]
    \centering\scalebox{1}{
    \begin{tabular}{|c|c|c|} \hline
         Dataset & Length $n$ & Alphabet size $\sigma$ \\ \hline \hline
         DNA & 209,714,087 & 15  \\ \hline
         PROTEINS & 209,006,085 & 24 \\ \hline
         XML & 204,198,133 & 95 \\ \hline
         SOURCES & 166,552,125 & 225 \\ \hline
         ENGLISH & 205,627,746 & 222 \\ \hline
         HUMAN & 3,136,819,257 & 15 \\ \hline
    \end{tabular}}
    \caption{Datasets characteristics.}
    \label{tab:datasets}
\end{table}

\paragraph{Parameters}
For \textsf{rBDA-compute}, \textsf{rrBDA-compute}, \textsf{rrBDA-index I} (\textsf{int} and \textsf{ext}) and \textsf{rrBDA-index II} (\textsf{int} and \textsf{ext}), we set $\ell$ to the pattern length, $b$ to $25$K, unless stated otherwise, and $r$ to  $\lceil4 \log\ell /\log \sigma\rceil$. We also set $M$ (RAM), used in \textsf{rrBDA-index I} (\textsf{ext}) and \textsf{rrBDA-index II} (\textsf{ext}), 
to $\Theta(|\mathcal{A}_{\ell,r}|)$ as this is the final size of index  $\mathcal{I}_{\ell,r}(S)$. 

\paragraph{Measures}
Four measures were used: 
index size; query time;
construction space; and construction time;  (see Section~\ref{sec:intro}). To measure the query and construction time, the \texttt{steady\_clock} class of \texttt{C++}  was used with the elapsed time measured in nanoseconds (ns). 
To measure the index size for each implementation as accurately as possible, the data structures required for querying the patterns were output to a file in secondary memory. The text size was not counted as part of the index size for any implementation. The construction space was measured by recording the maximum resident set size in kilobytes (KB) using the
\texttt{/usr/bin/time -v} command.

\paragraph{Results} In Sections~\ref{sec:rbd-anchors_results} to~\ref{sec:HS_results}, we present the results for the datasets of Table~\ref{tab:datasets}. To avoid cluttering the figures, we omit the results of \textsf{rrBDA-index I} (resp., \texttt{int} and \texttt{ext)}, as they were consistently outperformed by \textsf{rrBDA-index II} (resp., \texttt{int} and \texttt{ext)} in all four measures used.

\subsection{Comparing reduced bd-anchors to randomized reduced bd-anchors}\label{sec:rbd-anchors_results}

In the following, we compare reduced bd-anchors (\textsf{rBDA-compute}) to our new notion of randomized reduced bd-anchors (\textsf{rrBDA-compute}) along three dimensions (density, construction time, and construction space) using the first five datasets of Table~\ref{tab:datasets}. 

Figure~\ref{fig:rbda-construction-count} shows that the number of bd-anchors in both schemes, \textsf{rBDA-compute} and \textsf{rrBDA-compute}, decreases as $\ell$ increases, as expected. Furthermore, the trends of both schemes are the same. However, the randomized reduced bd-anchors are fewer than the reduced bd-anchors (and thus their density, i.e., their number divided by the string length $n$, is smaller), for every tested $\ell$ value and every tested dataset. Specifically, the number of randomized reduced bd-anchors is $17.8\%$ lower on average (over all datasets) and up to $46.8\%$ lower. The fact that randomized reduced bd-anchors are fewer is encouraging, since it leads to indexes with lower construction space.

Figure~\ref{fig:rbda-construction-time} shows that it is much more efficient to compute our notion of randomized reduced bd-anchors compared to reduced bd-anchors, and this is consistent over all tested $\ell$ values and datasets. Specifically, the time needed by \textsf{rrBDA-compute} was smaller than that of \textsf{rBDA-compute} by $32.4\%$ on average (over all datasets) and up to $77.6\%$ lower. On the other hand, the simple $\Theta(n\ell)$-time algorithm of~\cite{DBLP:conf/esa/LoukidesP21}, which computes reduced bd-anchors and is represented by the $\Theta(n\ell)$ line, is even slower than \textsf{rBDA-compute}. In particular, \textsf{rBDA-compute} becomes more than \emph{two orders of magnitude faster} than the $\Theta(n\ell)$ algorithm as $\ell$ increases. 

As suggested by Theorems~\ref{the:main} and~\ref{the:main2}, the runtime of \textsf{rBDA-compute} and of \textsf{rrBDA-compute} should not be affected by $\ell$ in the case of a uniformly random input string. Even if our datasets are real and thus not uniformly random, we observe that the runtime of \textsf{rBDA-compute} and \textsf{rrBDA-compute} in Figure~\ref{fig:rbda-construction-time} is not affected too much by the value of $\ell$; an exception to this is the case of SOURCES, which is explained by the fact that this dataset, in particular, is very far from being \emph{uniformly random}, which is assumed by our theoretical findings. Specifically, in this dataset, the number of $(w,k)$-minimizers is much larger than what is expected on average with increasing $\ell$, thus more ties need to be broken (Lemma~\ref{lem:main}), and this is why the runtime increases with $\ell$. 
Nevertheless, even in this bad case, \textsf{rBDA-compute} is up to $112\times$ faster than the implementation of the simple $\Theta(n\ell)$-time algorithm~\cite{DBLP:conf/esa/LoukidesP21}, and \textsf{rrBDA-compute} is up to $177\times$ faster.

\begin{figure*}[ht]
\begin{center}
\hspace{3.5cm}
\begin{subfigure}[b]{0.55\textwidth}
  \includegraphics[width=6cm]{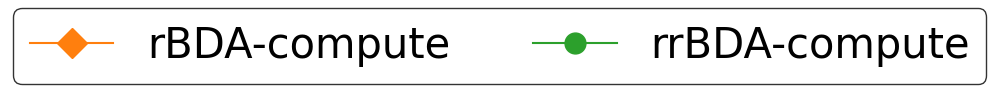}
 \end{subfigure}
 \end{center}
 
\subfloat[DNA]
  {\includegraphics[width=.2\linewidth]
  {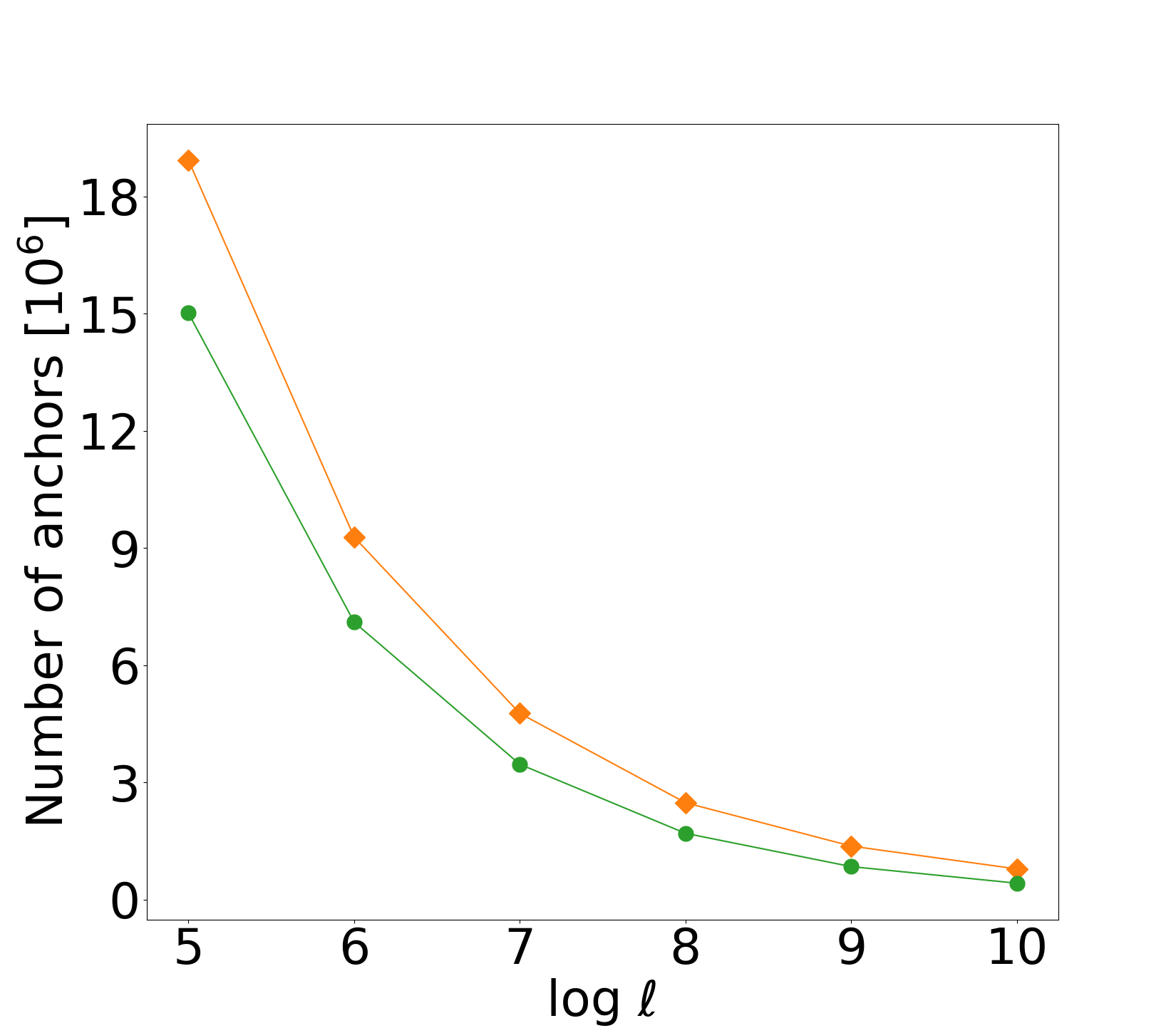}}
\subfloat[PROTEINS]
  {\includegraphics[width=.2\linewidth]
  {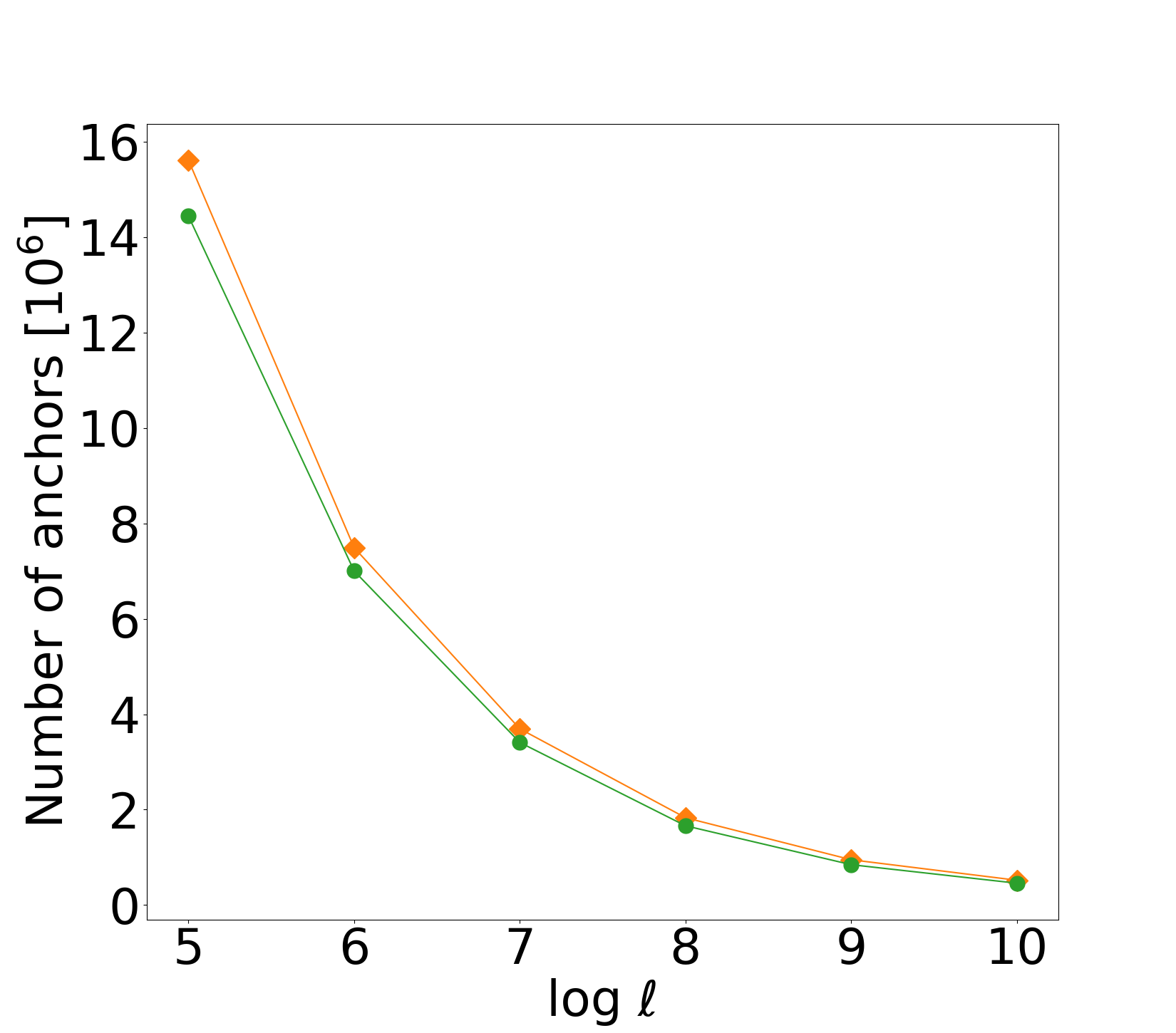}}
\subfloat[XML]
  {\includegraphics[width=.2\linewidth]
  {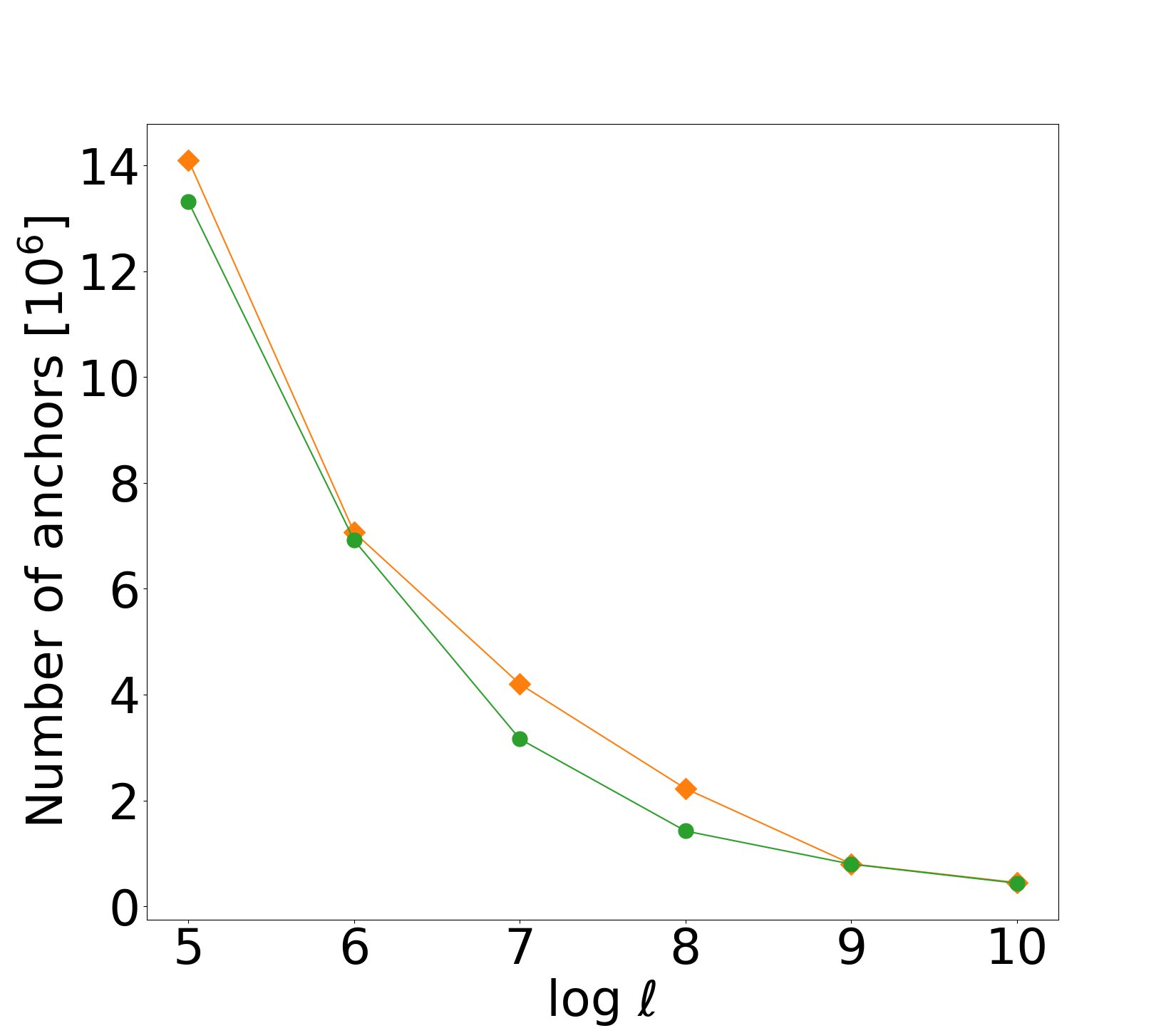}}  
\subfloat[SOURCES]
  {\includegraphics[width=.2\linewidth]
  {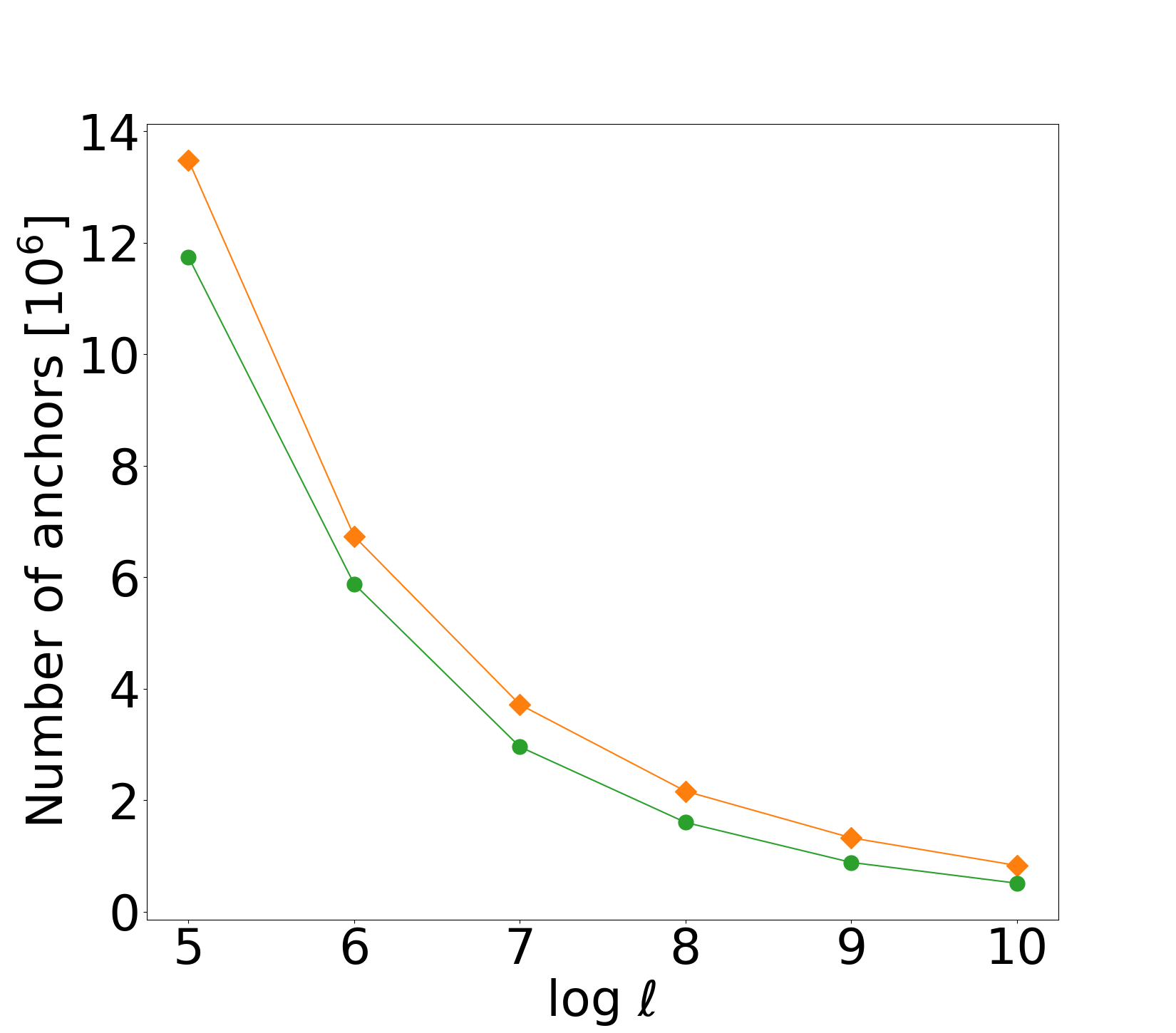}}
\subfloat[ENGLISH]
  {\includegraphics[width=.2\linewidth]
  {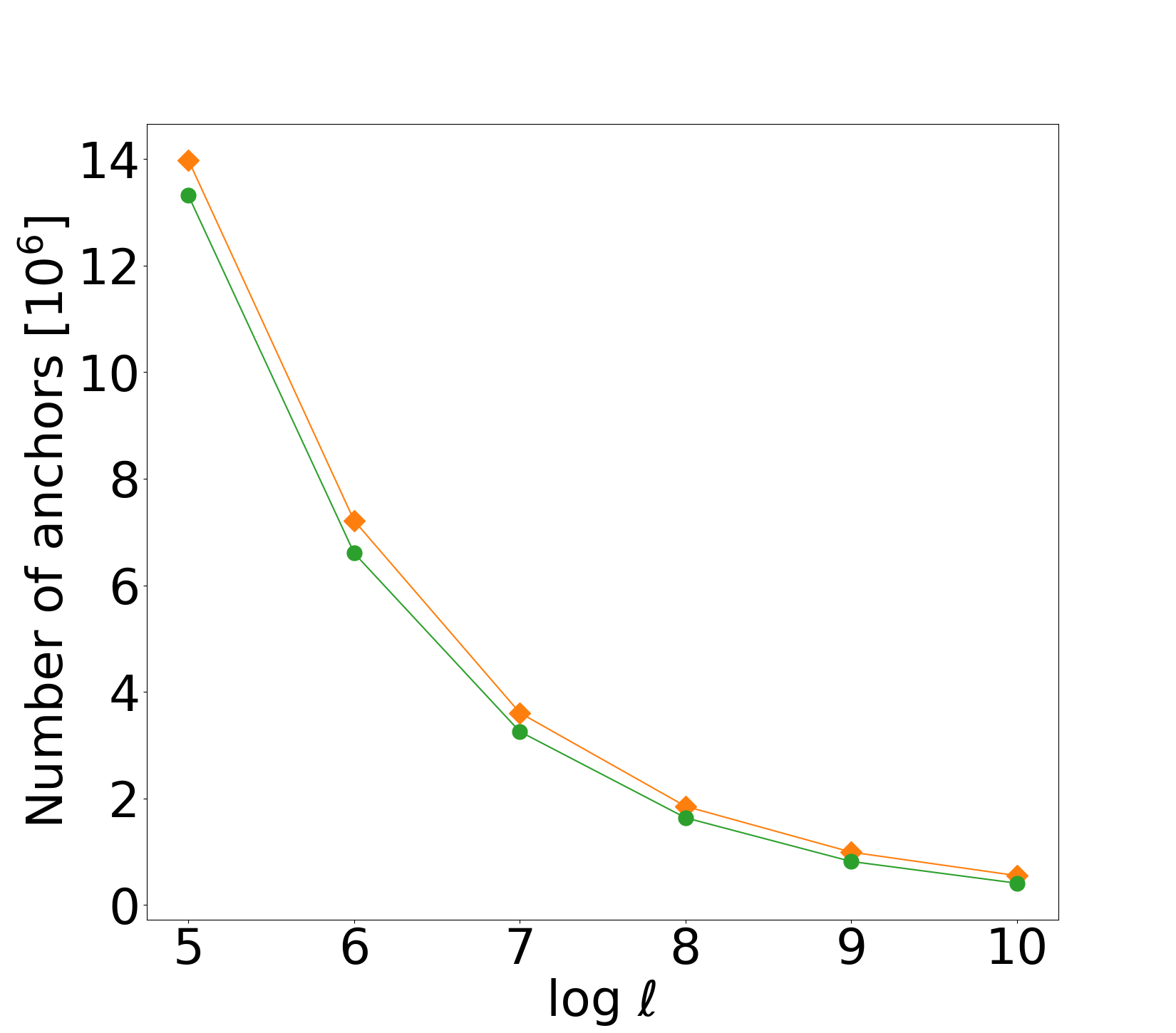}}
   \caption{\centering \label{fig:rbda-construction-count}Number of bd-anchors output by rBDA-compute and rrBDA-compute for varying $\ell$ (log-scale).}
\end{figure*}

\begin{figure*}[ht]
\begin{center}
\hspace{2.5cm}
\begin{subfigure}[b]{0.55\textwidth}
  \includegraphics[width=7.5cm]{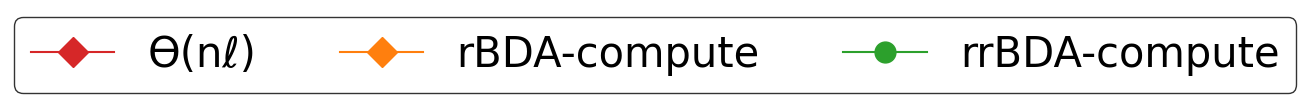}
 \end{subfigure}
 \end{center}
 
\subfloat[DNA]
  {\includegraphics[width=.2\linewidth]
  {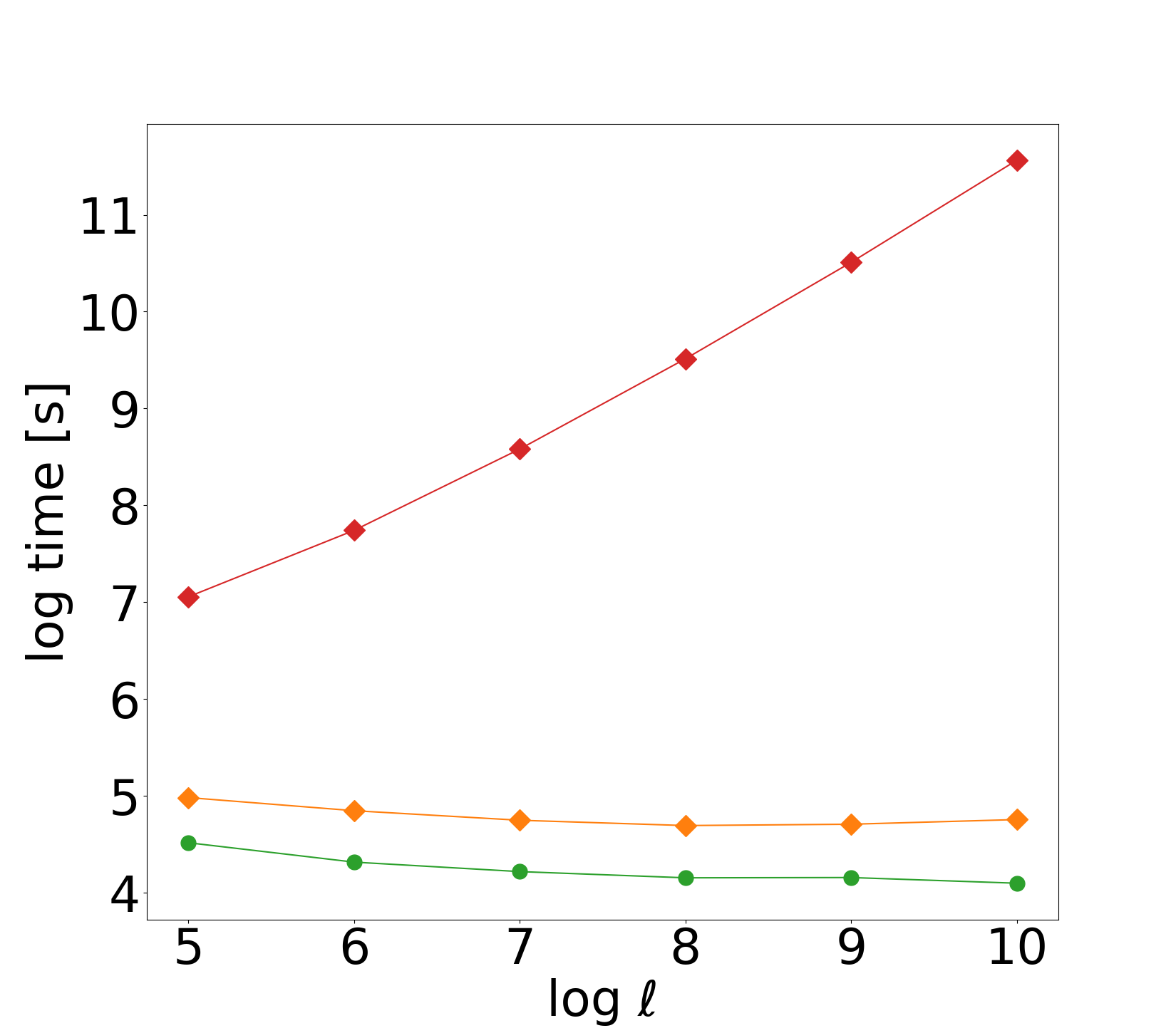}}
\subfloat[PROTEINS]
  {\includegraphics[width=.2\linewidth]
  {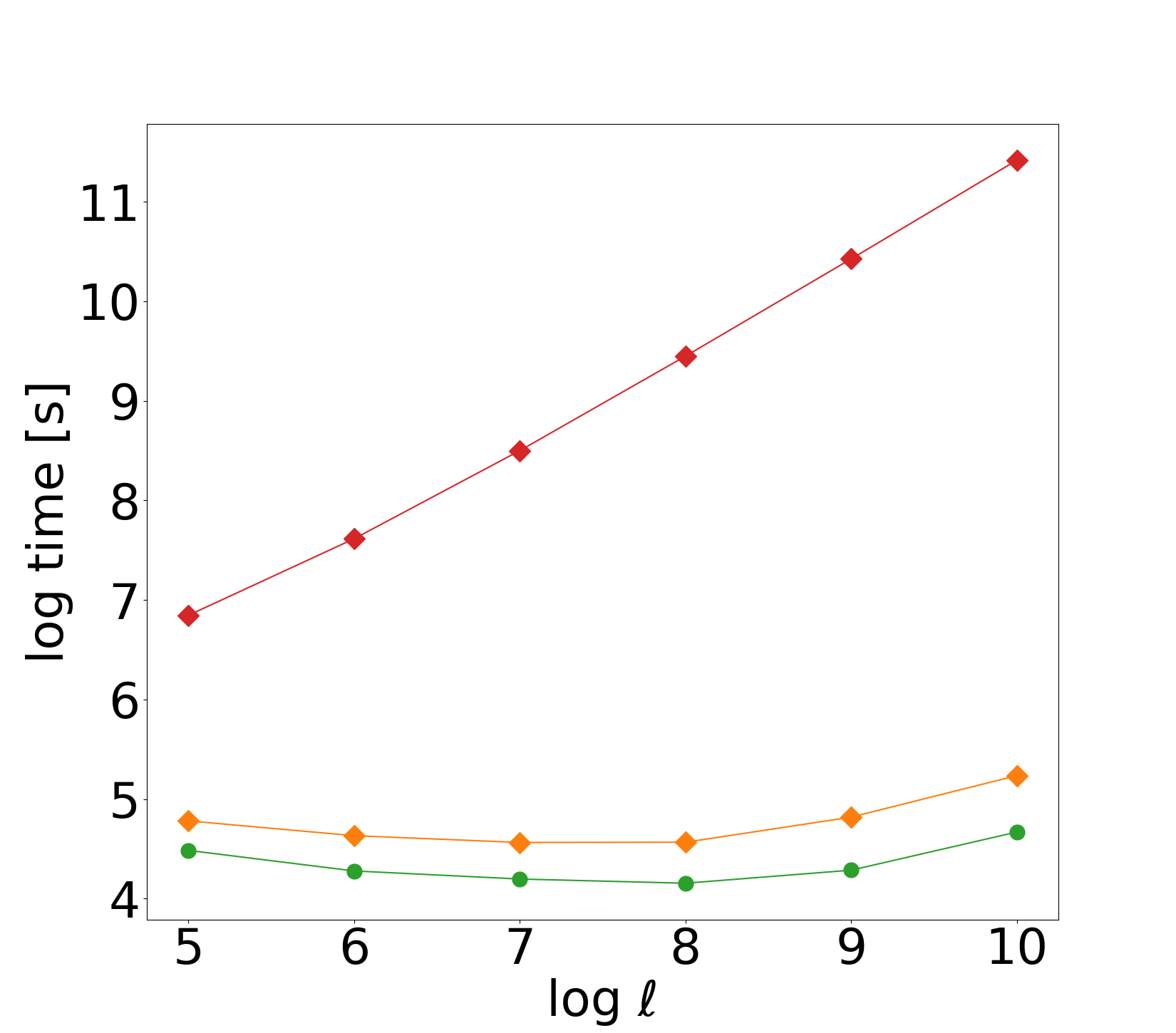}}
\subfloat[XML]
  {\includegraphics[width=.2\linewidth]
  {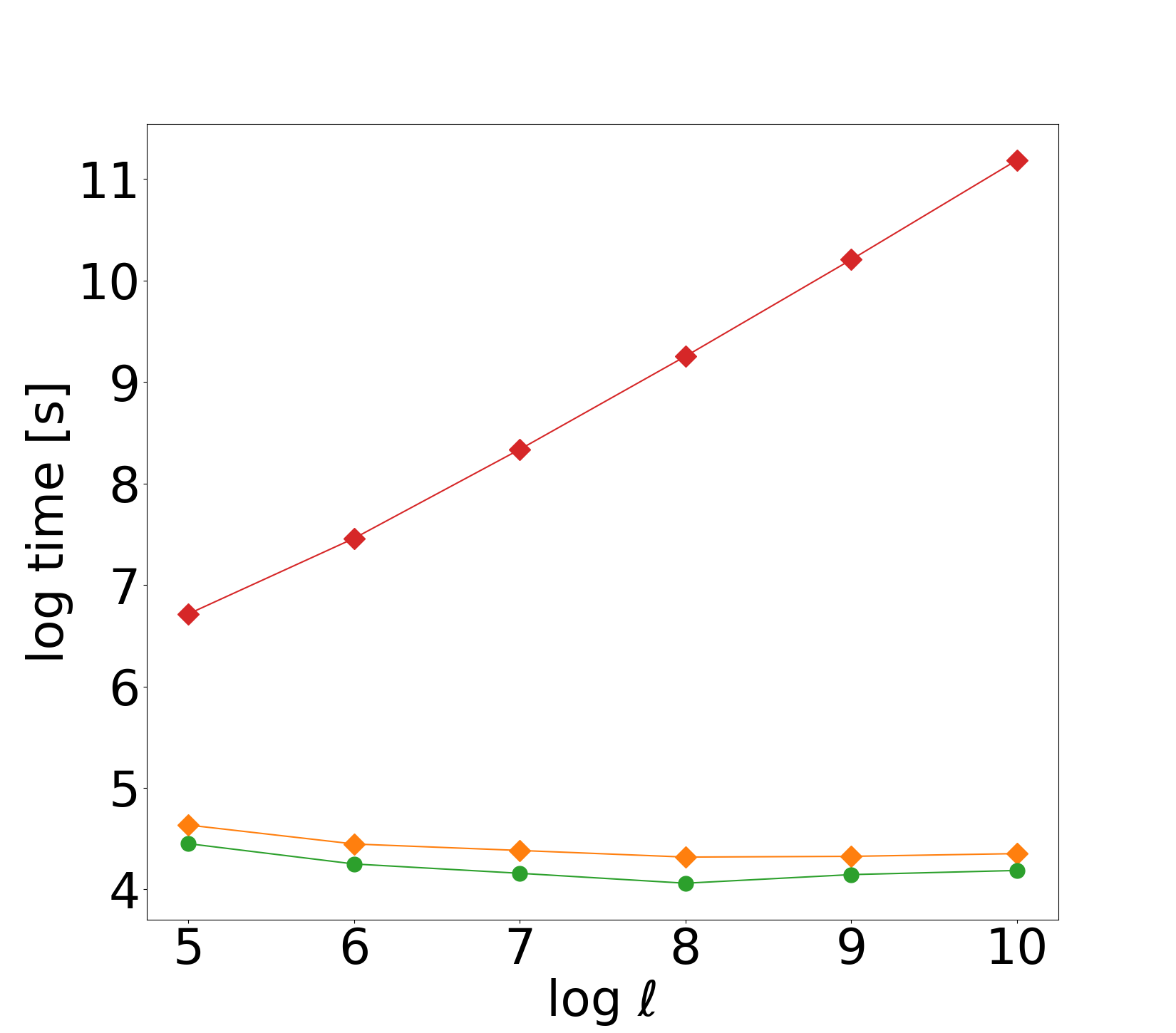}}  
\subfloat[SOURCES]
  {\includegraphics[width=.2\linewidth]
  {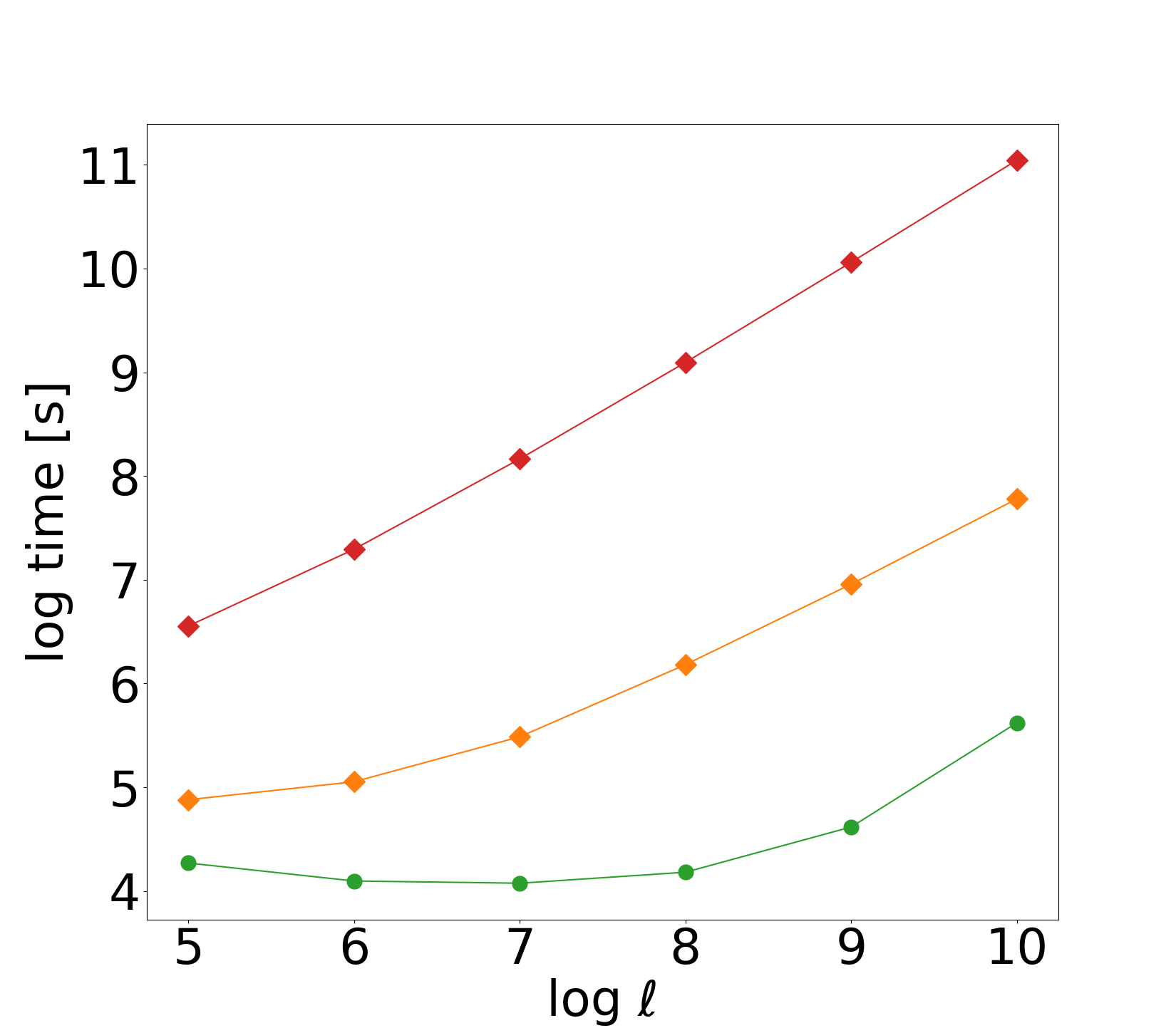}}
\subfloat[ENGLISH]
  {\includegraphics[width=.2\linewidth]
  {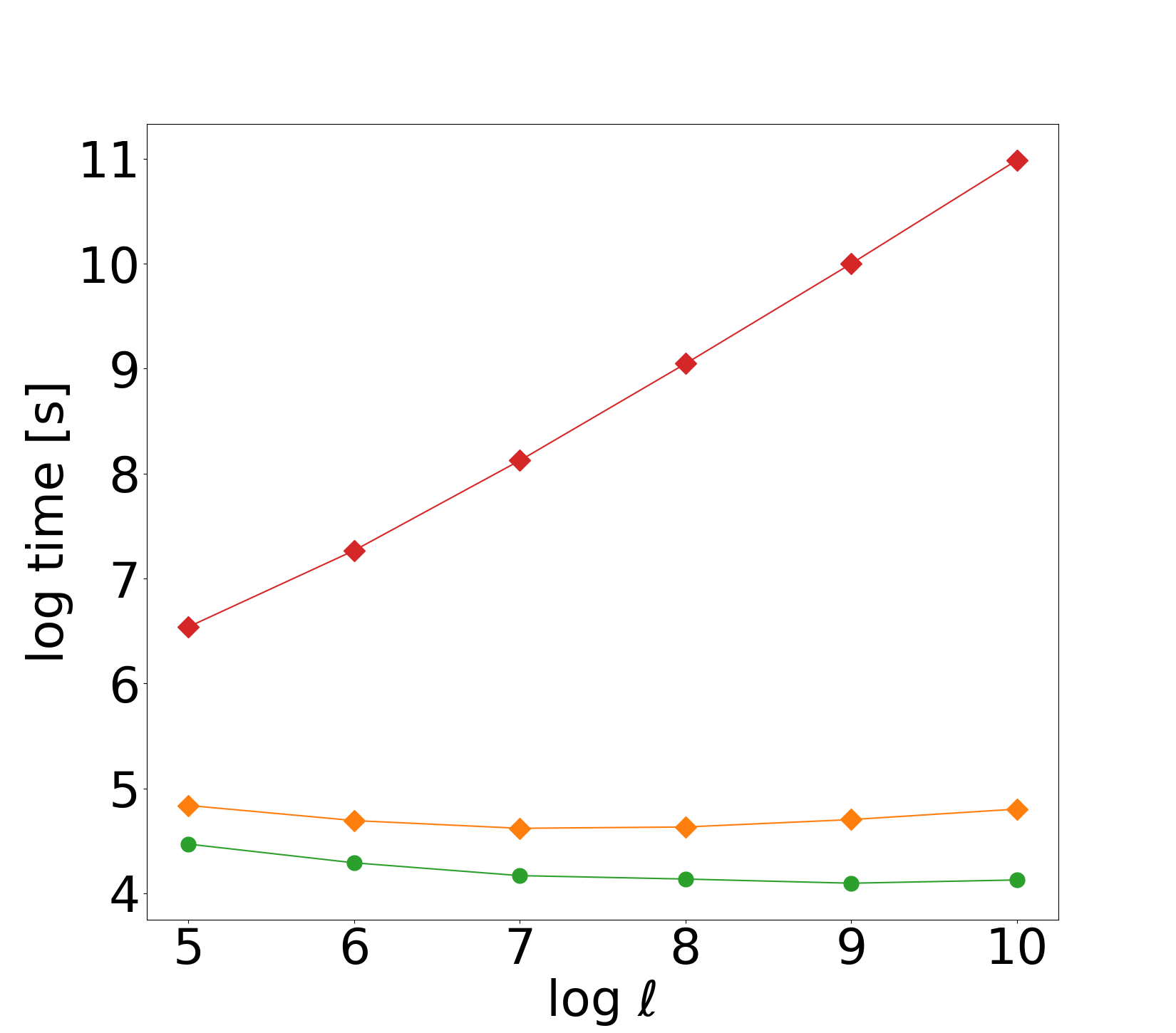}}
   \caption{\centering \label{fig:rbda-construction-time}Elapsed time to construct the set of reduced and randomized bd-anchors (seconds in log-scale) for varying $\ell$ (log-scale).}
\end{figure*}

Figure~\ref{fig:rbda-construction-space} shows that, for all tested $\ell$ values and datasets, \textsf{rBDA-compute} and \textsf{rrBDA-compute} have a \emph{similar} trend to the $\Theta(n\ell)$-time implementation~\cite{DBLP:conf/esa/LoukidesP21}.   
In particular, as $\ell$ increases the space required to construct bd-anchors is reduced.
This is expected since for increasing $\ell$ the size of the set of bd-anchors decreases. 
Also, note that the space for \textsf{rrBDA-compute} is smaller by $12.8\%$ on average (over all datasets) compared to that for \textsf{rBDA-compute} and by up to $49.1\%$ compared to that for \textsf{rBDA-compute}. These space-time results establish the practical usefulness of our Theorems~\ref{the:main} and~\ref{the:main2}. 

\begin{figure*}[ht]
\begin{center}
\hspace{2.5cm}
\begin{subfigure}[b]{0.55\textwidth}
  \includegraphics[width=7.5cm]{legend_anchors.png}
 \end{subfigure}
 \end{center}
 
\subfloat[DNA]
  {\includegraphics[width=.2\linewidth]
  {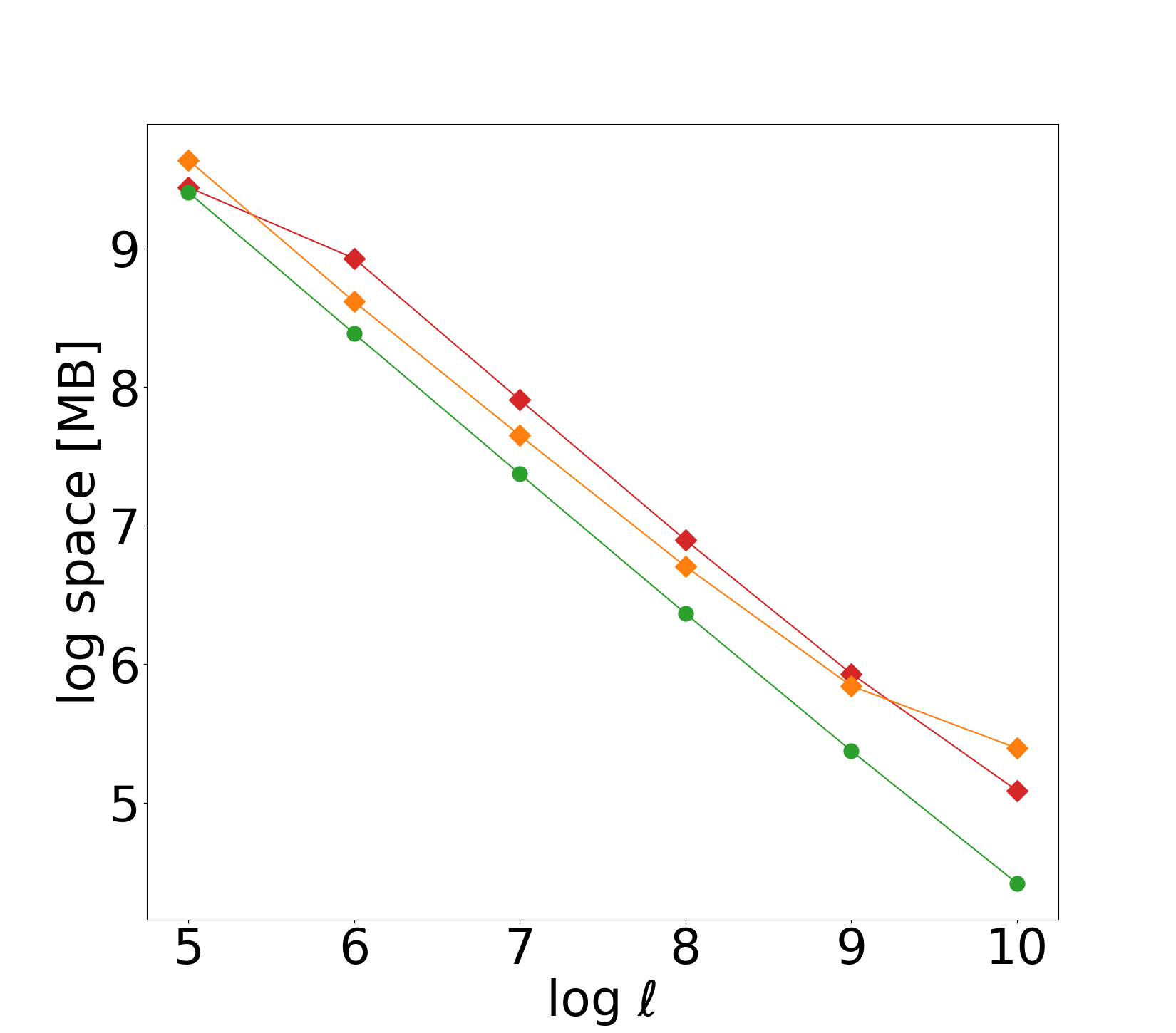}}
\subfloat[PROTEINS]
  {\includegraphics[width=.2\linewidth]
  {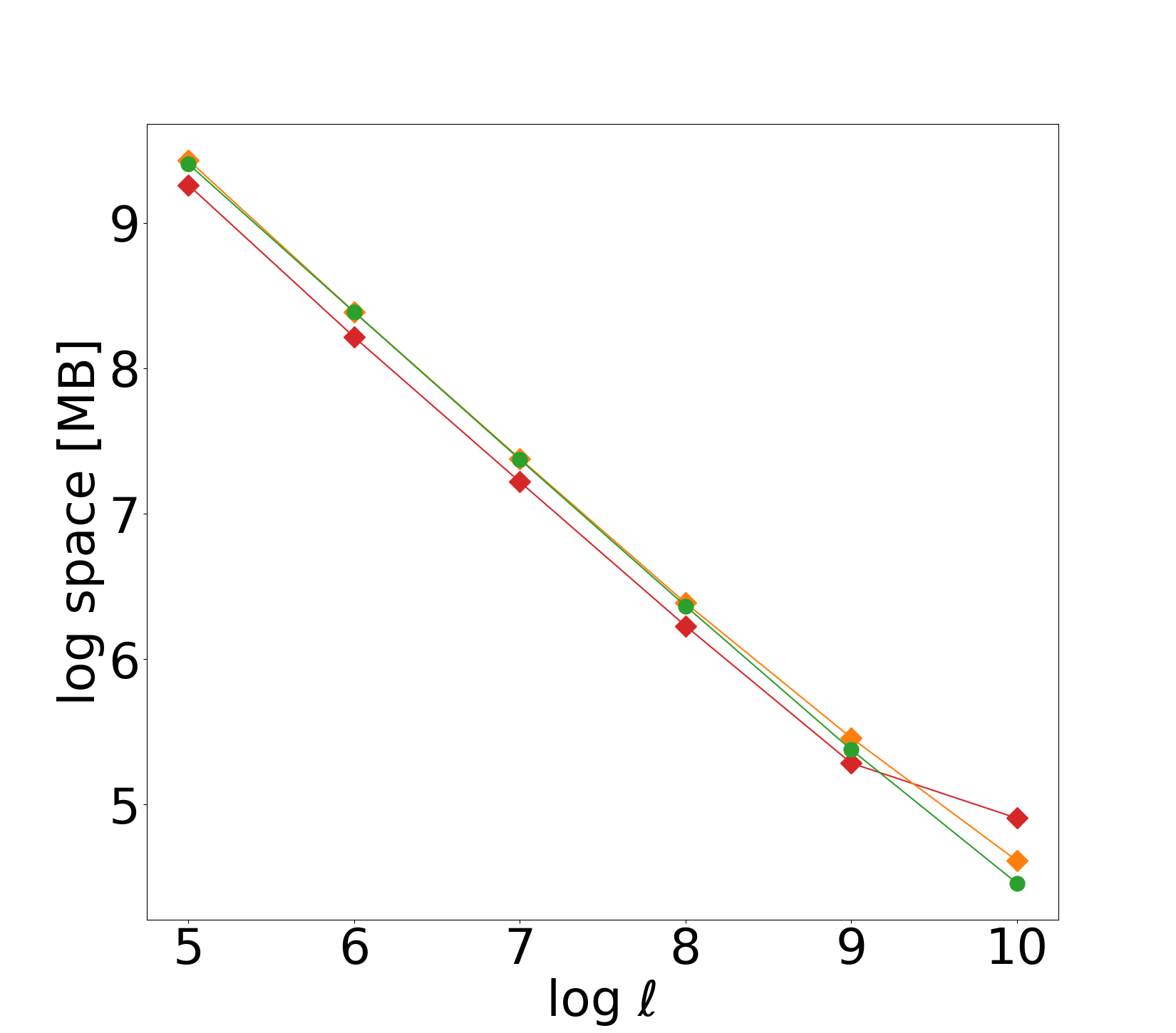}}
\subfloat[XML]
  {\includegraphics[width=.2\linewidth]
  {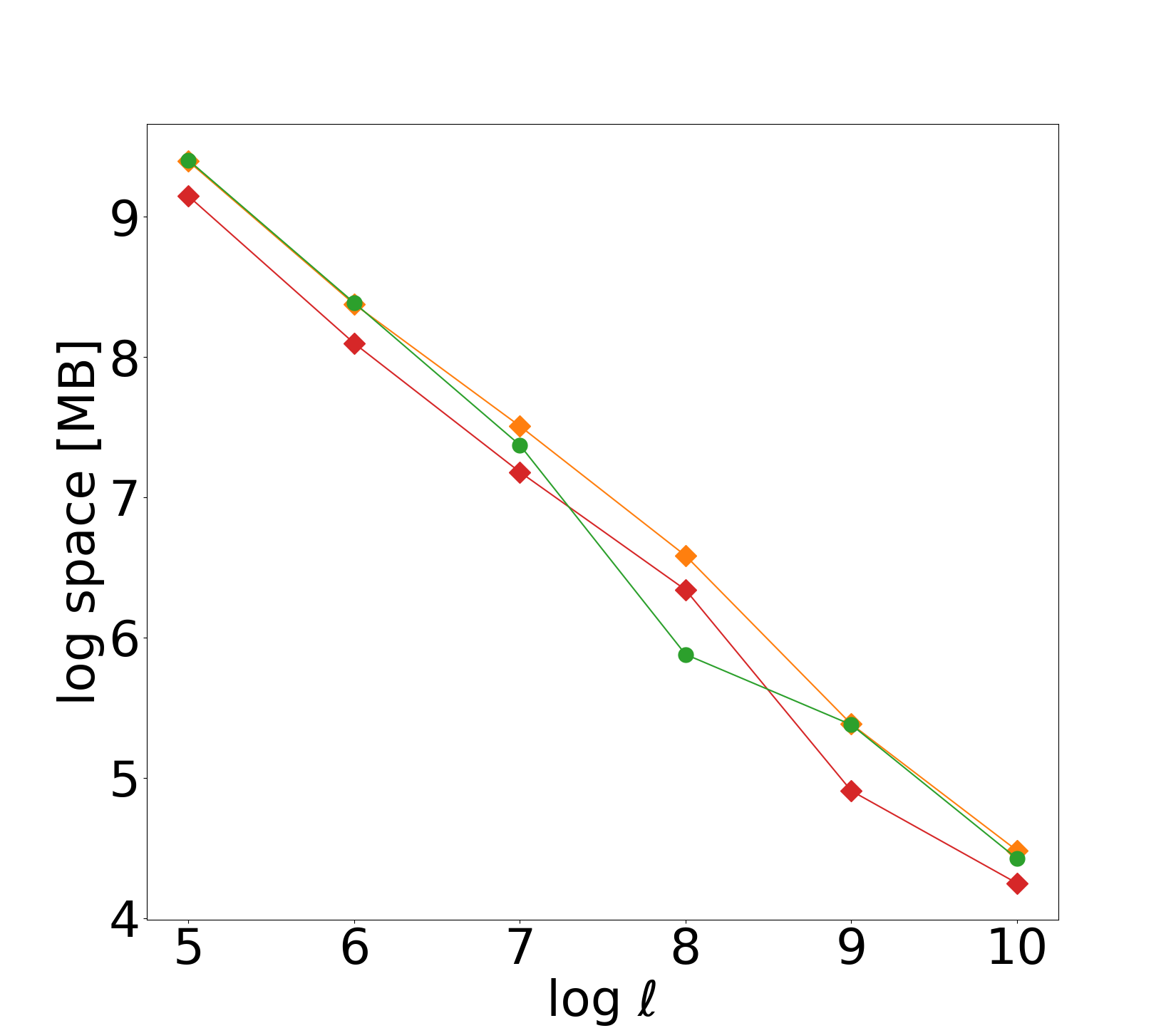}}  
\subfloat[SOURCES]
  {\includegraphics[width=.2\linewidth]
  {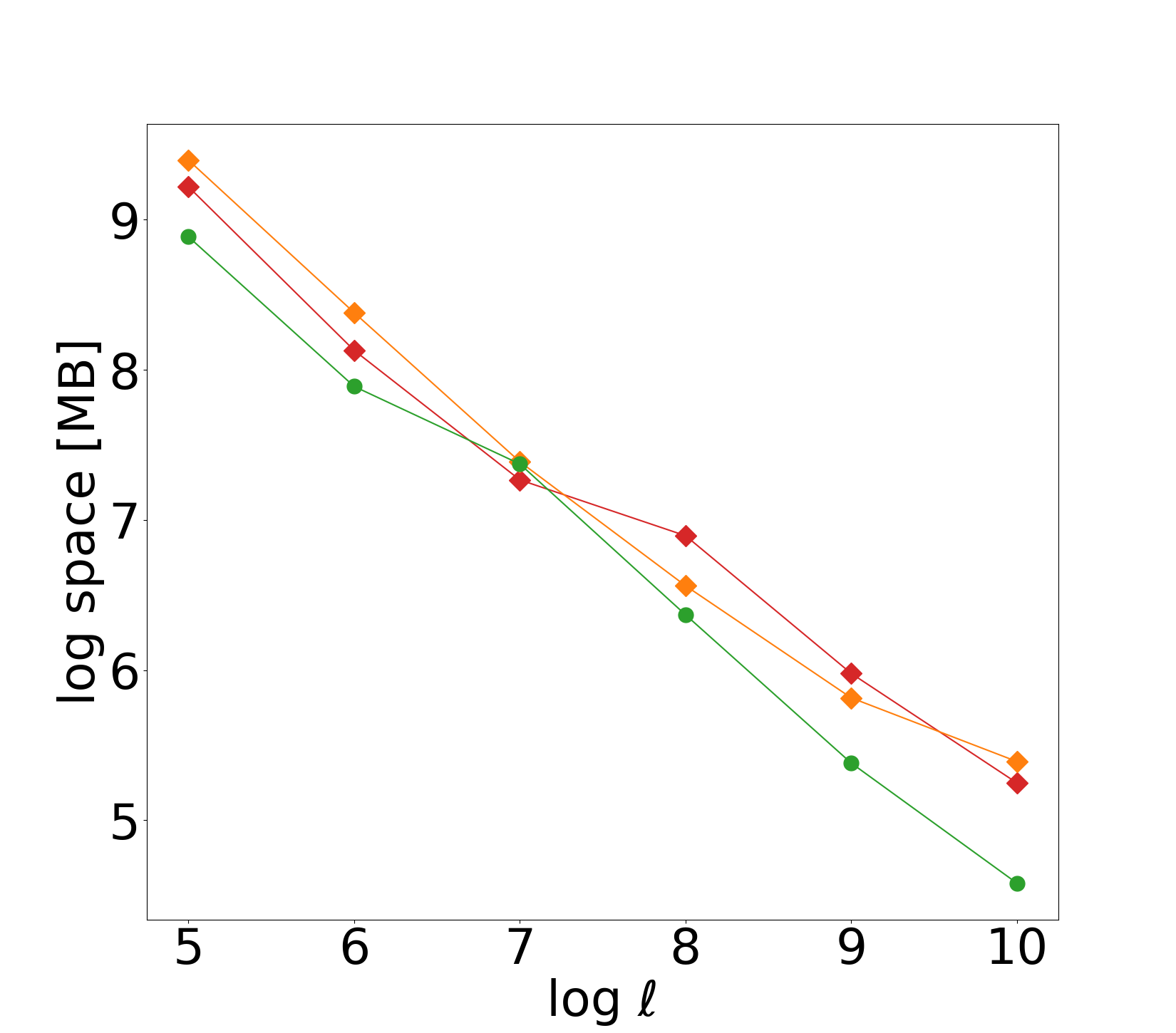}}
\subfloat[ENGLISH]
  {\includegraphics[width=.2\linewidth]
  {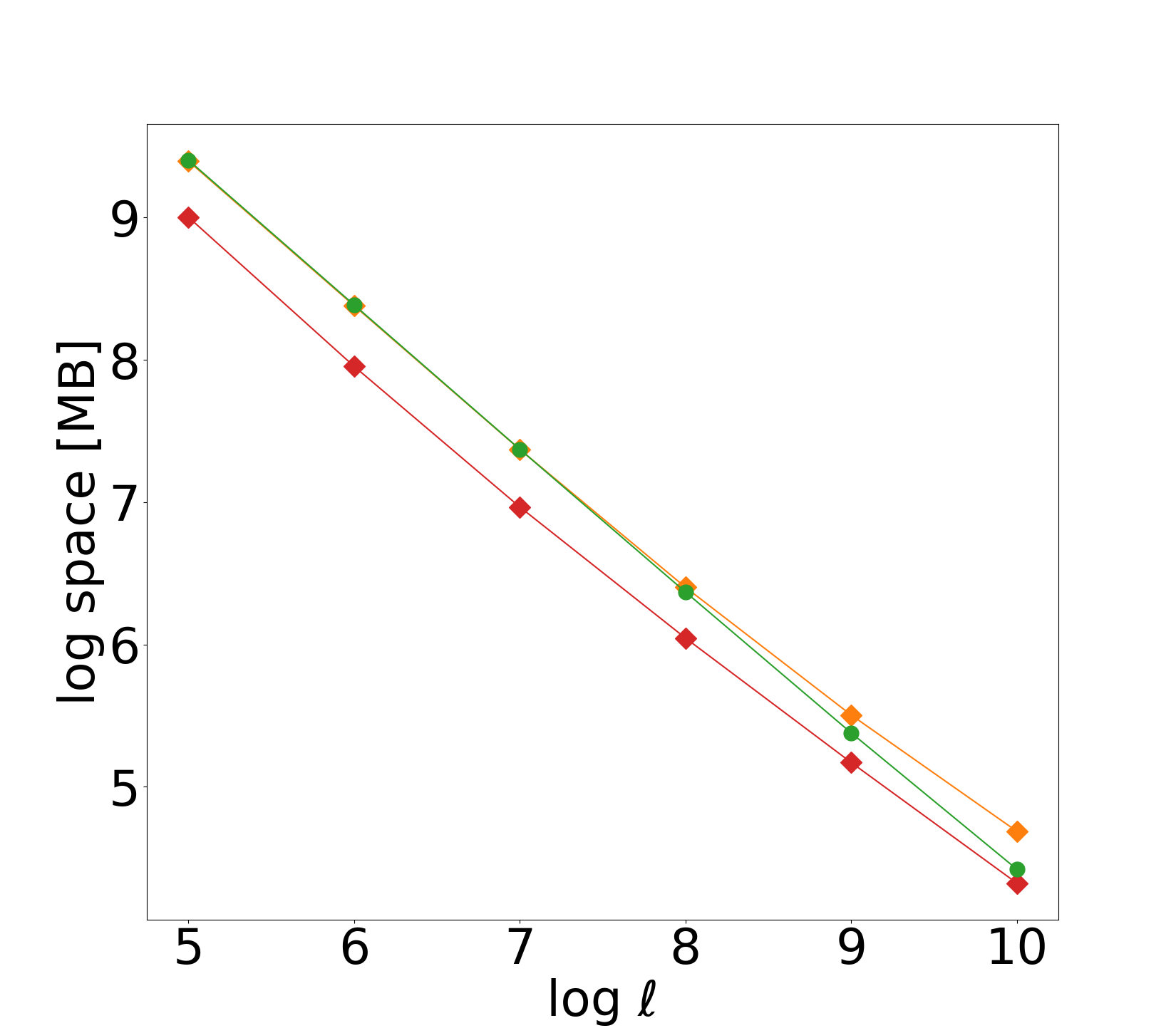}}
 \caption{\centering \label{fig:rbda-construction-space}Space required to construct the set of reduced and randomized bd-anchors (MB in log-scale) for varying $\ell$ (log-scale).}
\end{figure*}

\subsection{Index size}\label{sec:size_results}
Figure~\ref{fig:index-construction-size} shows the index size for the first five datasets of Table~\ref{tab:datasets} when using the \textsf{rrBDA-compute} implementation of Section~\ref{sec:imp} and \textsf{rrBDA-index II} (both \texttt{int} and \texttt{ext} versions have the same index size).  As expected, the index size occupied by \textsf{rrBDA-index II} decreases with increasing $\ell$. Notably, for all datasets and $\ell\geq 64$, \textsf{rrBDA-index II} is smaller than all other indexes except for the \textsf{FM-index}. This is remarkable, given that our index is not compressed and as we will show it allows much faster querying. When $\ell \geq 512$, \textsf{rrBDA-index II} \emph{outperforms all indexes} for all datasets. Specifically, when $\ell=512$, \textsf{rrBDA-index II} requires $59.1\%$ less space than \textsf{FM-index} on average (over all datasets), while when $\ell=1024$ our index requires $77.9\%$ less space. 

\begin{figure*}[ht]
\begin{subfigure}[b]{0.49\textwidth}
\hspace{2.5cm}
  \includegraphics[width=13cm]{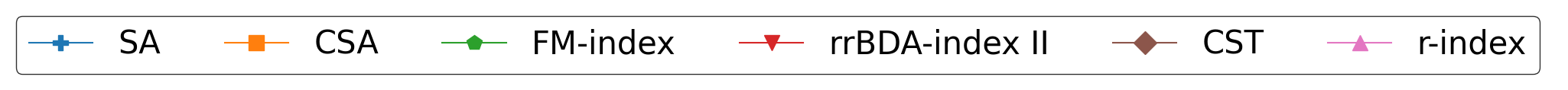}
 \end{subfigure}
 
\subfloat[DNA]
  {\includegraphics[width=.2\linewidth]
  {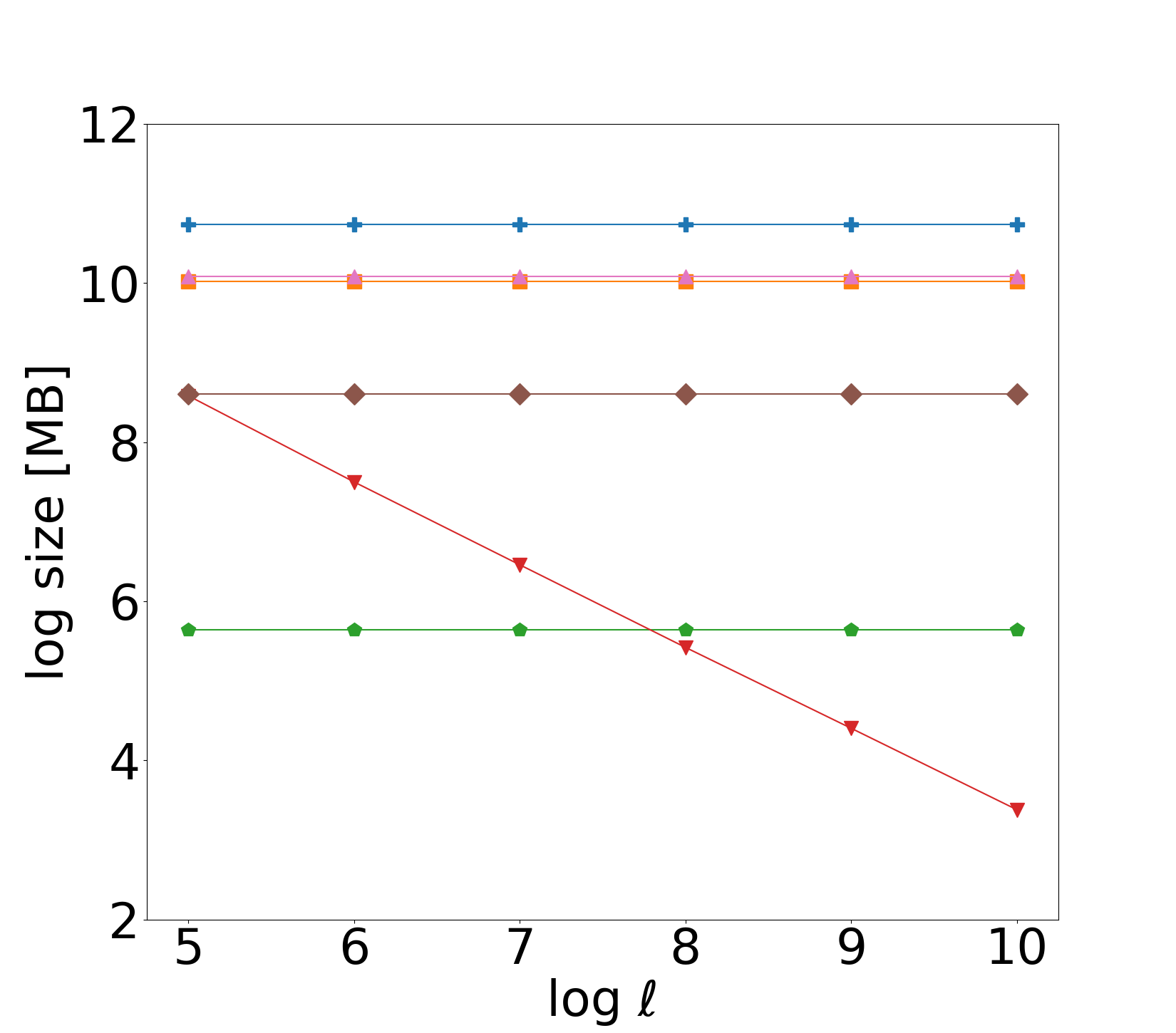}}
\subfloat[PROTEINS]
  {\includegraphics[width=.2\linewidth]
  {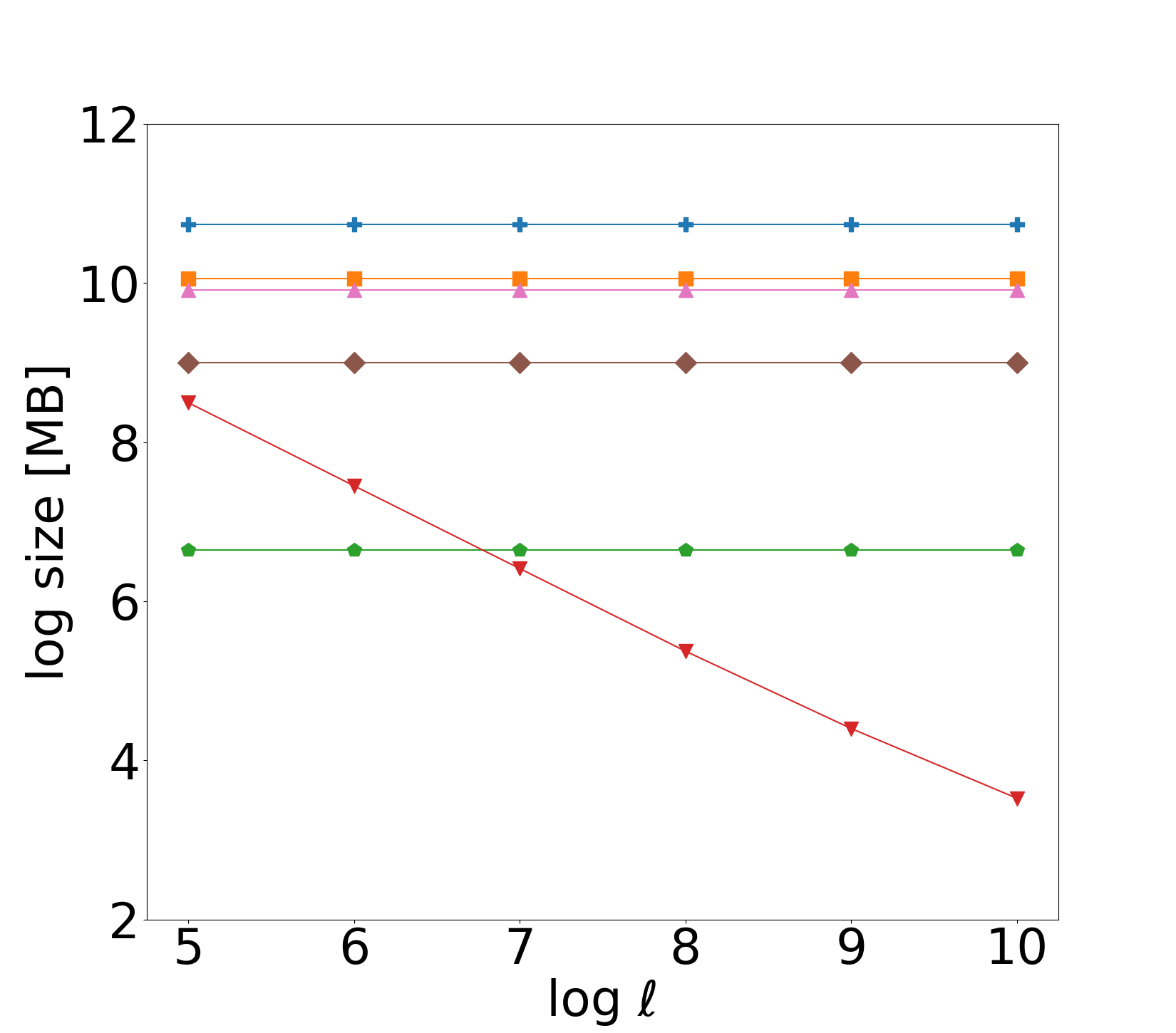}}
\subfloat[XML]
  {\includegraphics[width=.2\linewidth]
  {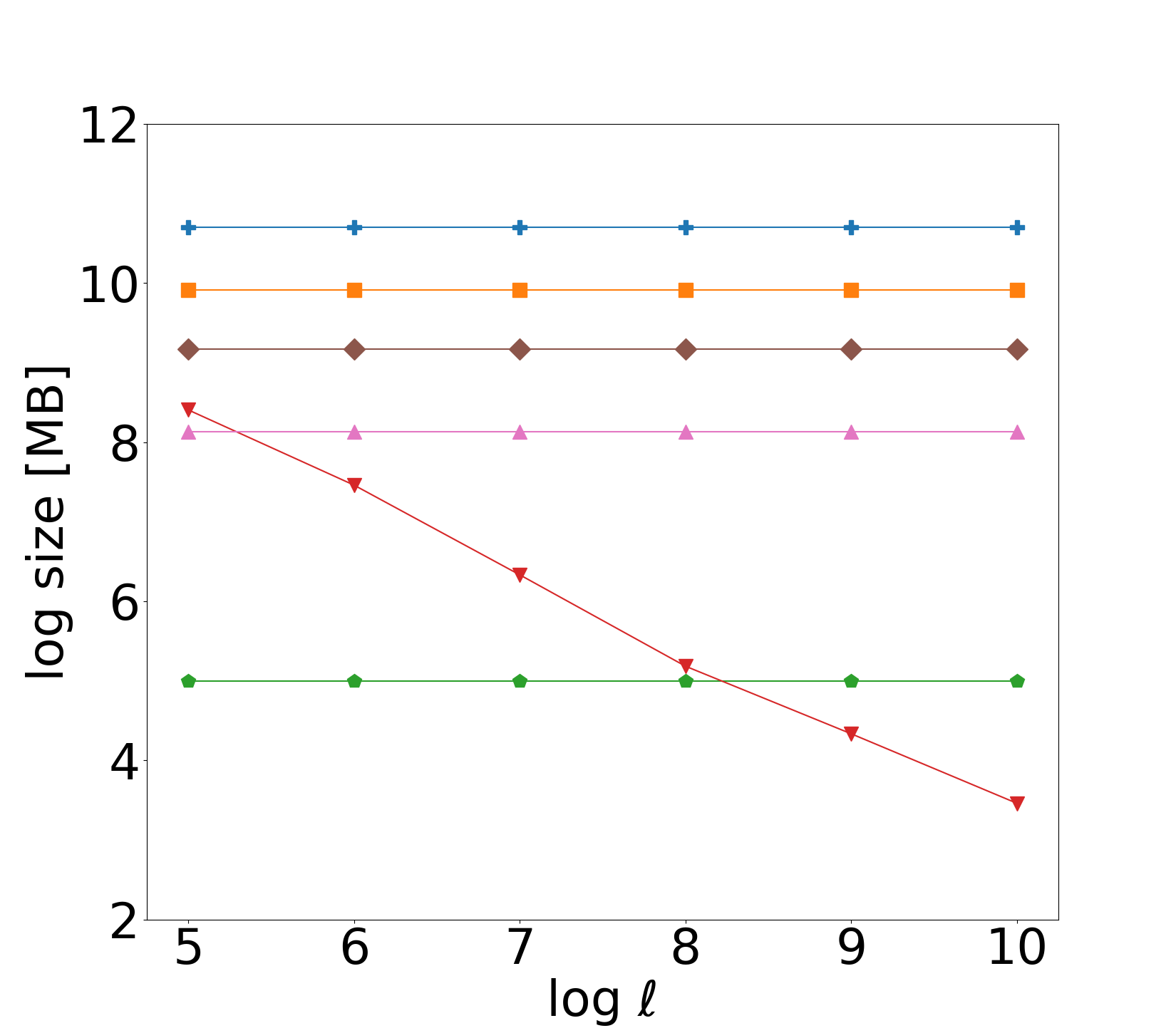}}  
\subfloat[SOURCES]
  {\includegraphics[width=.2\linewidth]
  {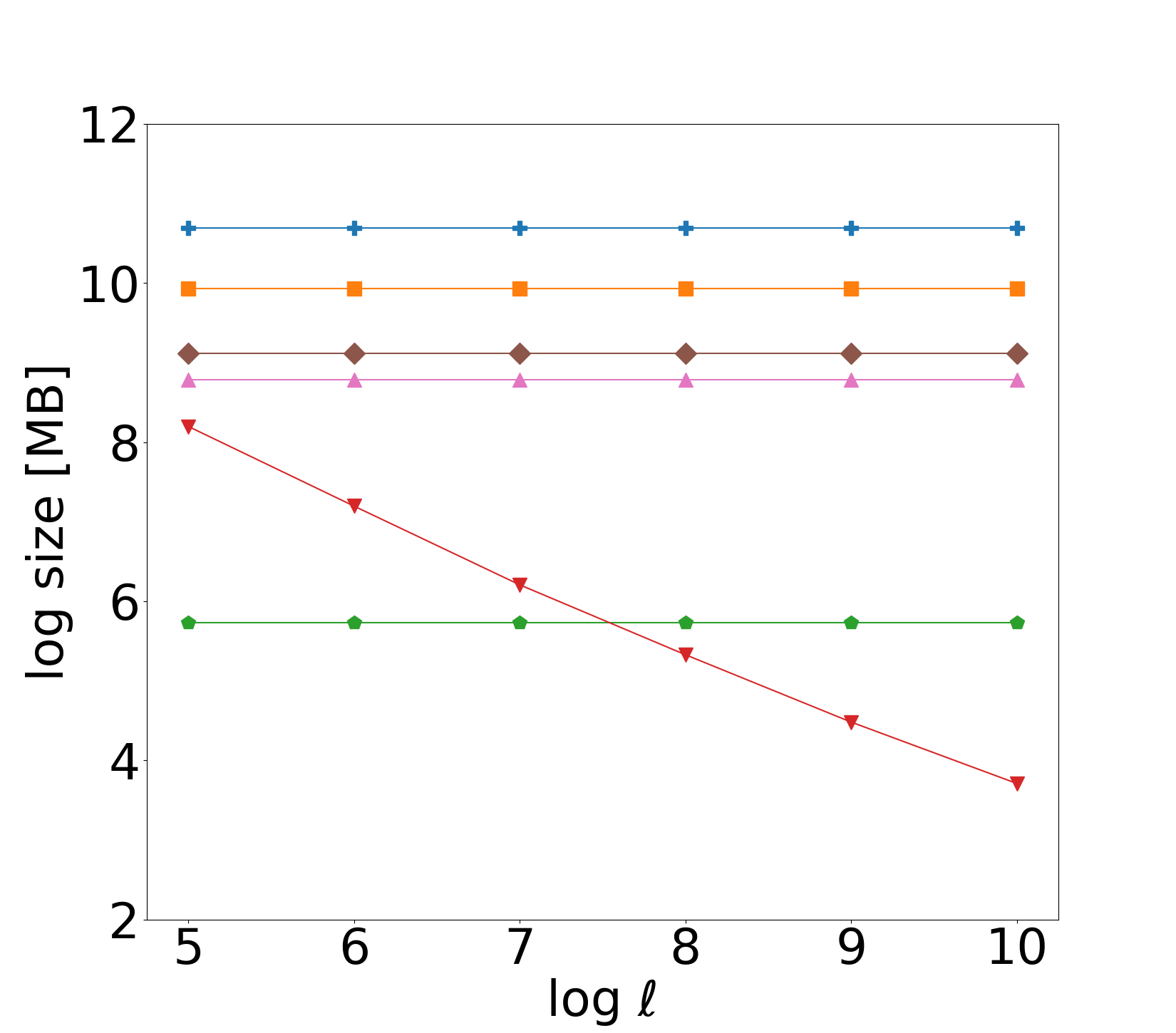}}
\subfloat[ENGLISH]
  {\includegraphics[width=.2\linewidth]
  {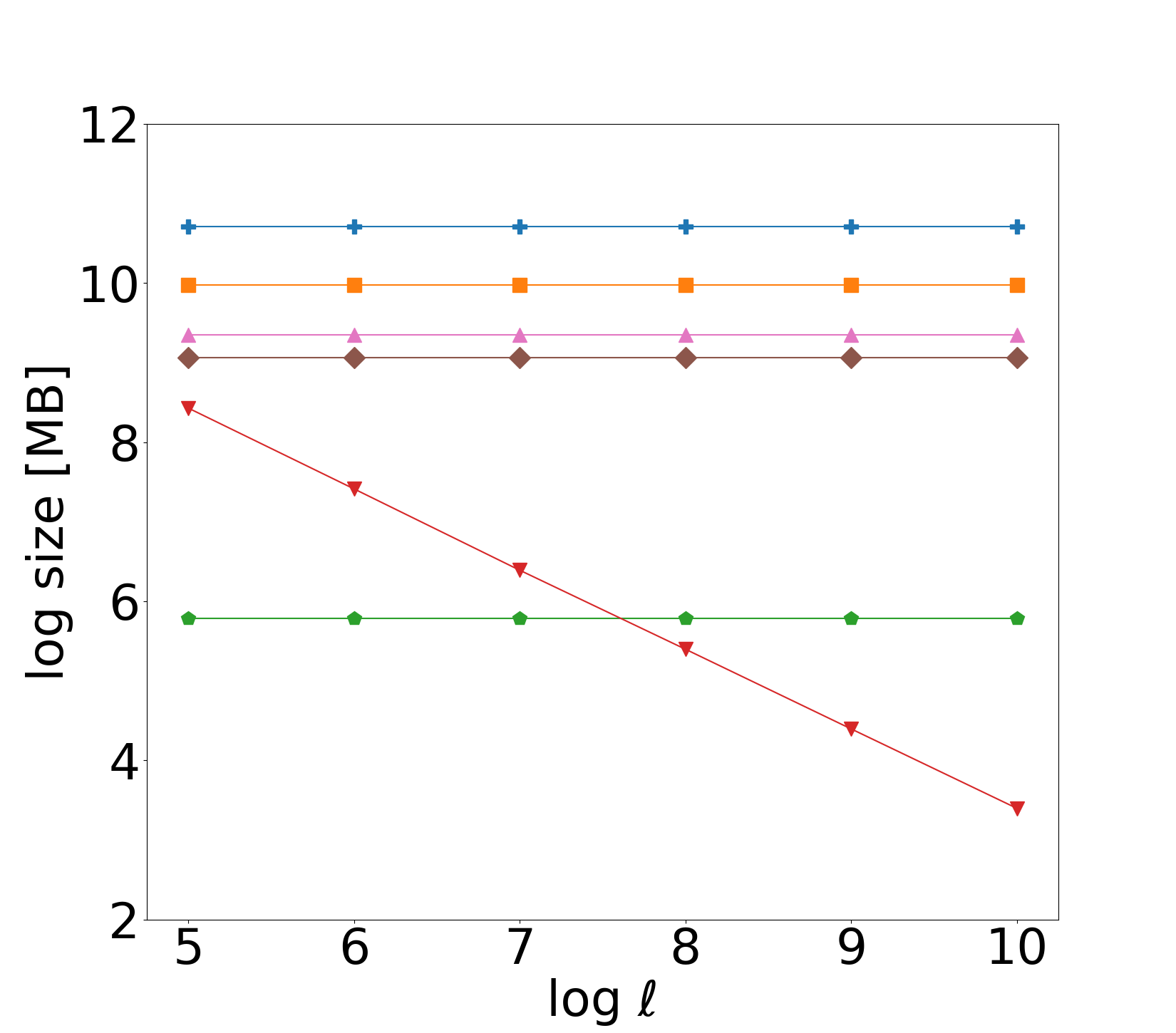}}
  \caption{\label{fig:index-construction-size}Size of different indexes (MB in log-scale) for varying $\ell$ (log-scale).}
\end{figure*}

\subsection{Query time}\label{sec:query_time_results}
Figure~\ref{fig:pattern-matching-time} shows the average query time (over all patterns) for the first five datasets in Table~\ref{tab:datasets} when  
$\ell=|P|$. Both \textsf{rrBDA-index II (int)} and \textsf{rrBDA-index II (ext)} have the same query time, so we write  \textsf{rrBDA-index II} to refer to either. 
For all datasets and $\ell\geq 64$, \textsf{rrBDA-index II} is up to \emph{several orders of magnitude} faster than the compressed indexes, especially for large alphabets, which is consistent with the observations made in~\cite{fm-index_large_alphabet_1,fm-index_large_alphabet}. Notably, \textsf{rrBDA-index II} is even faster than the \textsf{SA}, with the single exception of $\ell=32$ in the XML dataset.   

\begin{figure*}[ht]
\begin{subfigure}[b]{0.49\textwidth}
\hspace{2.5cm}
  \includegraphics[width=13cm]{legend_size.png}
 \end{subfigure}
 
\subfloat[DNA]
  {\includegraphics[width=.2\linewidth]
  {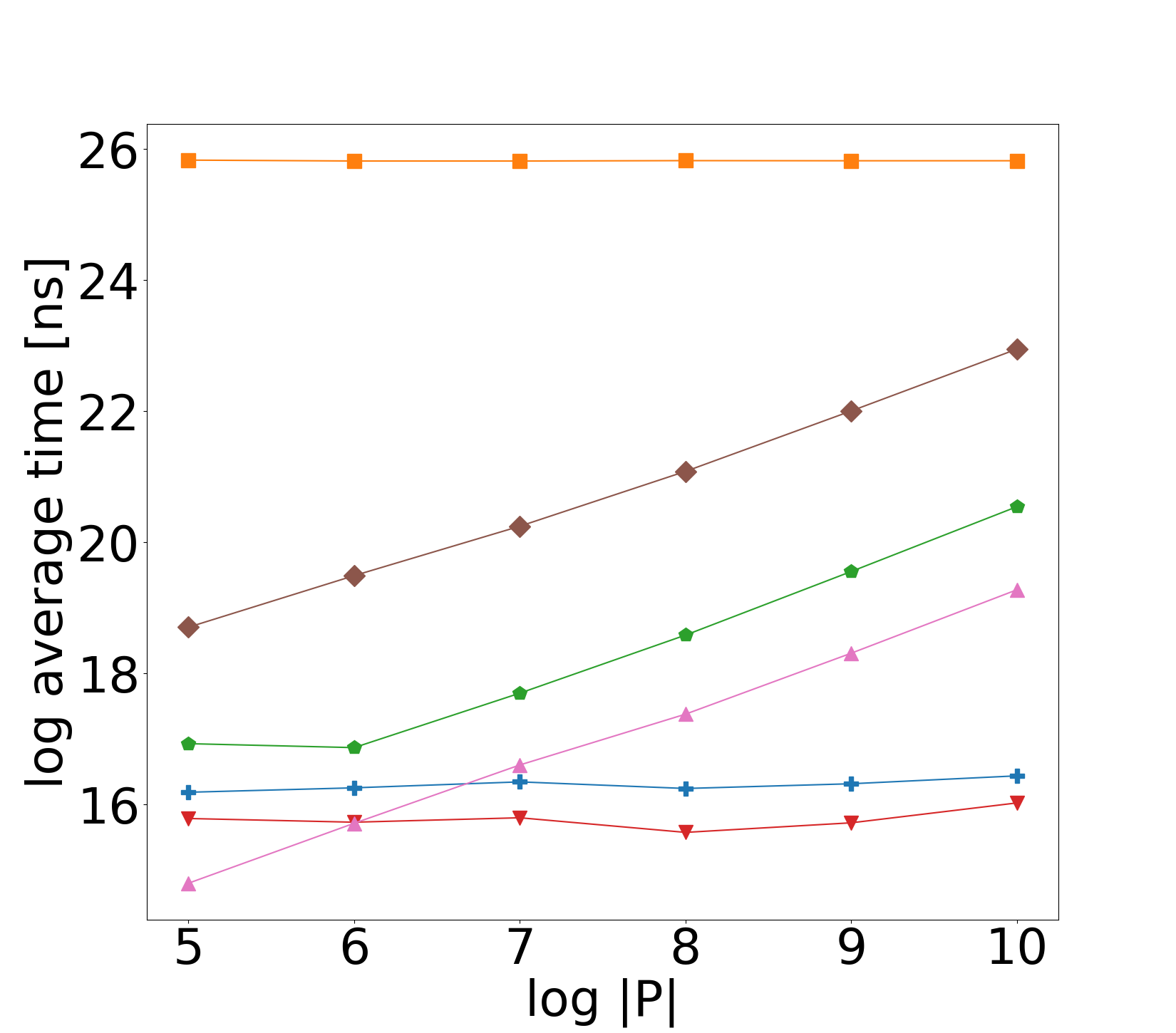}}
\subfloat[PROTEINS]
  {\includegraphics[width=.2\linewidth]
  {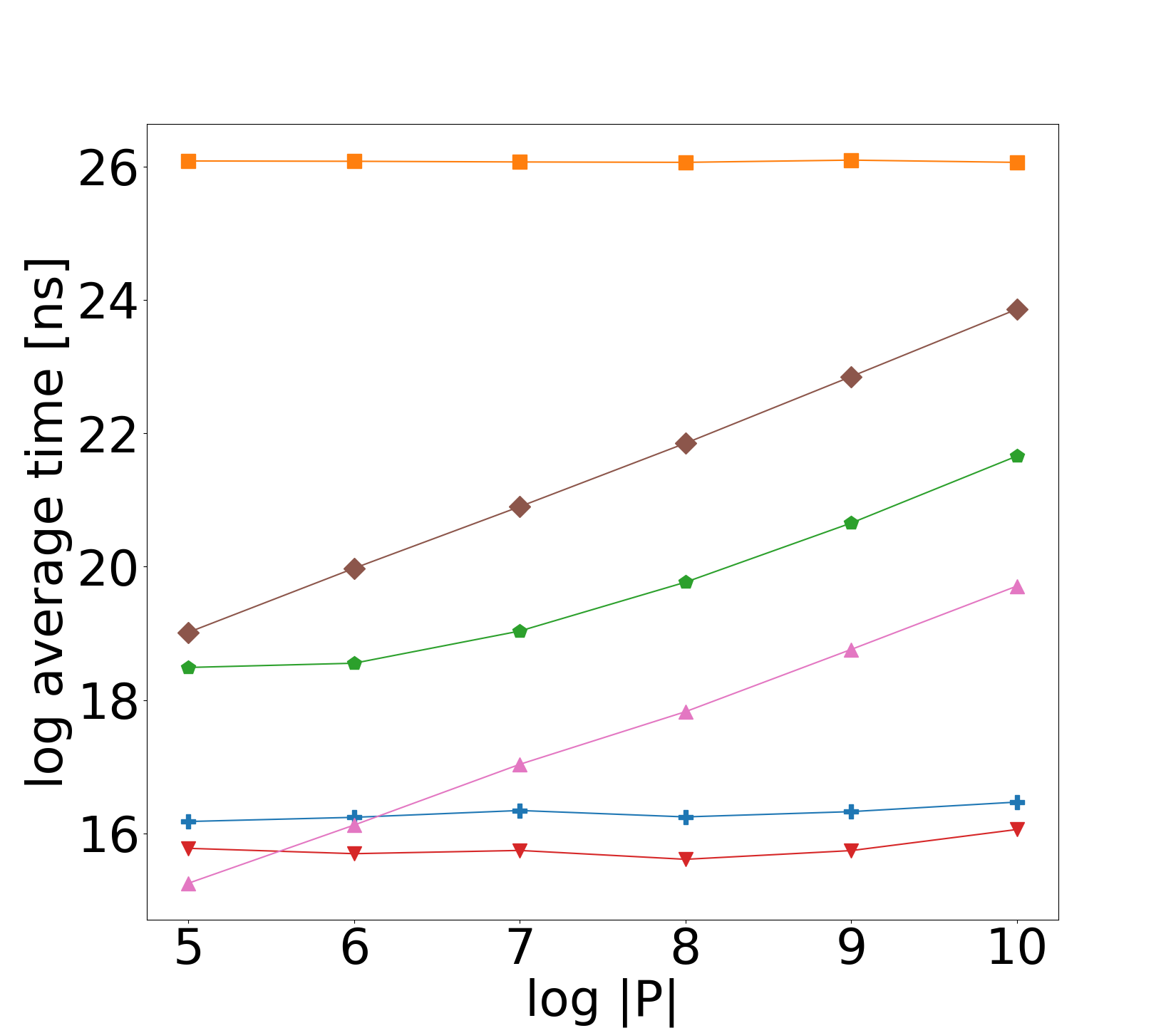}}
\subfloat[XML]
  {\includegraphics[width=.2\linewidth]
  {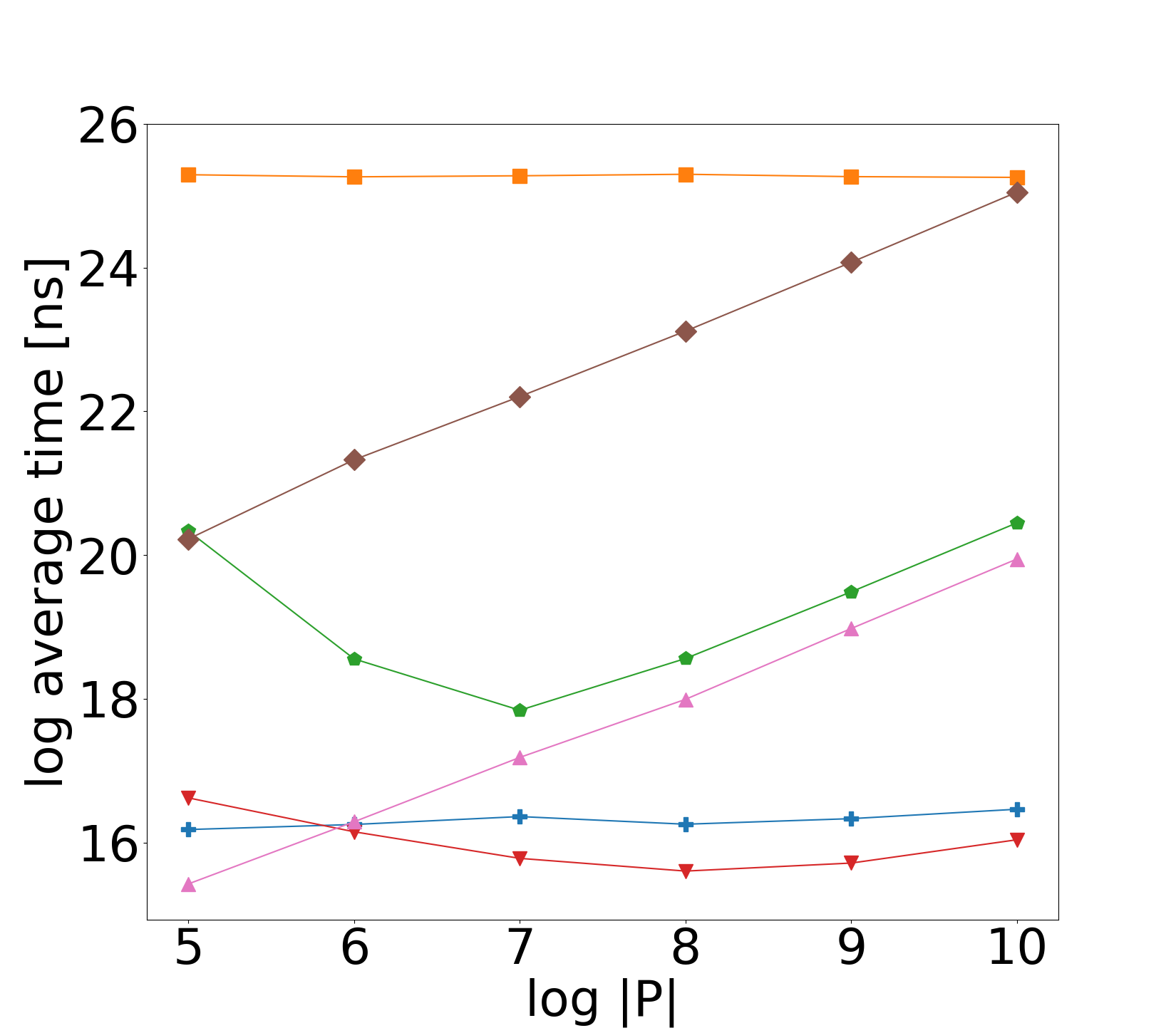}}  
\subfloat[SOURCES]
  {\includegraphics[width=.2\linewidth]
  {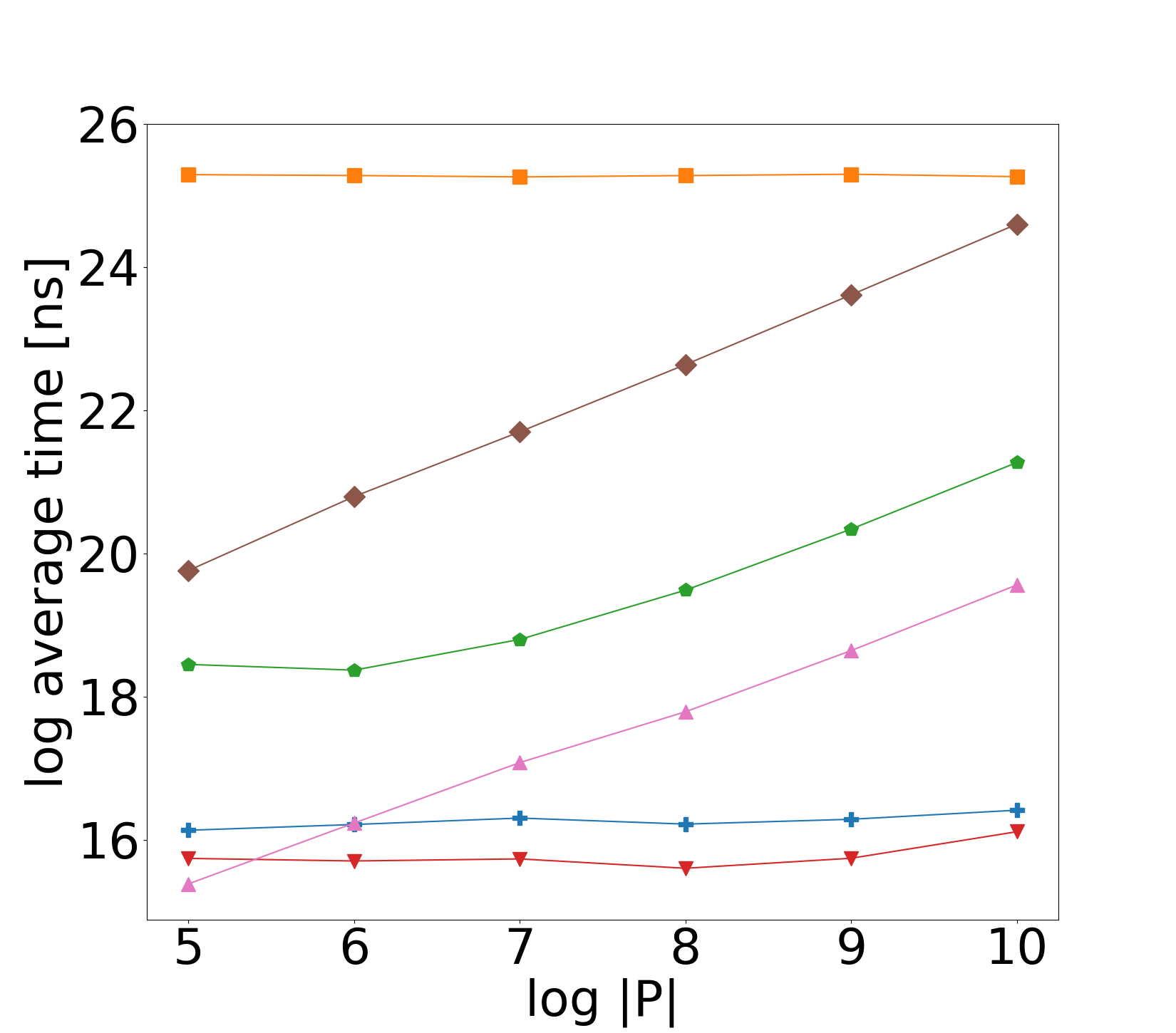}}
\subfloat[ENGLISH]
  {\includegraphics[width=.2\linewidth]
  {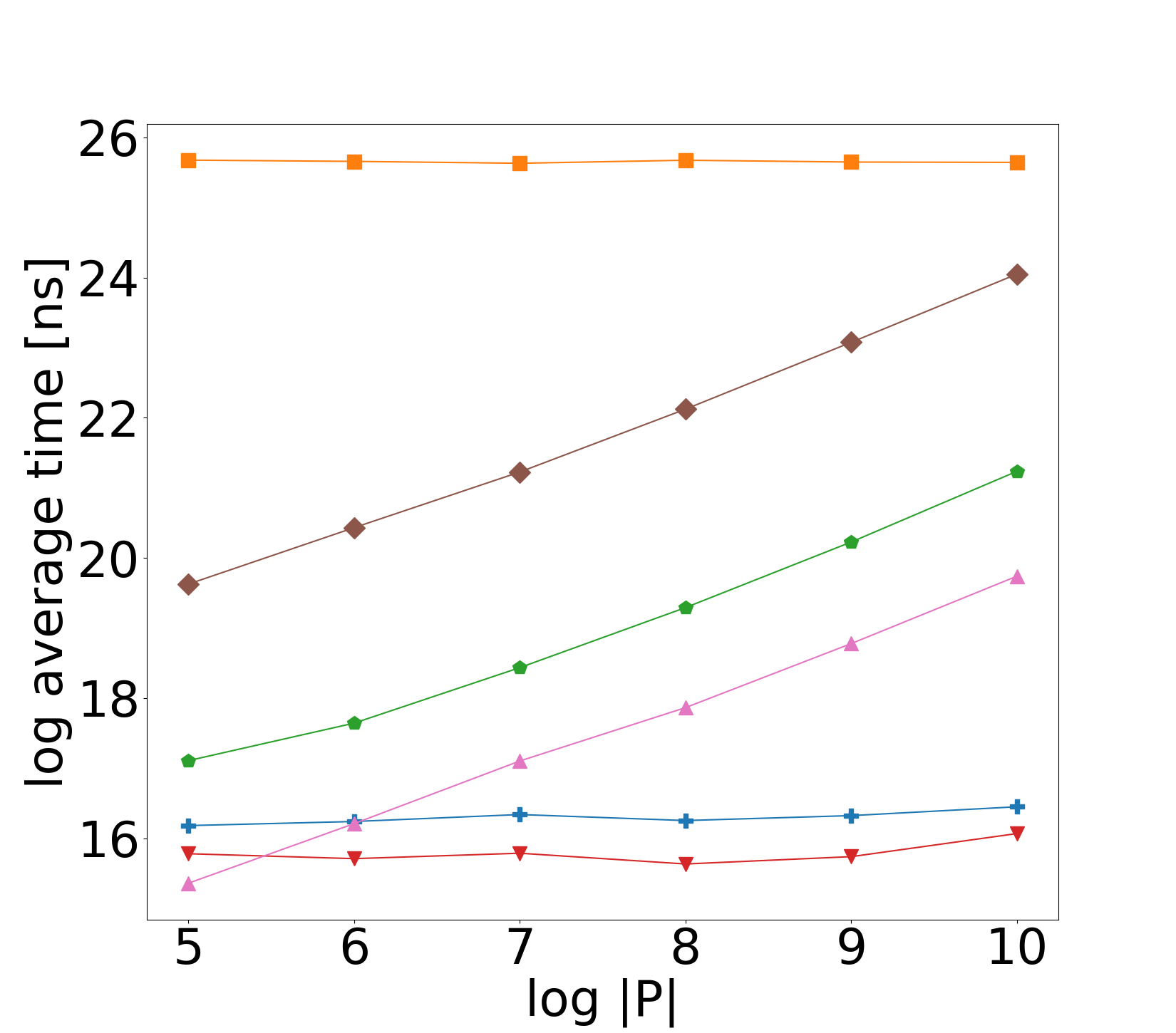}}
\caption{\label{fig:pattern-matching-time}Average time for pattern matching (nanoseconds in log-scale) for varying $|P|$ (log-scale).}
\end{figure*}

\subsection{Index construction space}\label{sec:con_space_results}
Figure~\ref{fig:index-construction-space} shows the index construction space required for the first five datasets of Table~\ref{tab:datasets}. Notably, for all datasets when $\ell \geq 128$, \textsf{rrBDA-index II (int)} and \textsf{rrBDA-index II (ext)} outperform all other indexes. As expected, the construction space required by \textsf{rrBDA-index II (int)} decreases with increasing $\ell$. The construction space of \textsf{rrBDA-index II (ext)} decreases with $\ell$ initially but remains almost the same for $\ell \geq 256$. This is because approximately 600MB of space (internal memory) are used by the implementations constructing arrays \textsf{SA} and \textsf{LCP} in external memory.  

\begin{figure*}[ht]
\begin{subfigure}[b]{0.49\textwidth}
\hspace{0.5cm}
  \includegraphics[width=16cm]{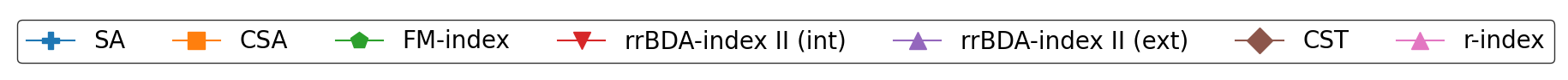}
 \end{subfigure}
 
\subfloat[DNA]
  {\includegraphics[width=.2\linewidth]
  {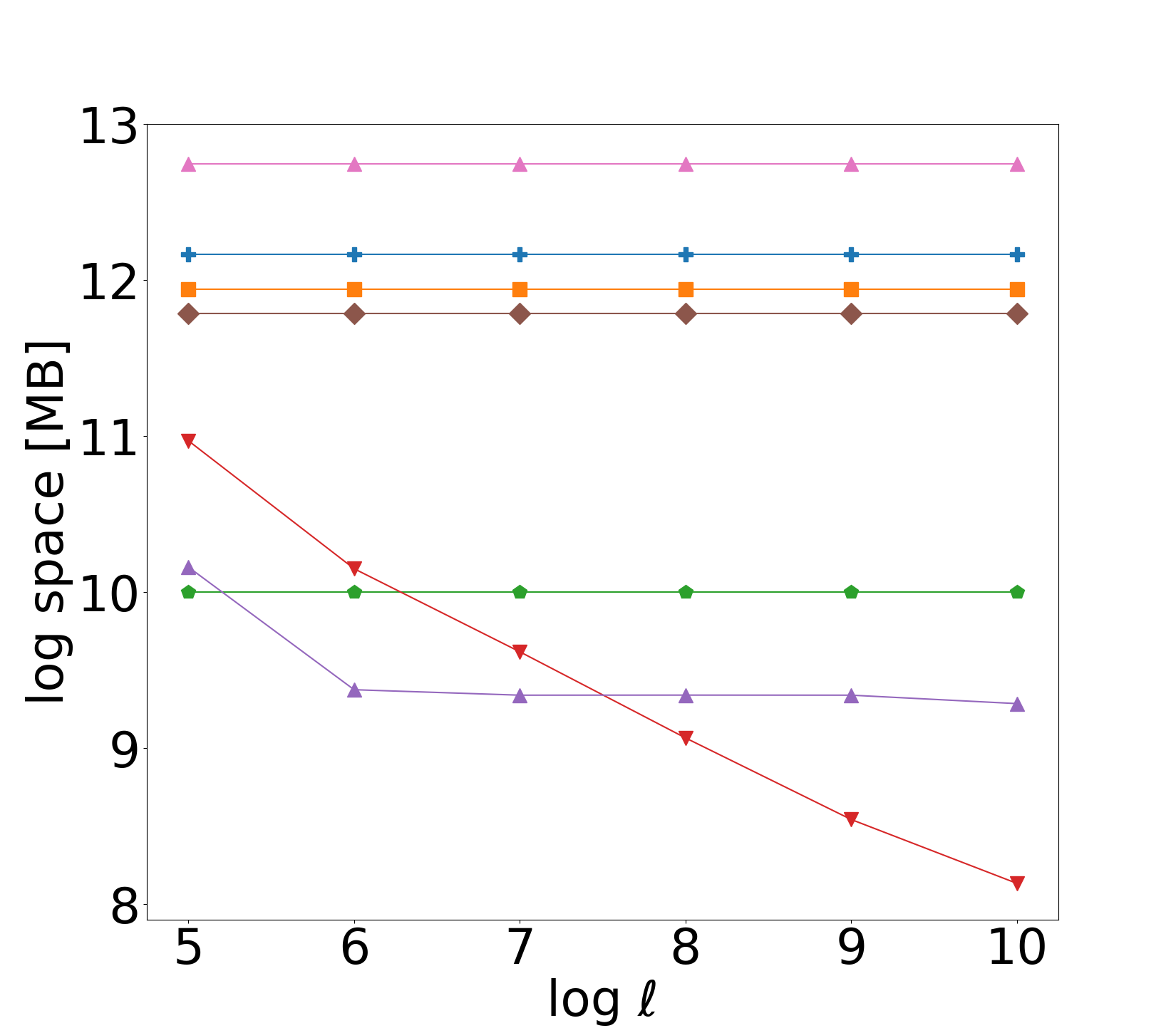}}
\subfloat[PROTEINS]
  {\includegraphics[width=.2\linewidth]
  {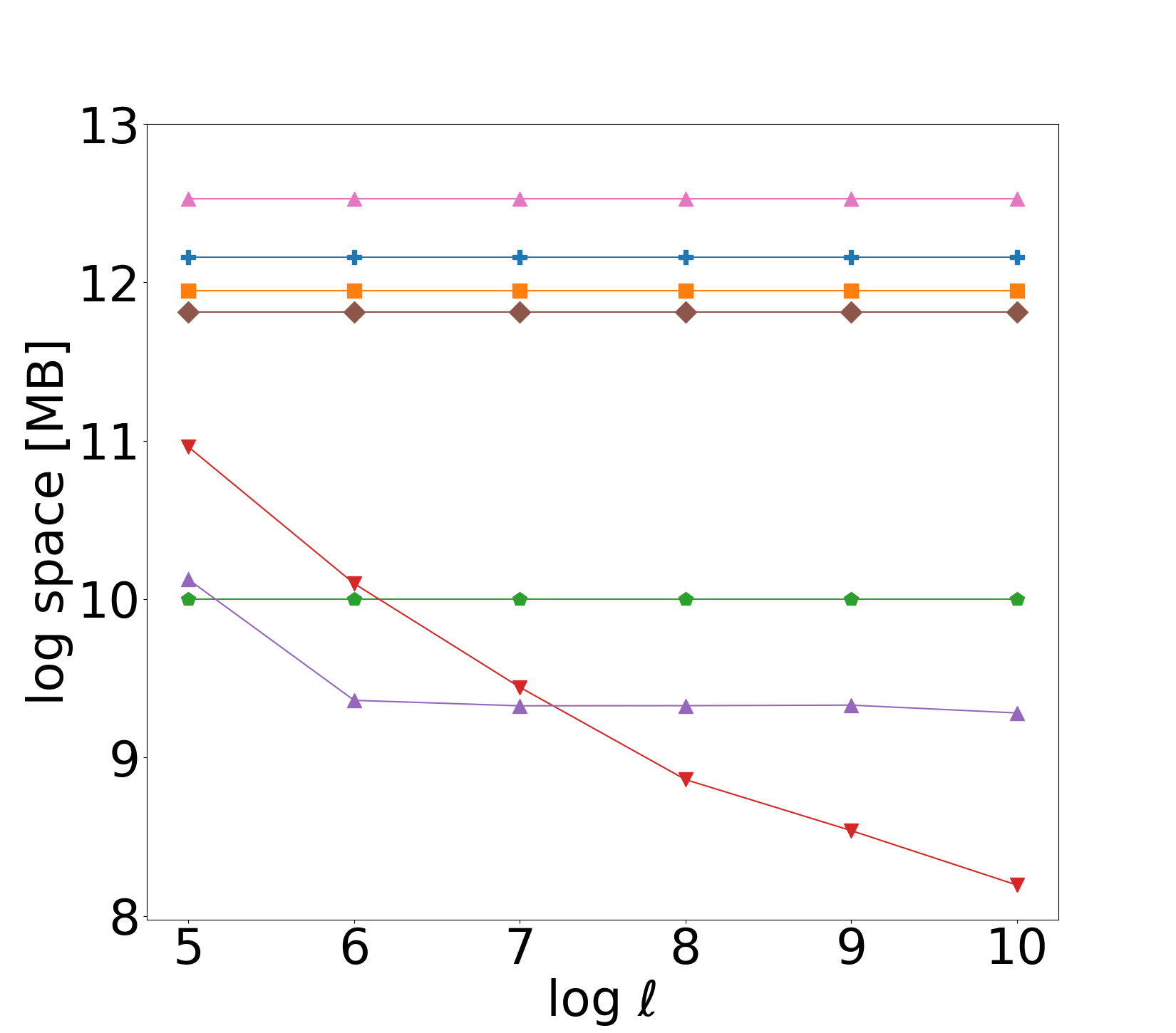}}
\subfloat[XML]
  {\includegraphics[width=.2\linewidth]
  {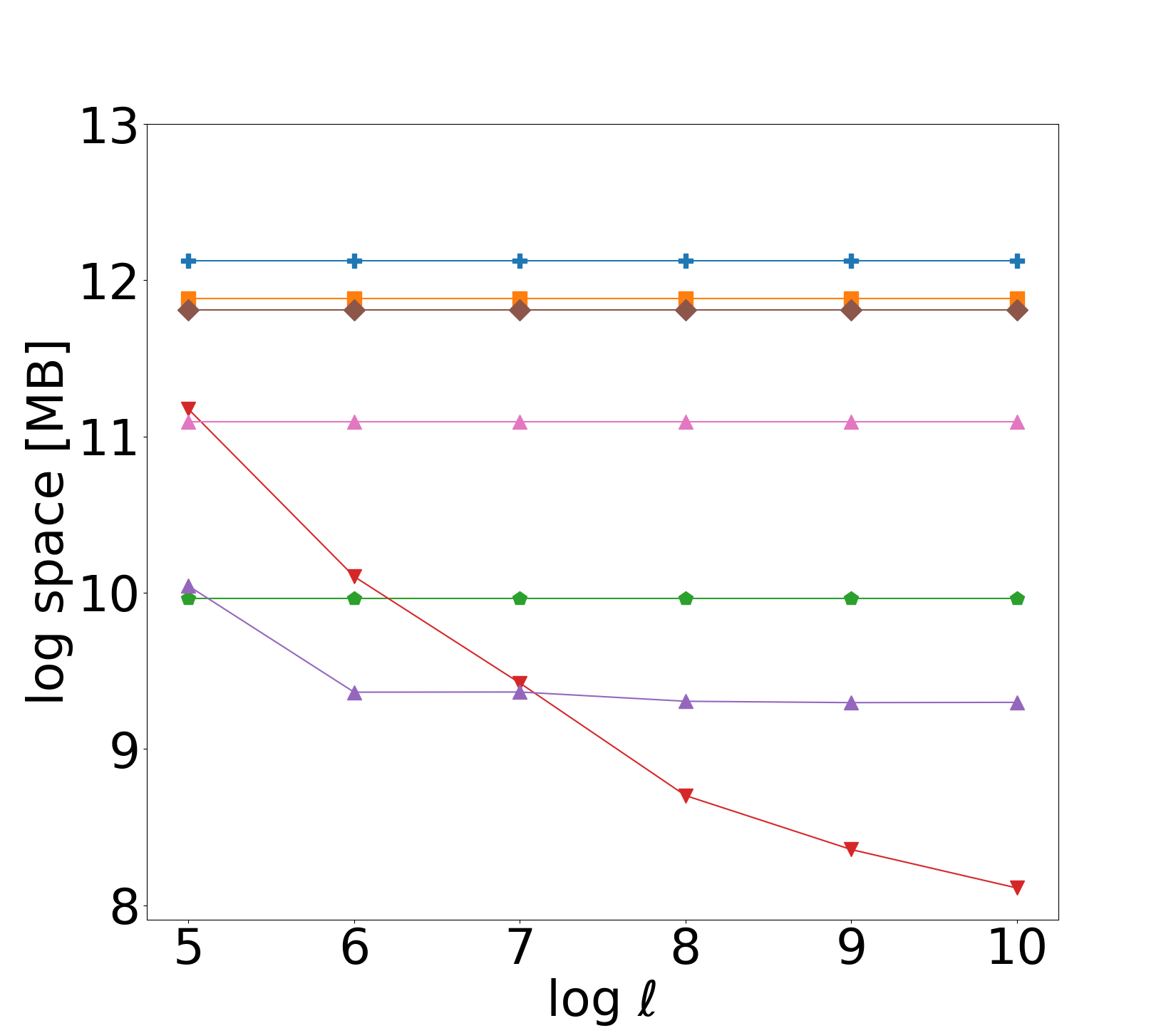}}  
\subfloat[SOURCES]
  {\includegraphics[width=.2\linewidth]
  {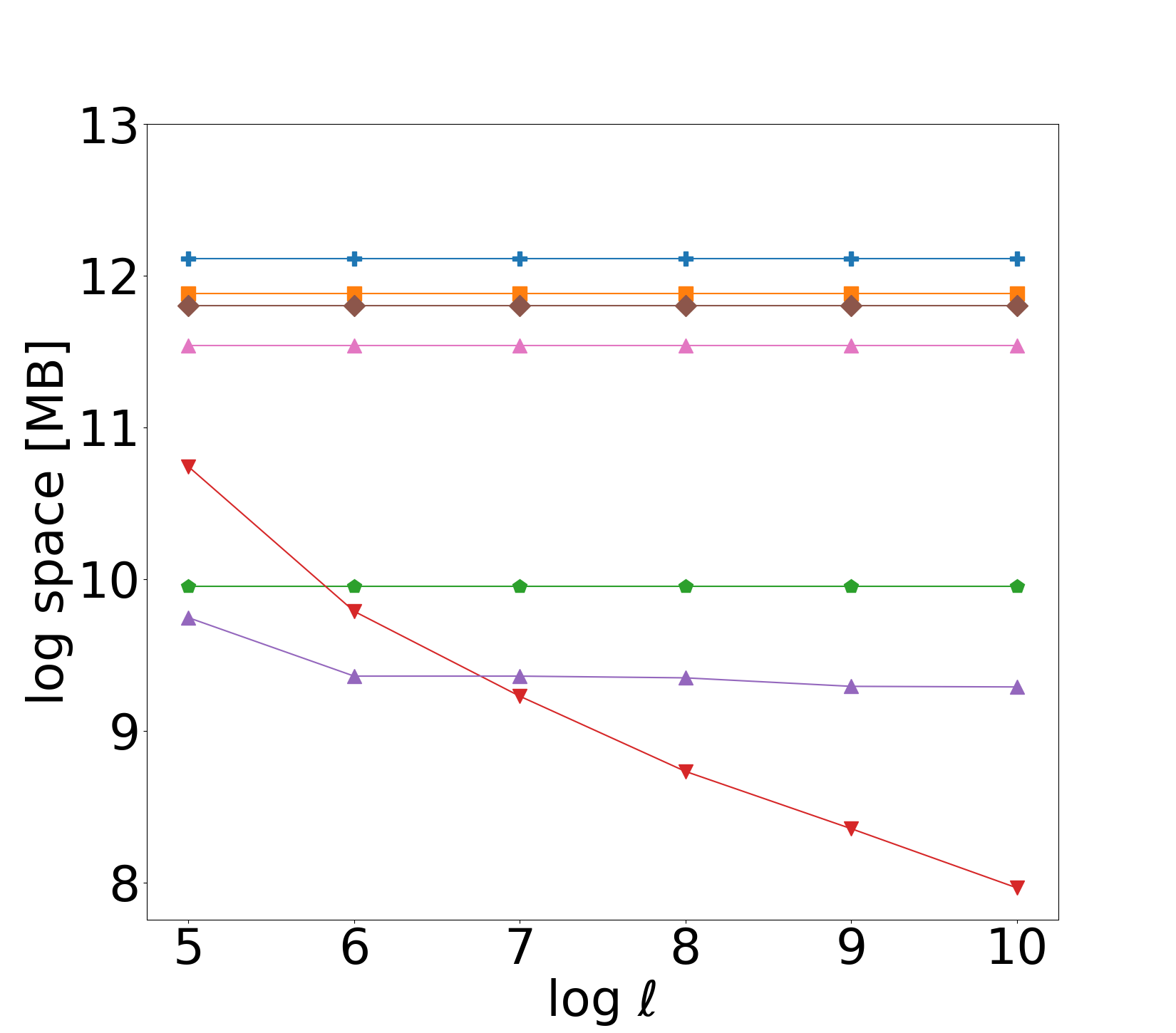}}
\subfloat[ENGLISH]
  {\includegraphics[width=.2\linewidth]
  {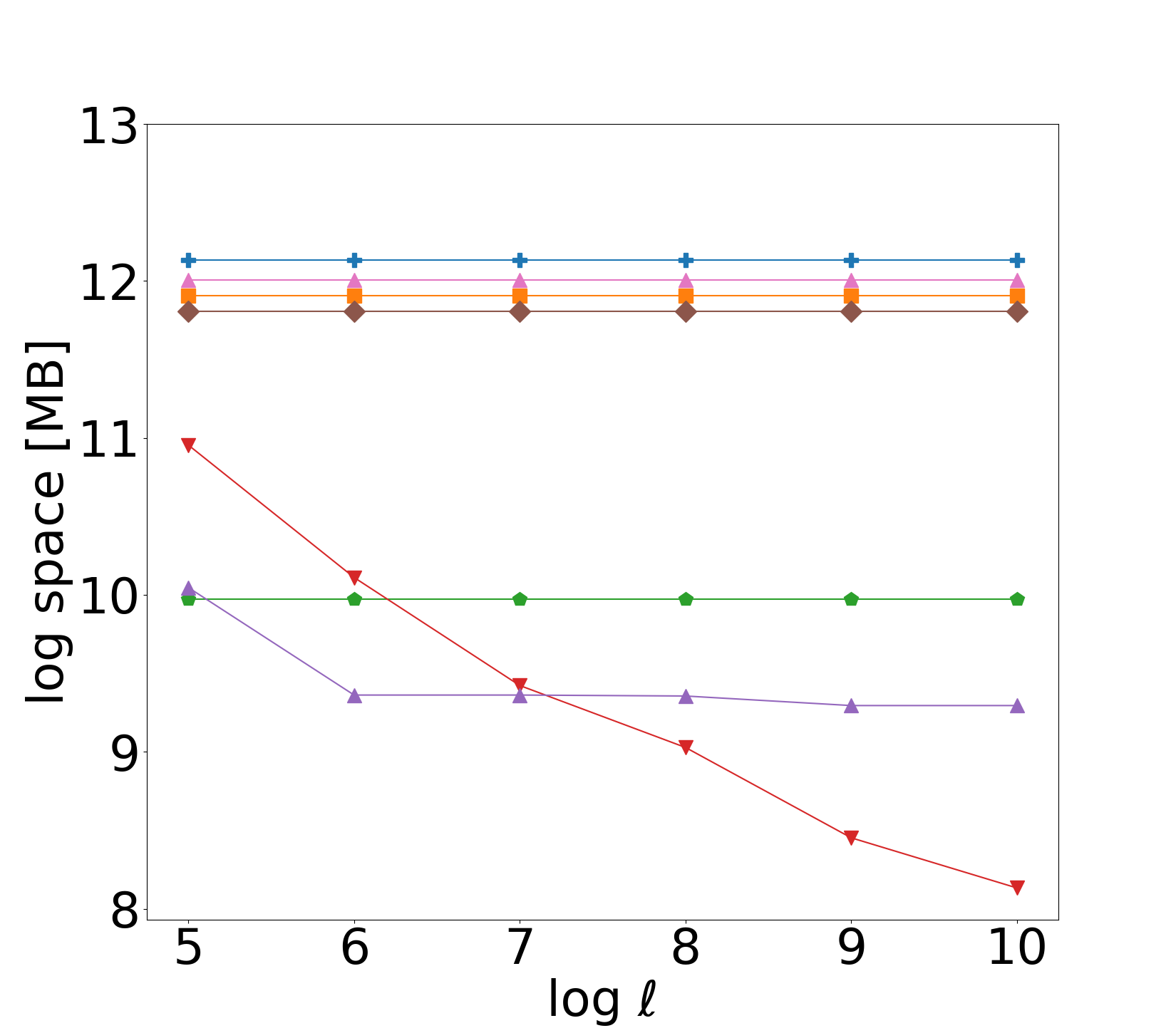}}
 \caption{\label{fig:index-construction-space}Space required to construct different indexes (MB in log-scale) for varying $\ell$ (log-scale).}
\end{figure*}

\subsection{Index construction time}\label{sec:con_time_results}
Figure~\ref{fig:index-construction-time} shows the time needed to construct all indexes for the first five datasets of Table~\ref{tab:datasets}. 
We set $b=25$K for \textsf{rrBDA-index II (int)} and \textsf{rrBDA-index II (ext)}. For all datasets and $\ell$ values, \textsf{rrBDA-index II (ext)} outperforms \textsf{CSA} and is  outperformed by the \textsf{SA}, \textsf{FM-index}, \textsf{CST}, and \textsf{r-index}. 
It also outperforms \textsf{rrBDA-index II (int)} for $\ell=32$ but performs worse  for all larger $\ell$ values. This is because, as $\ell$ increases, the number of bd-anchors decreases significantly, and thus \textsf{rrBDA-index II (int)} becomes faster (e.g., it performs similarly to one or more of the competitors for $\ell=1024$). On the other hand, as $\ell$ increases,  \textsf{rrBDA-index II (ext)} still needs to construct arrays \textsf{SA} and \textsf{LCP} in external memory, which is the bottleneck for the construction time.  

\begin{figure*}[ht]
\begin{subfigure}[b]{0.49\textwidth}
\hspace{0.5cm}
  \includegraphics[width=16cm]{legend_time_space.png}
 \end{subfigure}
 
\subfloat[DNA]
  {\includegraphics[width=.2\linewidth]
  {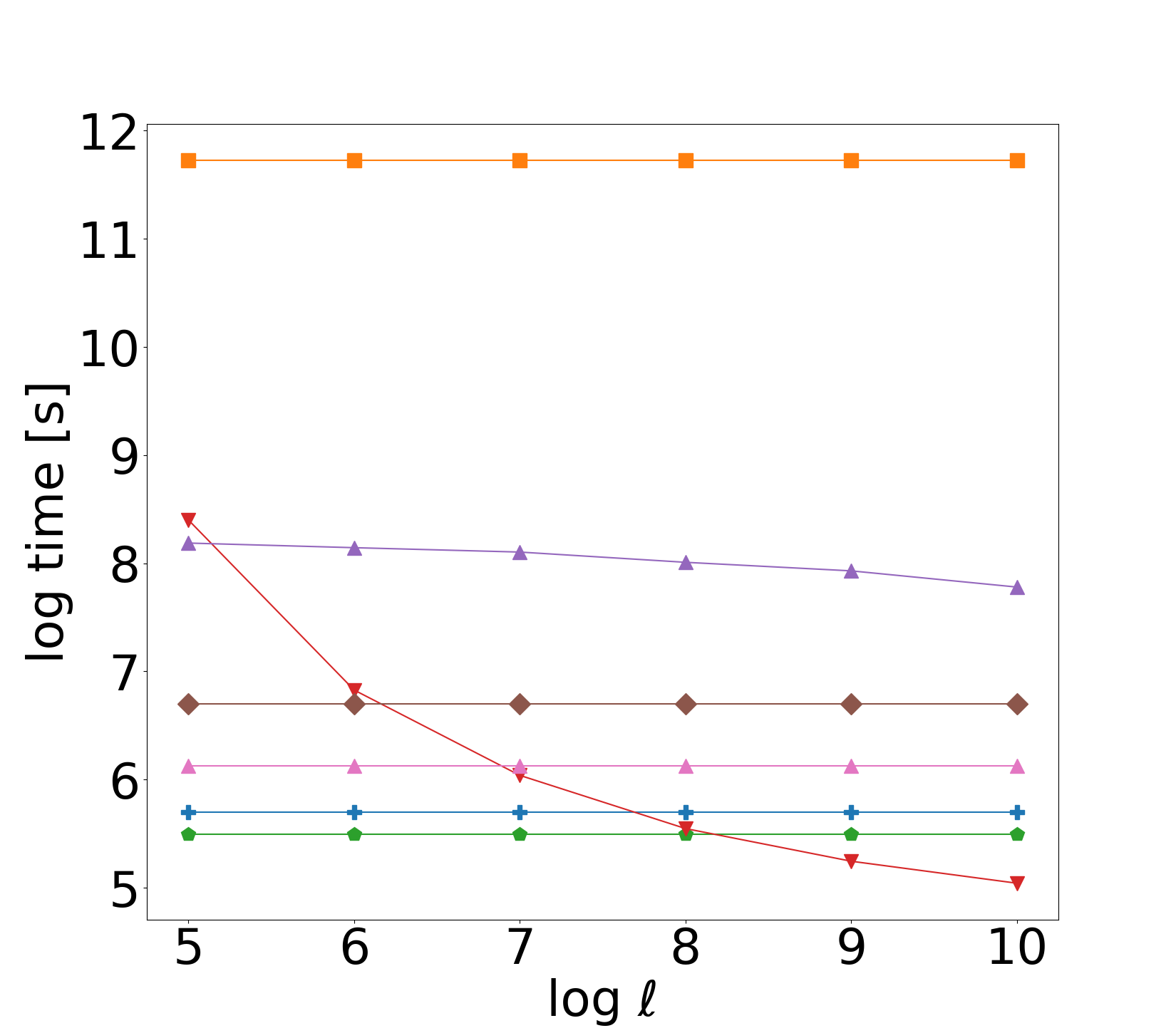}}
\subfloat[PROTEINS]
  {\includegraphics[width=.2\linewidth]
  {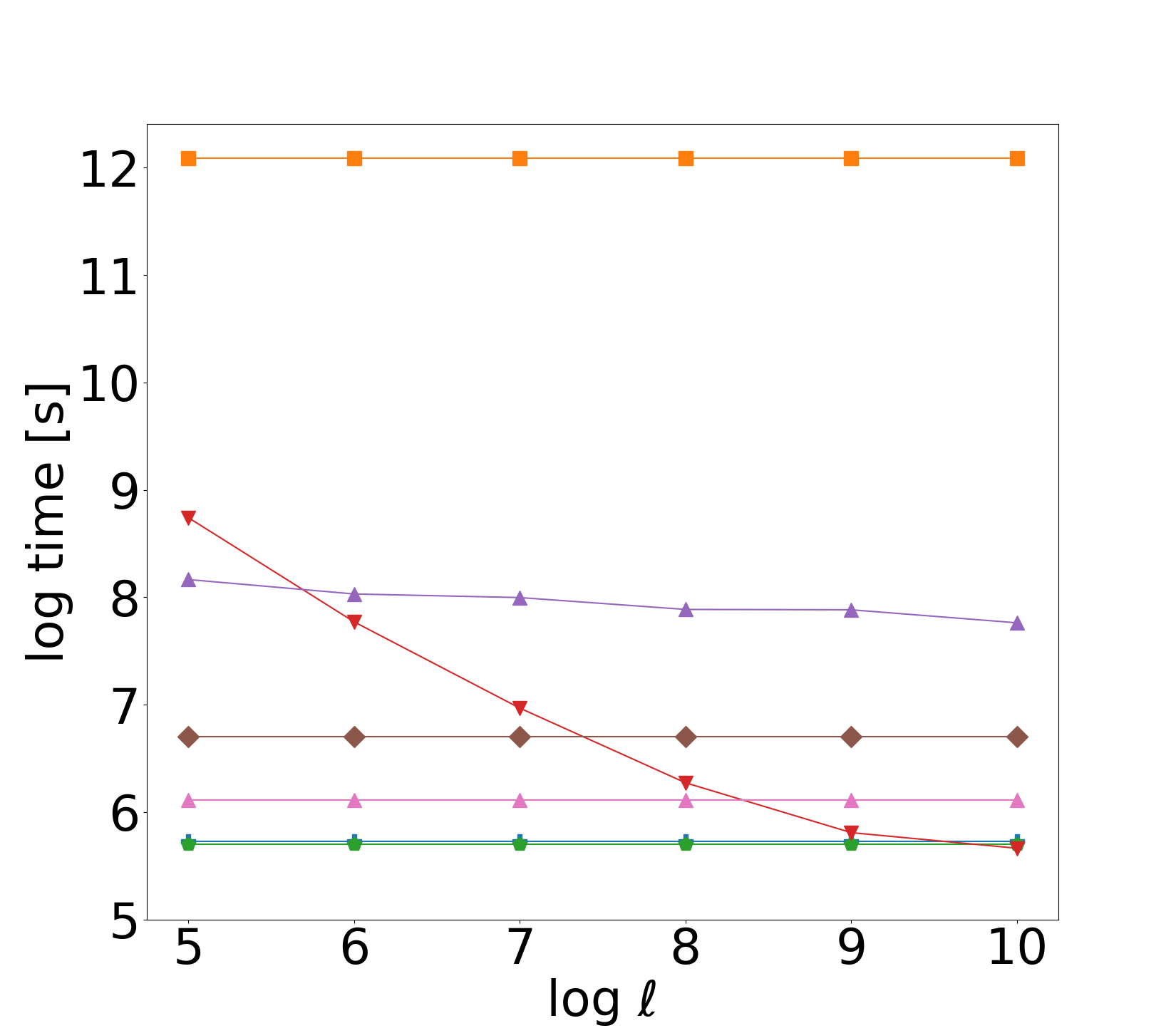}}
\subfloat[XML]
  {\includegraphics[width=.2\linewidth]
  {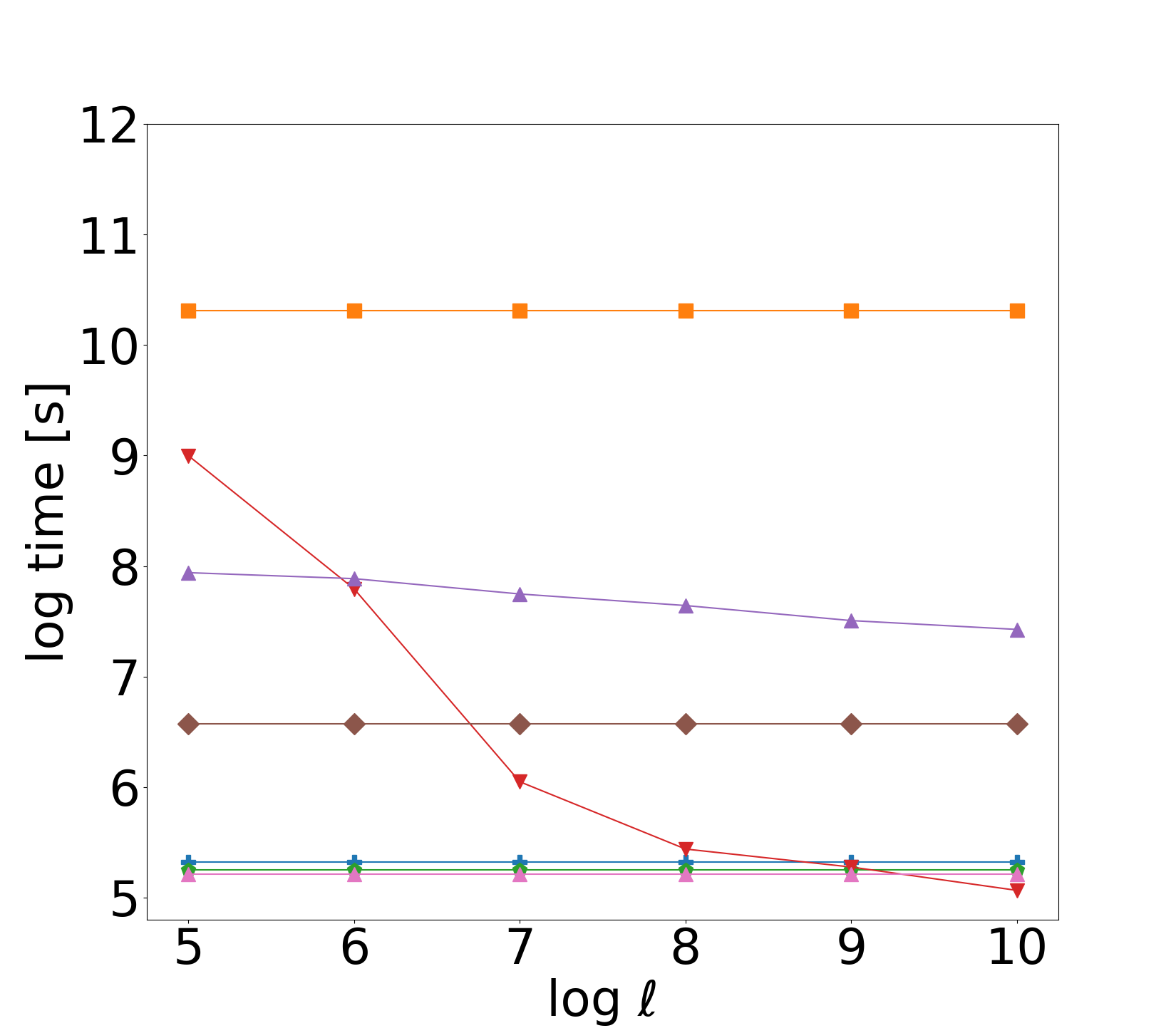}}  
\subfloat[SOURCES]
  {\includegraphics[width=.2\linewidth]
  {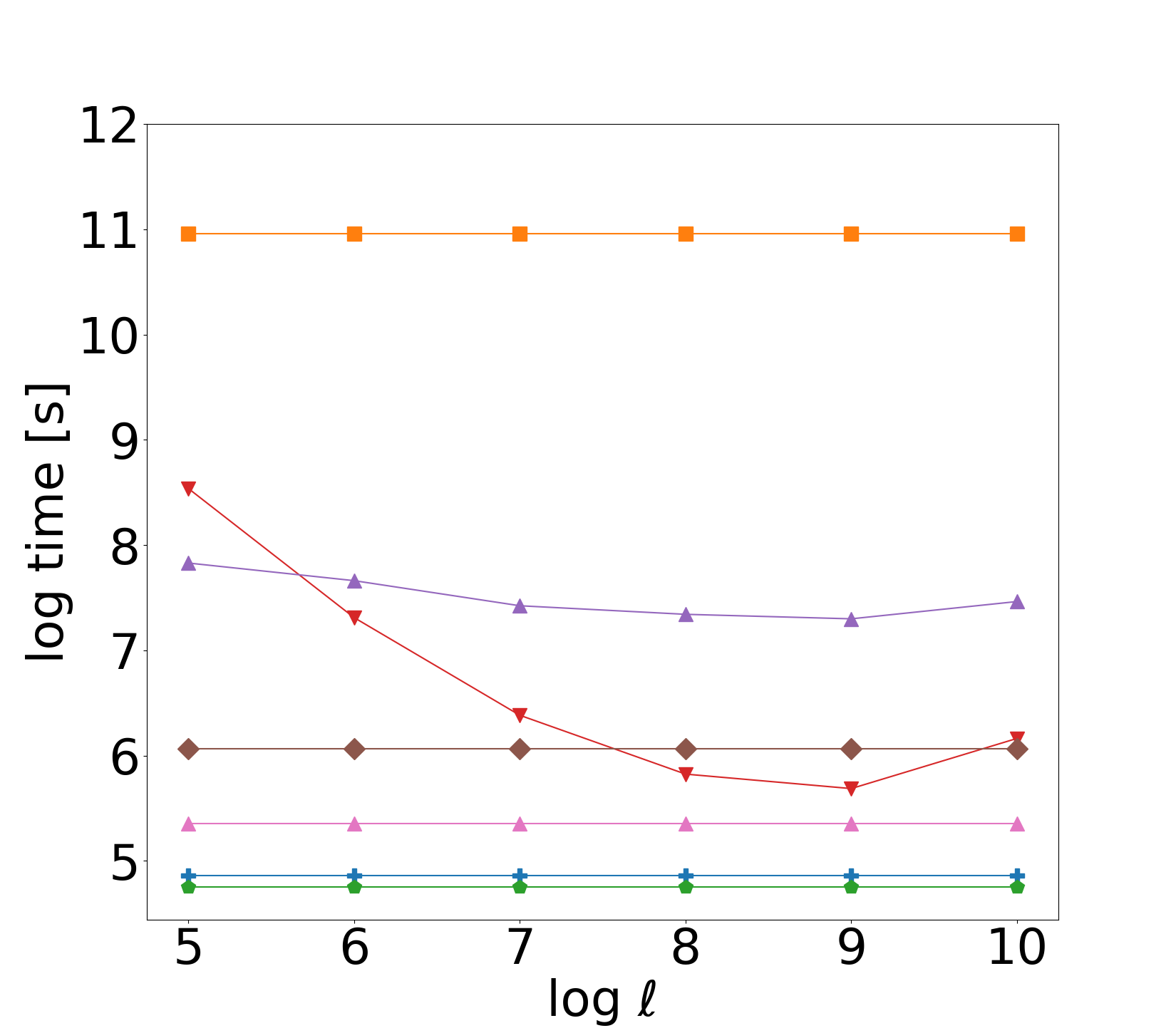}}
\subfloat[ENGLISH]
  {\includegraphics[width=.2\linewidth]
  {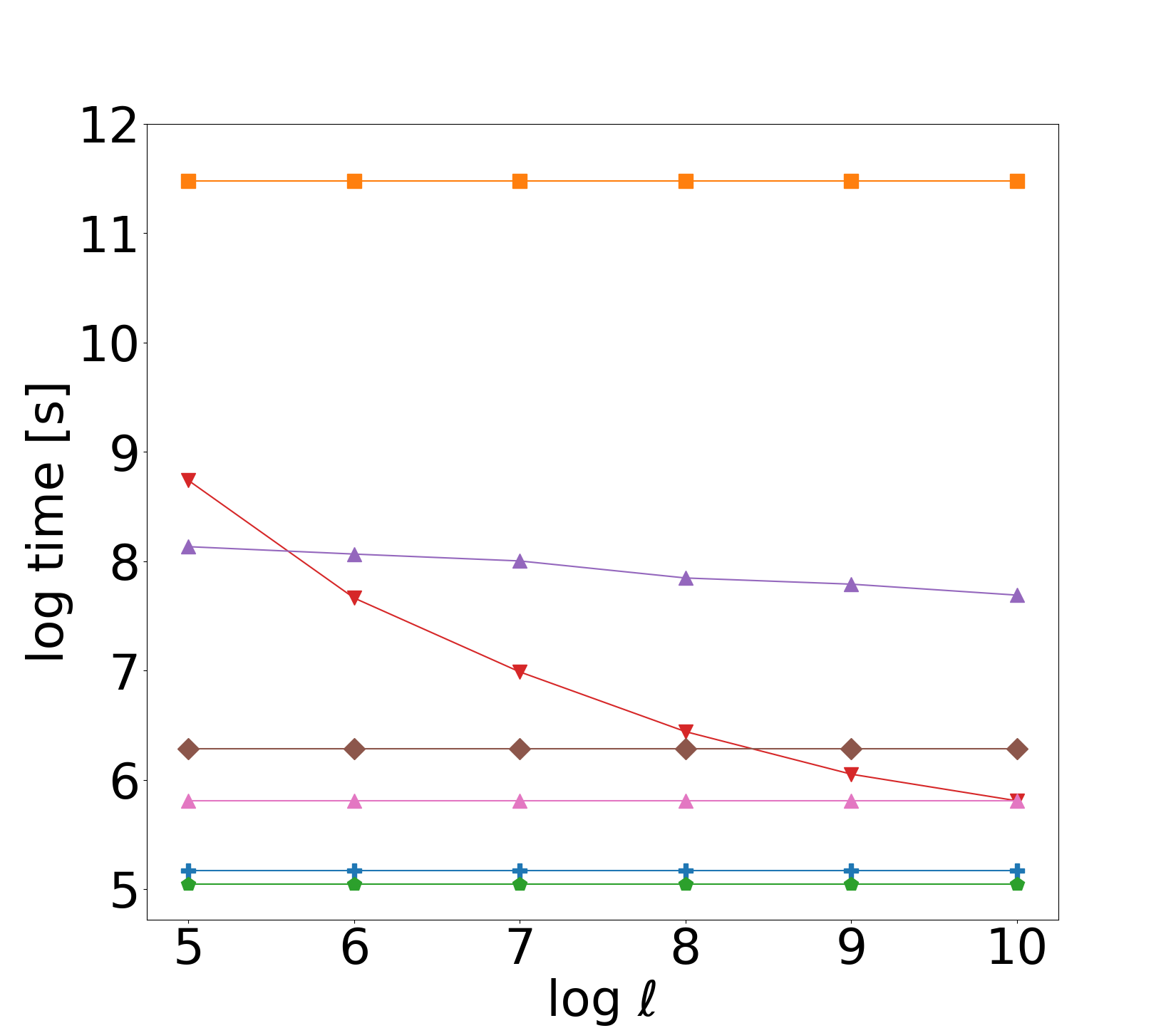}}
\caption{\label{fig:index-construction-time}Elapsed time to construct different indexes (seconds in log-scale) for varying $\ell$ (log-scale).}
\end{figure*}

\subsection{Results on the HUMAN dataset}\label{sec:HS_results}

Figure~\ref{fig:human_anchors} shows that the number of randomized reduced bd-anchors (\textsf{rrBDA-compute}) is much smaller than that of randomized bd-anchors (\textsf{rBDA-compute}) ($25.2\%$ on average) and that both of these numbers decrease as $\ell$ increases.  
This result is analogous to that of Figure~\ref{fig:rbda-construction-count}. 

Figure~\ref{fig:human_size} shows that 
\textsf{rrBDA-index II} (both \textsf{int} and \textsf{ext} versions have the same index size) outperforms both \textsf{FM-index} and \textsf{r-index} for $\ell\geq 64$ in terms of index size (we excluded the other indexes as they were outperformed by the latter two). 
For example, for $\ell=2^{14}=16,384$, \textsf{rrBDA-index II} takes less than 16MB, while \textsf{FM-index} and \textsf{r-index} take 1GB and 16GB, respectively. 

\begin{figure*}[ht]
 \subfloat[]
 {\includegraphics[width=0.2\linewidth]{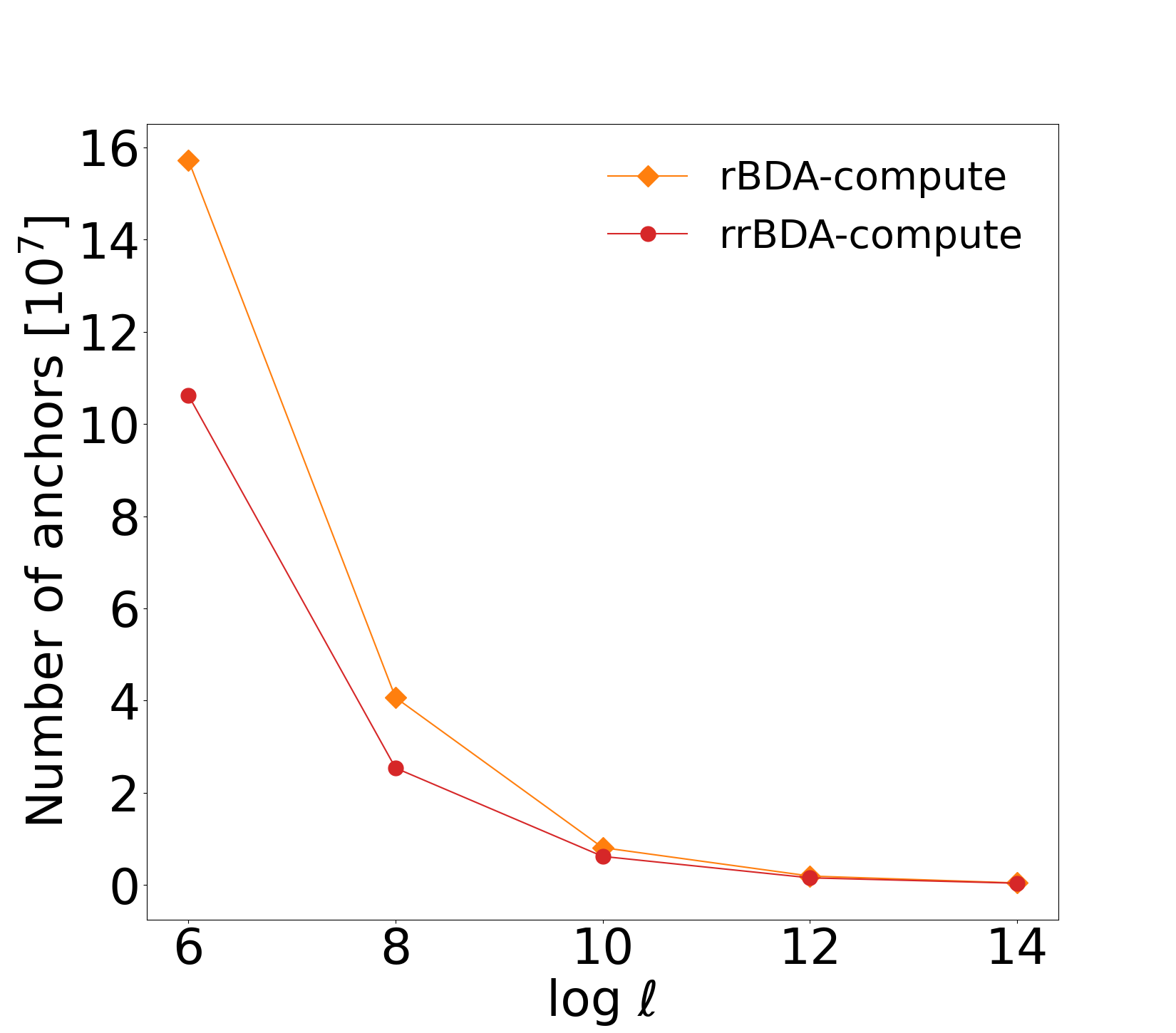}\label{fig:human_anchors}}
  \subfloat[]
{\includegraphics[width=0.2\linewidth]{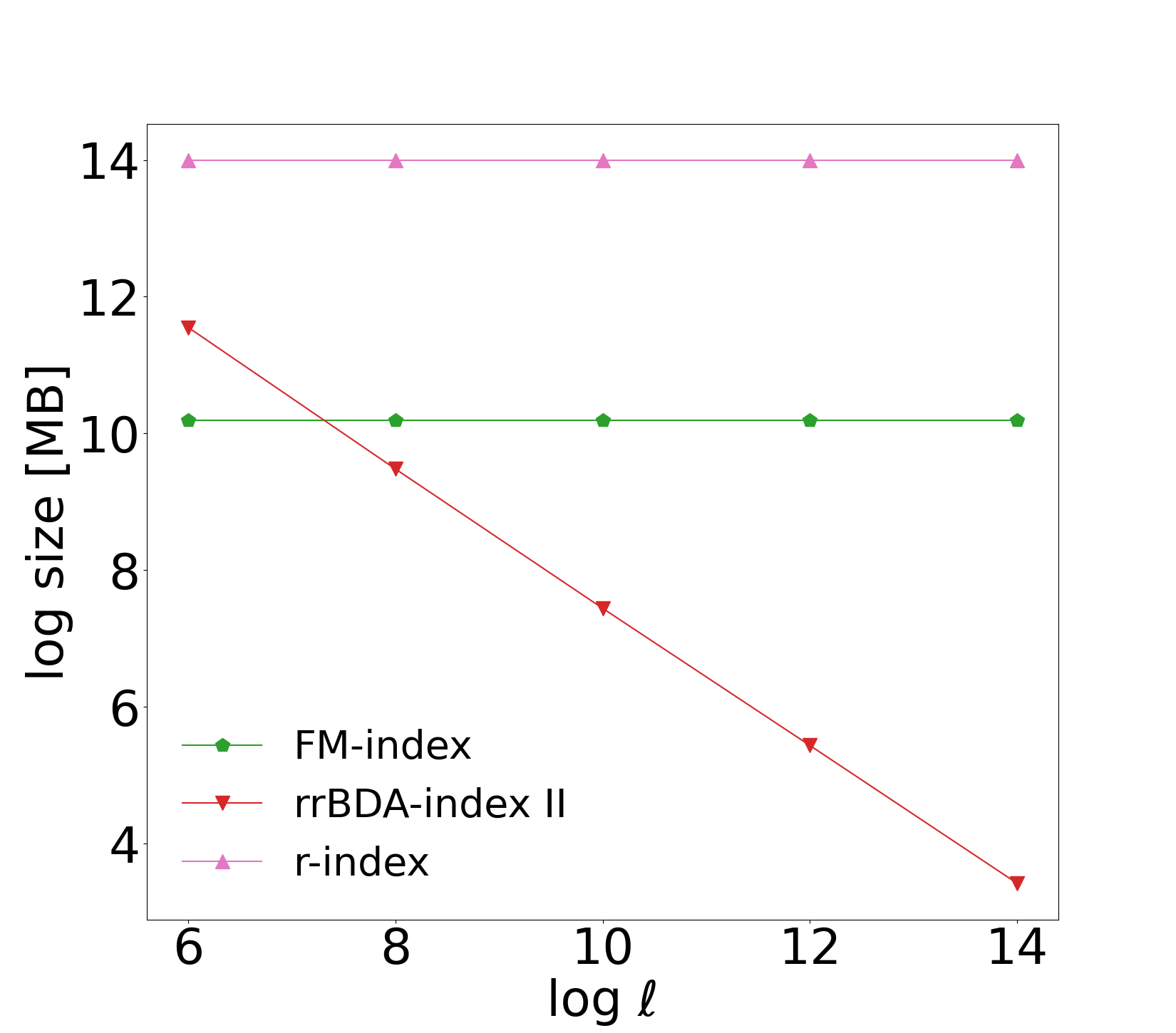}\label{fig:human_size}}
\subfloat[]
{\includegraphics[width=0.2\linewidth]{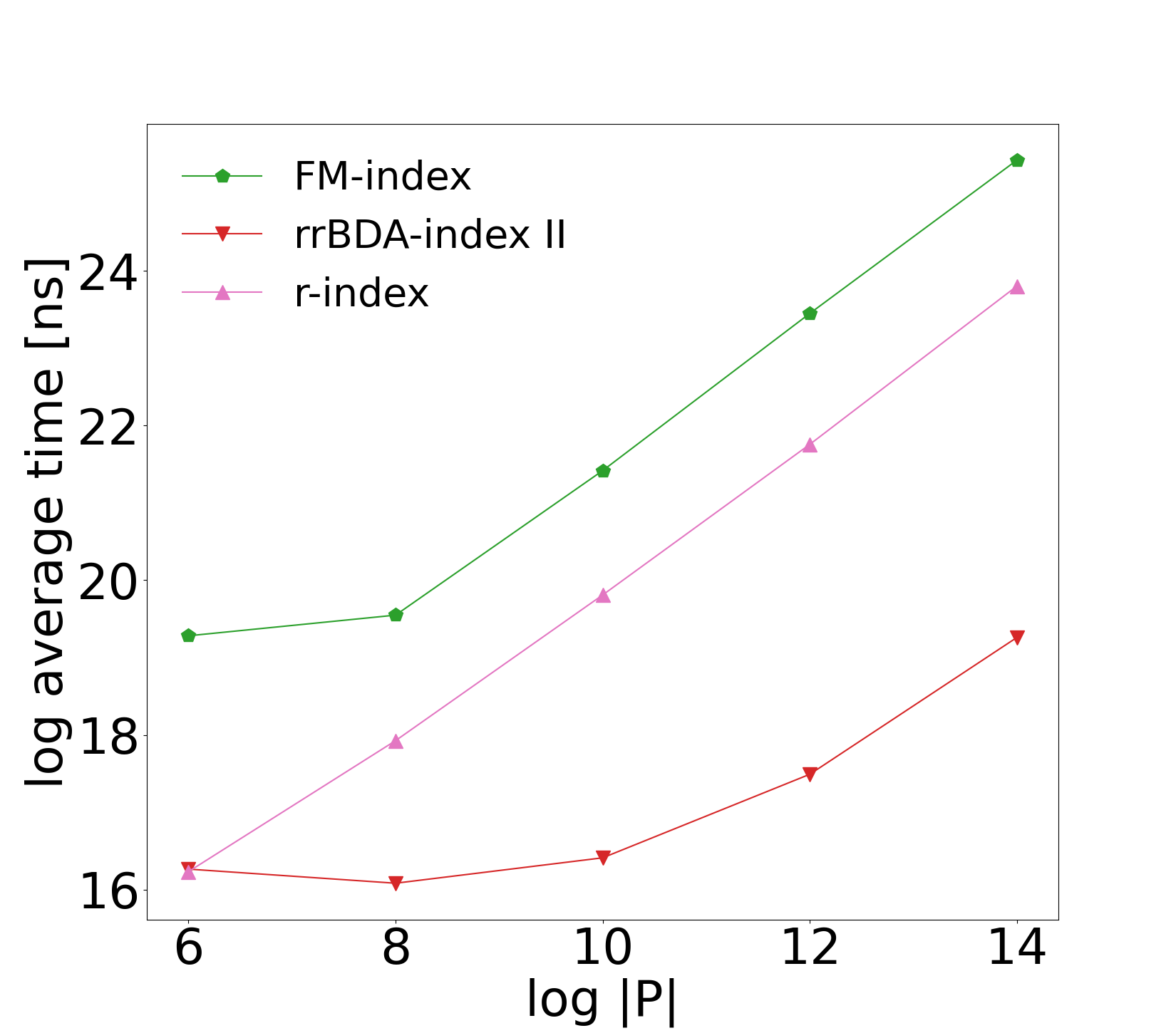}\label{fig:human_pattern}}
\subfloat[]
{\includegraphics[width=0.2\linewidth]{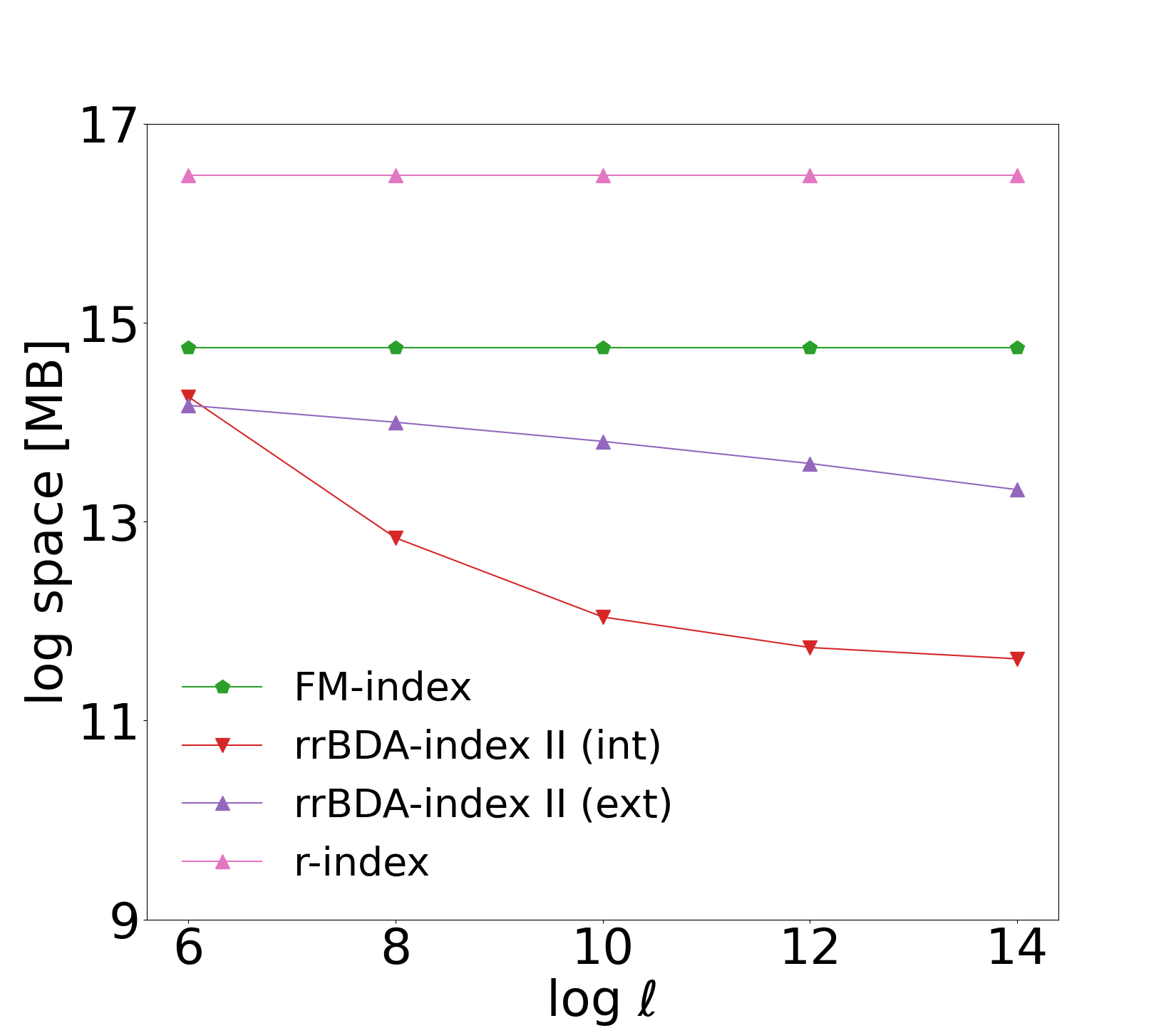}\label{fig:human_space}}
\subfloat[]
{\includegraphics[width=0.2\linewidth]{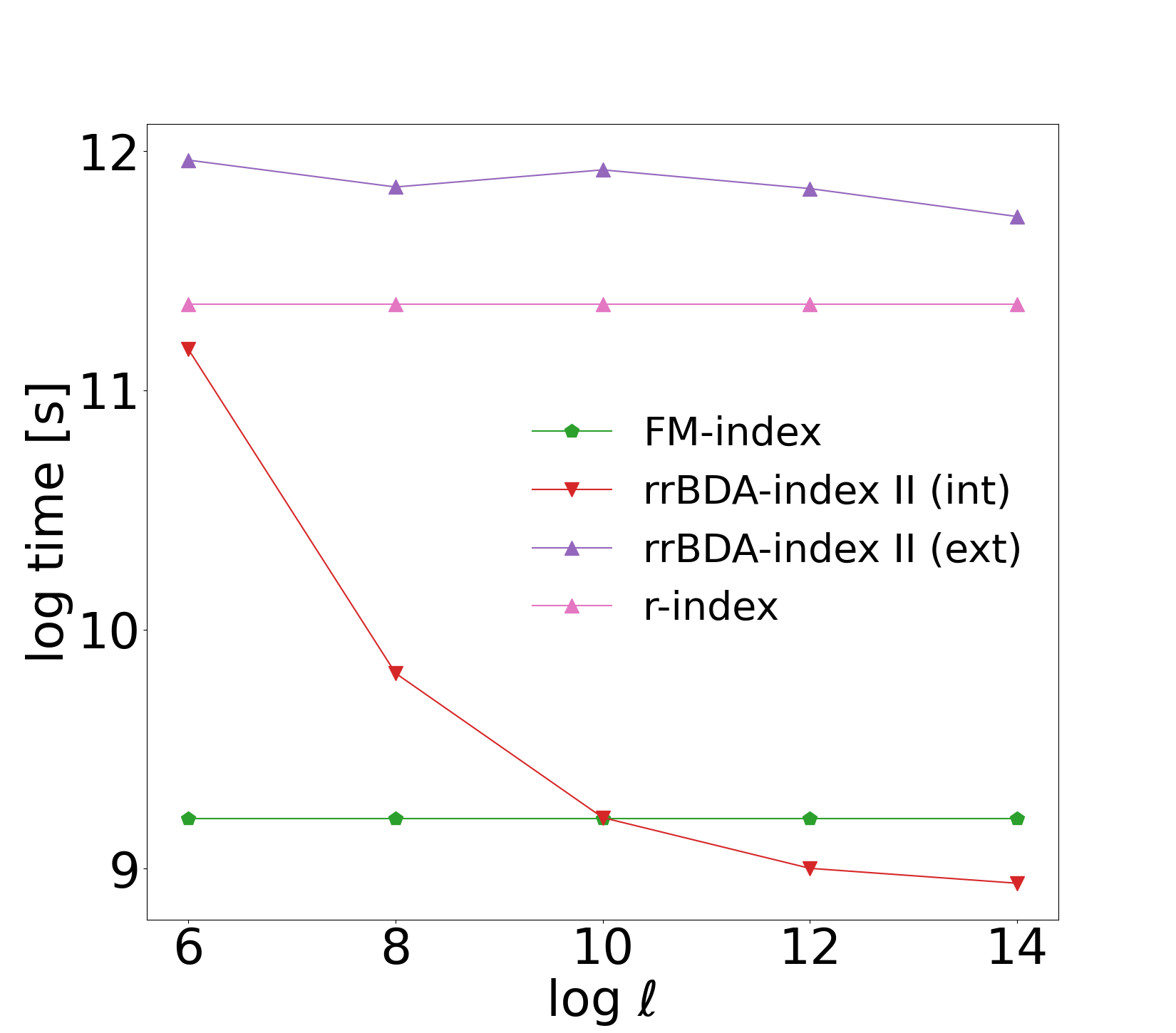}\label{fig:human_time}}
\caption{\label{fig:human}(a) Number of bd-anchors for HUMAN and varying $\ell$ (log-scale). (b) Size of different indexes (MB in log-scale) for HUMAN and varying $\ell$ (log-scale). (c) Average time for pattern matching (nanoseconds in log-scale) on HUMAN and for varying $|P|$ (log-scale). (d) Space required to construct different indexes (MB in log-scale) for HUMAN and varying $\ell$ (log-scale). (e) Elapsed time to construct different indexes (seconds in log-scale) for HUMAN and varying $\ell$ (log-scale).}
\end{figure*}

Figure~\ref{fig:human_pattern} shows the average query time for HUMAN  with $\ell=|P|$. As $\ell$ increases, \textsf{rrBDA-index II} becomes more than one order of magnitude faster than the \textsf{FM-index} and \textsf{r-index}. Note that for all $\ell$, \textsf{rrBDA-index II} is faster than the \textsf{FM-index} and \textsf{r-index}.

Figure~\ref{fig:human_space} shows that
\textsf{rrBDA-index II (int)} outperforms all other indexes in terms of construction space, and the difference between \textsf{rrBDA-index II (int)} and these indexes becomes larger as $\ell$ increases. The results are analogous to those in Figure~\ref{fig:index-construction-space}. 
For example, for $\ell=2^{14}$ \textsf{rrBDA-index II (int)} requires at least $2\times$ less space than \textsf{rrBDA-index II (ext)} and more than $4\times$ less space than \textsf{FM-index}. At the same time, \textsf{rrBDA-index II (ext)} substantially outperforms \textsf{FM-index} and \textsf{r-index} for all tested $\ell$ values.

Figure~\ref{fig:human_time} shows that 
\textsf{rrBDA-index II (int)} performs better as $\ell$ increases, and even outperforms the best competitor, \textsf{FM-index} for $\ell>2^{10}$.\footnote{We set $b = 130K$, as this dataset is much larger than those above~\cite{DBLP:journals/pvldb/AyadLP23}.} At the same time, \textsf{rrBDA-index II (ext)} performs worse than \textsf{FM-index} and \textsf{r-index}. The reason that \textsf{rrBDA-index II (ext)} performs worse is the construction of \textsf{SA} and the \textsf{LCP} array, as explained before. 
The results are analogous to those for DNA in Figure~\ref{fig:index-construction-time}. 

The presented results highlight the scalability of our construction algorithms as well as the superiority of \textsf{rrBDA-index II} for long patterns. 

\section{Final remarks}\label{sec:finale}

We have shown that our implementation of the bd-anchors index should be the practitioners' choice for long patterns. When $\ell \geq 512$, our index outperforms all indexes for all datasets in terms of index size. When $\ell \geq 64$, our index outperforms all indexes for all datasets in terms of query time. When $\ell \geq 64$, our index outperforms all indexes for all datasets in terms of construction space. The construction time of our index is up to $5\times$ slower than the state of the art. However, we believe that this does not make our index less practical as in many cases of practical interest our index can be constructed only once and queried multiple times.

An obvious criticism against the bd-anchors index is that its theoretical guarantees rely on the average-case model and on the assumption that we have at hand a lower bound on $\ell$. For the former, we have shown, using benchmark datasets (and not using synthetic uniformly random ones), that our index generally outperforms the state of the art for long patterns, which is the main claim here. For the latter, in most real-world applications we can think of, we do have a lower bound on the length of patterns (see Section~\ref{sec:intro:contributions}). For example, in long-read alignment, this bound is in the order of several hundreds. Finally, the state-of-the-art compressed indexes also have a crucial assumption, which is implicit and thus oftentimes largely neglected. The assumption is that the size $\sigma$ of the alphabet is much smaller than the length $n$ of the text. In the worst case, i.e., when $\sigma=\Theta(n)$, the state-of-the-art compressed indexes have no advantage compared to the suffix array or the suffix tree: they occupy $n\log\sigma=\Theta(n\log n)$ bits. 

The theoretical highlight of this paper is Theorem~\ref{the:worst-case}. The underlying index has $\cO(n/\ell)$ size, it can be constructed in near-optimal time using $\cO(n/\ell)$ working space, and it supports near-optimal pattern matching queries. Since it works for any integer $\ell\in[1,n]$, it can be seen as a parameterization of the classic suffix tree index~\cite{DBLP:conf/focs/Farach97}. Our construction is currently randomized; it is an interesting open problem to make it deterministic.

\section*{Acknowledgments}
The authors thank Panagiotis Charalampopoulos for helpful suggestions on Theorem~\ref{the:worst-case}.
This work is partially supported by the PANGAIA and ALPACA projects that have received funding from the European Union's Horizon 2020 research and innovation programme under the Marie Skłodowska-Curie grant agreements No 872539 and 956229, respectively.

\bibliographystyle{plain}
\bibliography{references}

\end{document}